\documentclass[runningheads,10pt]{llncs}

\pdfoutput=1

\usepackage[english]{babel}
\usepackage{amsmath,amssymb,eurosym,ifthen,enumitem,mathpartir}
\usepackage[backref=page]{hyperref}
\usepackage{color}
\hypersetup{
  colorlinks=true,
  citecolor=blue,
  linkcolor=red,
  urlcolor=violet}
\usepackage[hang,centerlast]{subfigure}

\usepackage{tikz}
\usetikzlibrary{automata,patterns,topaths,shapes,calc,arrows}
\pgfrealjobname{mscs2012-arxiv}

\title{Quantifying Opacity\thanks{Part of this work has been published in the proceedings of \textsc{Qest}'10~\cite{berard10}.}}

\titlerunning{Quantifying Opacity}
\author{B. B\'erard\inst{1} \and J. Mullins\inst{2} and M. Sassolas\inst{3}}

\authorrunning{B. B\'erard, J. Mullins and M. Sassolas}

\tocauthor{B. B\'erard, J. Mullins and M. Sassolas}

\institute{ Universit\'e Pierre \& Marie Curie, LIP6/MoVe, CNRS UMR 7606, Paris, France 
\and  \'Ecole Polytechnique de Montr\'eal, Dept. of Comp. \& Soft. Eng., Montreal (Quebec), Canada 
\and D\'epartement d'informatique, Universit\'e Libre de Bruxelles, Bruxelles, Belgique
}

\date{December 21, 2010; Revised July 3, 2012}

\input{macros-arxiv}

\colorlet{auxdraw}{blue!25!black!70}
\colorlet{auxfill}{auxdraw!30}
\colorlet{auxtext}{black}
\tikzstyle{aux}=[draw=auxdraw,fill=auxfill,text=auxtext]
\tikzstyle{phee}=[draw=auxdraw,fill=auxfill]
\tikzstyle{every picture}+=[initial text=,>=stealth']
\tikzstyle{every state}=[aux]

\begin{document}

\maketitle

\begin{abstract}
  Opacity is a general language-theoretic framework in which several
  security properties of a system can be expressed. Its parameters are
  a predicate, given as a subset of runs of the system, and an
  observation function, from the set of runs into a set of
  observables. The predicate describes secret information in the
  system and, in the possibilistic setting, it is opaque if its
  membership cannot be inferred from observation.

  In this paper, we propose several notions of quantitative opacity
  for probabilistic systems, where the predicate and the observation
  function are seen as random variables. Our aim is to measure (i) the
  probability of opacity leakage relative to these random variables
  and (ii) the level of uncertainty about membership of the predicate
  inferred from observation. We show how these measures extend
  possibilistic opacity, we give algorithms to compute them for
  regular secrets and observations, and we apply these computations on
  several classical examples. We finally partially investigate the
  non-deterministic setting.
\end{abstract}

\section{Introduction}

\subsubsection*{Motivations.}
Opacity~\cite{mazare05} is a very general framework where a wide range
of security properties can be specified, for a system interacting with
a passive attacker. This includes for instance anonymity or
non-interference~\cite{goguen82}, the basic version of which states
that high level actions cannot be detected by low level
observations. Non-interference alone cannot capture every type of
information flow properties. Indeed, it expresses the complete absence
of information flow yet many information flow properties, like
anonymity, permits some information flow while peculiar piece of
information is required to be kept secret.  The notion of opacity was
introduced with the aim to provide a uniform description for security
properties e.g.  non-interference, noninference, various notions of
anonymity, key compromise and refresh, downgrading,
etc. \cite{bryans08}. Ensuring opacity by control was further studied
in~\cite{dubreil10}.

The general idea behind opacity is that a passive attacker should not
have worthwhile information, even though it can observe the system
from the outside. The approach, as many existing information
flow-theoretic approaches, is possibilistic.  We mean by this that non
determinism is used as a feature to model the random mechanism
generation for all possible system behaviors.  As such, opacity is not
accurate enough to take into account two orthogonal aspects of
security properties both regarding evaluation of the information
gained by a passive attacker.

The first aspect concerns the quantification of security properties.
If executions leaking information are negligible with respect to the
rest of executions, the overall security might not be compromised.
For example if an error may leak information, but appears only in
$1\%$ of cases, the program could still be considered safe. The
definitions of opacity~\cite{alur06,bryans08} capture the existence of at least one perfect leak, but do not
grasp such a measure.

The other aspect regards the category of security properties a system
has to assume when interacting with an attacker able to make
inferences from experiments on the
base of statistical analysis.  For example, if every time the system
goes \emph{bip}, there is $99\%$ chances that action $a$ has been
carried out by the server, then every \emph{bip} can be guessed to
have resulted from an $a$.  Since more and more security protocols
make use of randomization to reach some security
objectives~\cite{chaum88,reiter98}, it becomes important to extend
specification frameworks in order to cope with it.

\subsubsection*{Contributions.}
In this paper we investigate several ways of extending opacity to a
purely probabilistic framework.  Opacity can be defined either as the
capacity for an external observer to deduce that a predicate was true
(asymmetrical opacity) or whether a predicate is true or false
(symmetrical opacity).  Both notions can model relevant security
properties, hence deserve to be extended.  On the other hand, two
directions can be taken towards the quantification of opacity.  The
first one, which we call \emph{liberal}, evaluates the degree of
non-opacity of a system: how big is the security hole?  It aims at
assessing the probability for the system to yield \emph{perfect}
information.  The second direction, which is called
\emph{restrictive}, evaluates how opaque the system is: how robust is
the security?  The goal here is to measure how reliable is the
information gained through observation.  This yields up to four
notions of quantitative opacity, displayed in
\tablename~\ref{tab:allopacities}, which are formally defined in this
paper.  The choice made when defining these measures was that a value
$0$ should be meaningful for opacity in the possibilistic sense.  As a
result, liberal measures are $0$ when the system is opaque and
restrictive ones are $0$ when the system is not.
\begin{table}[b]
\centering
\begin{tabular}{c|c|c}
& Asymmetric & Symmetric \\\hline
Liberal {\scriptsize (Security hole)} & \lponame (\lpo) & \lpsoname (\lpso)\\\hline
Restrictive {\scriptsize (Robustness)} & \hponame (\hpo) & \vponame (\vpo)
\end{tabular}
\caption{The four probabilistic opacity measures.}
\label{tab:allopacities}
\end{table}

Moreover, like opacity itself, all these measures can be instantiated
into several probabilistic security properties such as probabilistic
non-interference and anonymity.  We also show how to compute these
values in some regular cases and apply the method to the dining
cryptographers problem and the crowd protocols, re-confirming in
passing the correctness result of Reiter and Rubin~\cite{reiter98}.

Although the measures are defined in systems without nondeterminism,
they can be extended to the case of systems scheduled by an adversary.
We show that non-memoryless schedulers are requested in order to reach
optimum opacity measures.

\subsubsection*{Related Work.}
Quantitative measures for security properties were first advocated in
\cite{millen87} and~\cite{wittbold90}. In~\cite{millen87}, Millen makes
an important step by relating the non-interference property with the
notion of mutual information from information theory in the context of
a system modeled by a deterministic state machine.  He proves that the
system satisfies the non-interference property if and only if the
mutual information between the high-level input random variable and
the output random variable is zero. He also proposes mutual
information as a measure for information flow by showing how
information flow can be seen as a noisy probabilistic channel, but he
does not show how to compute this measure. In~\cite{wittbold90}
Wittbold and Johnson introduce {\em nondeducibility on strategies} in
the context of a non-deterministic state machine. A system satisfies nondeducibility on strategies
if the observer cannot deduce information from the observation by any
collusion with a secret user and using any adaptive strategies. 
They
observe that if such a system is run multiple times with feedback
between runs, information can be leaked by coding schemes across
multiple runs. In this case, they show that a discrete memoryless
channel can be built by associating a distribution with the noise
process. 
From then on, numerous studies were
devoted to the computation of (covert) channel capacity in various
cases (see e.g.~\cite{mantel2009}) or more generally information leakage.

In~\cite{smith09}, several measures of information leakage extending
these seminal works for deterministic or probabilistic programs with
probabilistic input are discussed. These measures quantify the
information concerning the input gained by a passive attacker
observing the output. Exhibiting programs for which the value of
entropy is not meaningful, Smith proposes to consider instead the
notions of vulnerability and min-entropy to take in account the fact
that some execution could leak a sufficiently large amount of
information to allow the environment to guess the remaining secret. As
discussed in Section~\ref{sec:comparison}, probabilistic opacity takes
this in account.

In~\cite{chatzikokolakis08}, in order to quantify anonymity, the
authors propose to model the system (then called {\em Information
  Hiding System}) as a noisy channel in the sense of Information
Theory: The secret information is modeled by the inputs, the
observable information is modeled by the outputs and the two set are
related by a conditional probability matrix.  In this context,
probabilistic information leakage is very naturally specified in terms
of mutual information and capacity. A whole hierarchy of probabilistic notions of
anonymity have been defined.
The approach  was completed
in~\cite{andres10} where anonymity is computed using regular
expressions.  More recently, in~\cite{alvim10}, the authors consider
Interactive Information Hiding Systems that can be viewed as channels
with memory and feedback.

In~\cite{boreale11a}, the authors analyze the asymptotic behaviour of attacker's error probability and information leakage in  Information Hiding Systems in the context of an attacker having the capabilities to make exactly one guess after observing $n$ independent  executions of the system while the secret information remains invariant.  Two cases are studied: the case in which each execution gives rise to a single observation and the case in which each state of an execution  gives rise to an observation in
the context of \emph{Hidden Markov Models}.  The relation of these sophisticated models of attacker with our attacker model is still to clarify.
 Similar models were also
studied in~\cite{mciver10}, where the authors define an ordering w.r.t.
probabilistic non-interference.

For systems modeled by process algebras, pioneering work was presented
in~\cite{lowe04,aldini04}, with channel capacity defined by counting
behaviors in discrete time (non-probabilistic) CSP~\cite{lowe04}, or
various probabilistic extensions of noninterference~\cite{aldini04} in
a generative-reactive process algebra.  Subsequent studies in this
area by~\cite{aldini09,boreale09,boreale10} also provide quantitative
measures of information leak, relating these measures with
noninterference and secrecy.  In~\cite{aldini09}, the authors
introduce various notions of noninterference in a Markovian process
calculus extended with prioritized/probabilistic zero duration actions
and untimed actions. In~\cite{boreale09} the author introduces two  notions  of information leakage in the (non-probabilistic)
$\pi$-calculus differing essentially in the assumptions made on the power of the attacker. The first one, called {\em absolute leakage}, corresponds to the average amount of information that was leaked to the attacker by the program in the context of an attacker with unlimited computational resources is defined in terms of conditional mutual information and follows the earlier results of Millen~\cite{millen87}. The second notion,  called {\em leakage rate}, corresponds to  the maximal number of bits of information that could be obtained per experiment  in the context in which the attacker can only perform a fixed number of tries, each yielding a binary outcome representing success or failure. Boreale also studies the relation between both notions of leakage and proves that they are consistent.
The author  also investigates compositionality
of leakage.  Boreale et al.~\cite{boreale10} propose a very general framework for reasoning about information leakage in  a sequential process calculus over a semiring   with some appealing applications to information leakage analysis when instantiating  and interpreting the semiring. It appears to be a promising scheme for specifying and analysing regular quantitative information flow like we do in Section~\ref{sec:computing}.

\medskip
Although the literature on quantifying information leakage or
channel capacity is dense, few works actually tried to extend general
opacity to a probabilistic setting. A notion of probabilistic opacity
is defined in~\cite{lakhnech05}, but restricted to properties whose
satisfaction depends only on the initial state of the run.  The
opacity there corresponds to the probability for an observer to guess
from the observation whether the predicate holds for the run.  In that
sense our restrictive opacity (Section~\ref{sec:restrictive}) is close
to that notion.  However, the definition of~\cite{lakhnech05} lacks
clear ties with the possibilistic notion of opacity. Probabilistic
opacity is somewhat related to the notion of {\em view} presented
in~\cite{boreale11b} as authors include, like we do, a predicate to
their probabilistic model and observation function but probabilistic
opacity can hardly be compared with view. Indeed, on one hand,
although their setting is different (they work on Information Hiding
Systems extended with a view), our predicates over runs could be
viewed as a generalization of the predicates over a finite set of
states (properties). On the other hand, a view in their setting is an
arbitrary partition of the state space, whereas we partition the runs
into only two equivalence classes (corresponding to \emph{true} and
\emph{false}).

\subsubsection*{Organization of the paper.}
In Section~\ref{sec:background}, we recall the definitions of opacity
and the probabilistic framework used throughout the paper.
Section~\ref{sec:liberal} and~\ref{sec:restrictive} present
respectively the liberal and the restrictive version of probabilistic
opacity, both for the asymmetrical and symmetrical case.  We present
in Section~\ref{sec:computing} how to compute these measures
automatically if the predicate and observations are regular.
Section~\ref{sec:comparison} compares the different measures and what
they allow to detect about the security of the system, through
abstract examples and a case study of the Crowds protocol.  In
Section~\ref{sec:schedulers}, we present the framework of
probabilistic systems dealing with nondeterminism, and open problems
that arise in this setting.

\section{Preliminaries}\label{sec:background}

In this section, we recall the notions of opacity, entropy,
and probabilistic automata.

\subsection{Possibilistic opacity}\label{subsec:bg-opacity}

The original definition of opacity was given in~\cite{bryans08} for
transition systems.

Recall that a transition system is a tuple $\A=\langle \Sigma, Q,
\Delta, I\rangle$ where $\Sigma$ is a set of actions, $Q$ is a set of
states, $\Delta \subseteq Q \times \Sigma \times Q$ is a set of
transitions and $I \subseteq Q$ is a subset of initial states.  A
\emph{run} in $\A$ is a finite sequence of transitions written as:
\(\rho = q_0 \xrightarrow{a_1} q_1 \xrightarrow{a_2} q_2 \, \cdots \,
\xrightarrow{a_n} q_n\). For such a run, $\fst(\rho)$
(resp. $\lst(\rho)$) denotes $q_0$ (resp. $q_n$). We will also write
$\rho \cdot \rho'$ for the run obtained by concatenating runs $\rho$
and $\rho'$ whenever $\lst(\rho) = \fst(\rho')$.  The set of runs
starting in state $q$ is denoted by $Run_q(\A)$ and $Run(\A)$ denotes
the set of runs starting from some initial state: $Run(\A) =
\bigcup_{q \in I}Run_q(\A)$.

Opacity qualifies a predicate $\varphi$, given as a subset of
$Run(\A)$ (or equivalently as its characteristic function ${\bf
  1}_\varphi$), with respect to an \emph{observation function} ${\cal
  O}$ from $Run(\A)$ onto a (possibly infinite) set $Obs$ of
\emph{observables}.  Two runs $\rho$ and $\rho'$ are equivalent
w.r.t. ${\cal O}$ if they produce the same observable: ${\cal O}(\rho)
= {\cal O}(\rho')$.  The set ${\cal O }^{-1} (o)$ is called an
\emph{observation class}.  We sometimes write $[\rho]_{\cal O}$ for
${\cal O }^{-1} (\mathcal{O}(\rho))$.

A predicate $\varphi$ is opaque on $\A$ for ${\cal O}$ if for
every run $\rho$ satisfying $\varphi$, there is a run $\rho'$ not
satisfying $\varphi$ equivalent to $\rho$.

\begin{definition}[Opacity]\label{def:opacity}
  Let $\A$ be a transition system and ${\cal O}: Run(\A) \rightarrow
  Obs$ a surjective function called observation. A predicate $\varphi
  \subseteq Run(\A)$ is \emph{opaque} on $\A$ for ${\cal O}$ if, for
  any $o \in Obs $, the following holds:
\[{\cal O}^{-1}(o) \not \subseteq \varphi.\]
\end{definition}

However, detecting whether an event \emph{did not} occur can give as
much information as the detection that the same event \emph{did}
occur.  In addition, as argued in~\cite{alur06}, the asymmetry of this
definition makes it impossible
to use with refinement: opacity would not be ensured in a system
derived from a secure one in a refinement-driven engineering process.
More precisely, if $\A'$ refines $\A$ and a property $\varphi$ is
opaque on $\A$ (w.r.t $\mathcal{O}$), $\varphi$ is not guaranteed to
be opaque on $\A'$ (w.r.t $\mathcal{O}$).

Hence we use the symmetric notion of opacity,
where a predicate is symmetrically opaque if it is opaque as well as
its negation.  More precisely:

\begin{definition}[Symmetrical opacity]\label{def:symopacity}
A predicate $\varphi \subseteq Run(\A)$ is \emph{symmetrically opaque} on system $\A$ for observation function ${\cal O}$ if, for any $o \in Obs $, the following holds:
\[{\cal O}^{-1}(o) \not \subseteq \varphi \mbox{ and } {\cal O}^{-1}(o) \not \subseteq \overline\varphi.\]
\end{definition}

The symmetrical opacity is a stronger security requirement.  Security
goals can be expressed as either symmetrical or asymmetrical opacity,
depending on the property at stake.

For example non-interference and anonymity can be expressed by opacity
properties.  Non-interference states that an observer cannot know
whether an action $h$ of high-level accreditation occurred only by
looking at the actions with low-level of accreditation in the set $L$.
So non-interference is equivalent to the opacity of predicate
$\varphi_{NI}$, which is true when $h$ occurred in the run, with
respect to the observation function $\mathcal{O}_L$ that projects the
trace of a run onto the letters of $L$; see
Section~\ref{ex:niwithopacity} for a full example.  We refer
to~\cite{bryans08} and~\cite{lin11} for other examples of properties
using opacity.

When the predicate breaks the symmetry of a model, the asymmetric
definition is usually more suited.  Symmetrical opacity is however
used when knowing $\varphi$ or $\overline\varphi$ is equivalent from a
security point of view. For
example, a noisy channel with binary input can be seen as a system
$\A$ on which the input is the truth value of $\varphi$ and the output
is the observation $o \in Obs$.  If $\varphi$ is symmetrically opaque
on $\A$ with respect to $\mathcal{O}$, then this channel is not
perfect: there would always be a possibility of erroneous
transmission.  The ties between channels and probabilistic transition
systems are studied in~\cite{andres10} (see discussion in
Section~\ref{subsec:rpso}).

\subsection{Probabilities and information theory}\label{subsec:informationbackground}

Recall that, for a countable set $\Omega$, a \emph{discrete distribution} (or
  \emph{distribution} for short) is a mapping $\mu :
  \Omega \rightarrow [0, 1]$ such that $\sum_{\omega \in \Omega}
  \mu(\omega) = 1$. For any subset $E$ of $\Omega$,
  $\mu(E)=\sum_{\omega \in E} \mu(\omega)$. The set of all discrete
  distributions on $\Omega$ is denoted by $\mathcal{D}(\Omega)$.
  A \emph{discrete random variable} with values in a set $\Gamma$ is a mapping $Z: \Omega \rightarrow \Gamma$ where $[Z=z]$ denotes
the event $\{\omega \in \Omega \mid Z(\omega) = z\}$.
  
The \emph{entropy} of
$Z$ is a measure of the uncertainty or dually, information about $Z$,
defined by the expected value of $\log (\mu(Z))$:
\[H(Z) = -\sum_{z} {\mu(Z=z) \cdot \log(\mu(Z=z))}\]  where $\log$ is the
base 2 logarithm.

For two random variables $Z$ and $Z'$ on $\Omega$, the
\emph{conditional entropy} of $Z $ \emph{given the event} $[Z'=z']$
such that $\mu(Z'=z') \not=0$ is defined by:
\[H(Z|Z'=z') = -\sum_{z}  \left(\mu(Z=z|Z'=z')\cdot  \log(\mu(Z=z|Z'=z'))\right)\]
where $\mu(Z=z|Z'=z') =  \frac{\mu(Z=z, Z'=z')}{\mu(Z'=z')}$.
\medskip

The \emph{conditional entropy} of $Z$ \emph{given the random variable}
$Z'$  can be interpreted as the average entropy of $Z$ that remains
after the observation of $Z'$. It is defined by:
\[ H(Z|Z') = \sum_{z'} \mu(Z'=z') \cdot H(Z|Z'=z') \]

The \emph{vulnerability} of a random variable $Z$, defined by $V(Z) =
\max_{z} \mu(Z=z)$ gives the probability of the likeliest event of a
random variable.  Vulnerability evaluates the probability of a correct
guess in one attempt and can also be used as a measure of information
by defining \emph{min-entropy} and \emph{conditional min-entropy} (see
discussions in~\cite{smith09,andres10}).

\subsection{Probabilistic models}
\label{subsec:bg-probautomata}

In this work, systems are modeled using probabilistic automata
behaving as finite automata where non-deterministic choices for the
next action and state or termination are randomized: this is why they
are called ``\emph{fully} probabilistic''. We follow the model
definition of~\cite{segala95}, which advocates for the use of this
special termination action (denoted here by $\surd$ instead of
$\delta$ there). However, the difference lies in the model semantics:
we consider only finite runs, which involves a modified definition for
the (discrete) probability on the set of runs. Extensions to the
non-deterministic setting are discussed in
Section~\ref{sec:schedulers}.

Recall that a finite automaton (FA) is a tuple $\A = \langle \Sigma,
Q, \Delta, I, F\rangle$ where $\langle \Sigma, Q, \Delta, I \rangle$
is a finite transition system and $F \subseteq Q$ is a subset of final
states.  The automaton is deterministic if $I$ is a singleton and for
all $q \in Q$ and $a \in \Sigma$, the set $\{q' \mid (q,a,q') \in
\Delta\}$ is a singleton. Runs in $\A$, $Run_q(\A)$ and $Run(\A)$
are defined like in a transition system. A run of an FA is
\emph{accepting} if it ends in a state of $F$. The \emph{trace} of a
run \(\rho = q_0 \xrightarrow{a_1} q_1 \cdots \xrightarrow{a_n} q_n\)
is the word $\trace(\rho) = a_1 \cdots a_n \in \Sigma^*$.  The
\emph{language} of $\A$, written $\lang(\A)$, is the set of traces
of accepting runs starting in some initial state.

Replacing in a FA non-deterministic choices by choices based on a
discrete distribution and considering only finite runs result in a \emph{fully probabilistic
  finite automaton} (FPFA). Consistently with the standard notion of
substochastic matrices, we also consider a more general class of
automata, \emph{substochastic automata} (SA), which allow us
to describe subsets of behaviors
from FPFAs, see \figurename~\ref{fig:exfpaandsa} for examples. In both
models, no non-determinism remains, thus the system is to be
considered as autonomous: its behaviors do not depend on an exterior
probabilistic agent acting as a scheduler for non-deterministic
choices.

\begin{definition}[Substochastic automaton]
  Let $\surd$ be a new symbol representing a termination action. A
  \emph{substochastic automaton} (SA) is a tuple $\langle \Sigma, Q,
  \Delta, q_0 \rangle$ where $\Sigma$ is a finite set of actions, $Q$
  is a finite set of states, with $q_0 \in Q$ the initial state and
  $\Delta: Q \rightarrow ((\Sigma \times Q) \uplus \{\surd\}
  \rightarrow [0,1])$ is a mapping such that for any $q \in
  Q$, \[\sum_{x \in (\Sigma \times Q) \uplus \{\surd\}} \Delta(q)(x)
  \leq 1\] $\Delta$ defines substochastically the action and successor
  from the current state, or the termination action $\surd$.
\end{definition}
In SA, we write $q \rightarrow \mu$ for $ \Delta(q) = \mu$
and $q \xrightarrow{a} r$ whenever $q \rightarrow \mu$ and $\mu(a, r)
>0$. We also write $q \cdot \surd$ whenever $q \rightarrow \mu$ and
$\mu(\surd) >0$. In the latter case, $q$ is said to be a \emph{final
  state}.

\begin{definition}[Fully probabilistic finite automaton]
A \emph{fully probabilistic automaton} (FPFA) is a particular case of 
SA where for all $q \in Q$, $\Delta(q)= \mu$ is a distribution in
$\probset((\Sigma \times Q) \uplus \{\surd\})$ \textit{i.e.}
\[\sum_{x \in (\Sigma \times Q) \uplus \{\surd\}} \Delta(q)(x) =
1\] and for any state $q \in Q$ there exists a path (with non-zero
probability) from $q$ to some final state.
\end{definition}
Note that we only target finite runs and therefore we consider a
restricted case, where any infinite path has probability $0$.

Since FPFA is a subclass of SA, we overload
the metavariable $\A$ for both SA and FPFA.
The notation above allows to define a run for an SA like in a
transition system as a finite sequence of transitions written \(\rho =
q_0 \xrightarrow{a_1} q_1 \xrightarrow{a_2} q_2 \, \cdots
\xrightarrow{a_n} q_n\). The sets $Run_q(\A)$ and $Run(\A)$ are
defined like in a transition system.  A \emph{complete run} is a
(finite) sequence denoted by $\rho\cdot\surd$ where $\rho$ is a run
and $\Delta(\lst(\rho))(\surd) > 0$. The set $CRun(\A)$ denotes the
set of complete runs starting from the initial state. In this work, we
consider only such complete runs.

The \emph{trace} of a run for an SA $\A$ is defined like in finite
automata.  The \emph{language} of a substochastic automaton $\A$,
written $\lang(\A)$, is the set of traces of complete runs starting
in the initial state.

For an SA $\A$, a mapping $\prob_\A$ into $[0,1]$ can be defined
inductively on the set of complete runs by:
\[
\prob_\A(q \surd) = \mu(\surd) \qquad \mbox{ and } \qquad 
\prob_\A(q\xrightarrow{a} \rho)  =  \mu(a, r) \cdot \prob_\A(\rho)
\]
where $q \rightarrow \mu$ and $\fst(\rho) = r$.

The mapping $\prob_\A$ is then a discrete distribution on
$CRun(\A)$. Indeed, the $\surd$ action can be
seen as a transition label towards a new sink state $q_{\surd}$. Then,
abstracting from the labels yields a finite Markov chain, where
$q_{\surd}$ is the only absorbing state and the coefficients of the
transition matrix are $M_{q,q'}=\sum_{a \in
  \Sigma}\Delta(q)(a,q')$. The probability for a complete run to have
length $n$ is then the probability $p_n$ to reach $q_{\surd}$ in
exactly $n$ steps. Therefore, the probability of all finite complete
runs is $\prob(CRun(\A))=\sum_{n} p_n$ and a classical
result~\cite{KS76} on absorbing chains ensures that this probability
is equal to $1$.

\medskip
Since the probability space is not generated by
(prefix-closed) cones, this definition does not yield the same
probability measure as the one from~\cite{segala95}. Since opacity
properties are not necessarily prefix-closed, this definition is
consistent with our approach.

When $\A$ is clear from the context, $\prob_\A$ will simply be
written $\prob$.

Since $\prob_\A$ is a (sub-)probability on $CRun(\A)$, for any
predicate $\varphi \subseteq CRun(\A)$, we have $\prob(\varphi) =
\sum_{\rho \in \varphi} \prob(\rho)$.  The measure is extended to
languages $K \subseteq \lang(\A)$ by \(\prob(K) =
\prob\left(\trace^{-1}(K)\right) = \sum_{\trace(\rho) \in K}
\prob(\rho)\).

In the examples of \figurename~\ref{fig:exfpaandsa}, restricting the complete
runs of $\mathcal{A}_1$ to those satisfying $\varphi = \{\rho \mid
\trace(\rho) \in a^*\}$ yields the SA $\mathcal{A}_2$, and
$\prob_{\mathcal{A}_1}(\varphi) =
\prob_{\mathcal{A}_2}(CRun(\mathcal{A}_2)) = \frac12$.

\medskip

A non probabilistic version of any SA is obtained by forgetting any
information about probabilities.
\begin{definition}
  Let $\A = \langle \Sigma, Q, \Delta, q_0\rangle$ be an SA.  The
  (non-deterministic) finite automaton $unProb(\A) = \langle \Sigma,
  Q, \Delta', q_0, F\rangle$ is defined by:
\begin{itemize}
\item $\Delta' = \{(q,a,r) \in Q \times \Sigma \times Q \mid q
  \rightarrow \mu, \ \mu(a,r)>0\}$,
\item $F = \{q \in Q \mid q \rightarrow \mu, \mu(\surd)>0\}$ is the
  set of final states.
\end{itemize}
\end{definition}
It is easily seen that ${\cal L}(unProb(\A)) = {\cal L}(\A)$.

An observation function ${\cal O}: CRun(\A) \rightarrow Obs$ can also
be easily translated from the probabilistic to the non probabilistic
setting. For $\A'= unProb(\A)$, we define $unProb({\cal O})$ on
$Run(\A')$ by $unProb({\cal O}) (q_0 \xrightarrow{a_1} q_1 \cdots
q_n) = {\cal O}(q_0 \xrightarrow{a_1} q_1 \cdots q_n\surd)$.

\begin{figure}
\centering
\hfill{~}
\subfigure[FPFA $\mathcal{A}_1$]{\label{fig:exfpa}
\begin{tikzpicture}[node distance=4cm,thick,initial text=,auto]
\tikzstyle{state}+=[minimum size=7pt]
\node[state,initial] (q0) at (0,0) {};
\node[state] (q1) [right of=q0] {};
\node (f0) [node distance=0.75cm,below of=q0] {};
\node (f1) [node distance=0.75cm,below of=q1] {};

\path[->] (q0) edge [out=135,in=45,min distance=0.75cm] node {$a,\frac12$} (q0);
\path[->] (q0) edge node {$b,\frac14$} (q1);
\path[->] (q0) edge node {$\surd,\frac14$} (f0);
\path[->] (q1) edge node {$\surd,1$} (f1);
\end{tikzpicture}
}
\hfill{~}\hfill{~}
\subfigure[SA $\mathcal{A}_2$]{\label{fig:exsa}
\begin{tikzpicture}[node distance=6cm,thick,initial text=,auto]
\tikzstyle{state}+=[minimum size=7pt]
\node[state,initial] (q0) at (0,0) {};
\node (f0) [node distance=0.75cm,below of=q0] {};

\path[->] (q0) edge [out=135,in=45,min distance=0.75cm] node {$a,\frac12$} (q0);
\path[->] (q0) edge node {$\surd,\frac14$} (f0);
\end{tikzpicture}
}
\hfill{~}
\caption{$\mathcal{A}_2$ is the restriction of $\mathcal{A}_1$ to $a^*$.}
\label{fig:exfpaandsa}
\end{figure}
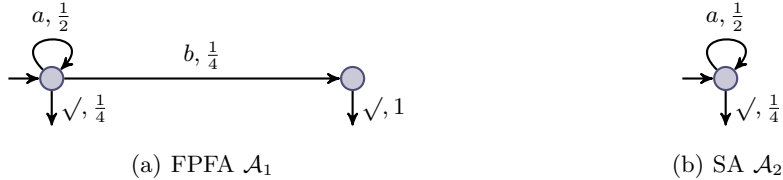

\section{Measuring non-opacity}\label{sec:liberal}

\subsection{Definition and properties}\label{subsec:liberal-definition}

One of the aspects in which the definition of opacity could be
extended to probabilistic automata is by relaxing the universal
quantifiers of Definitions~\ref{def:opacity} and~\ref{def:symopacity}.
Instead of wanting that \emph{every} observation class should not be
included in $\varphi$ (resp. $\varphi$ or $\overline\varphi$ for the
symmetrical case), we can just require that \emph{almost all} of them
do.  To obtain this, we give a measure for the set of runs leaking
information. To express properties of probabilistic opacity in an FPA
$\A$, the observation function $\cal O$ is considered as a random
variable, as well as the characteristic function ${\bf 1}_{\varphi}$
of $\varphi$.  Both the asymmetrical and the symmetrical notions of
opacity can be generalized in this manner.

\begin{definition}[Liberal probabilistic opacity]
  The \emph{liberal probabilistic opacity} or \lponame of predicate
  $\varphi$ on FPA $\A$, with respect to (surjective) observation
  function ${\cal O}: CRun \rightarrow Obs$ is defined by:
\[\lpo(\A,\varphi,{\cal O}) = \sum_{\substack{o \in Obs \\ 
{\cal O}^{-1}(o) \subseteq \varphi}} \prob(\mathcal{O} = o).\]
The \emph{liberal probabilistic symmetrical opacity} or \lpsoname is defined by:
\begin{eqnarray*}
\lpso(\A,\varphi,{\cal O}) &=& \lpo(\A,\varphi,{\cal O}) + \lpo(\A,\overline\varphi,{\cal O})\\
&=& \sum_{\substack{o \in Obs \\ 
{\cal O}^{-1}(o) \subseteq \varphi}} \prob(\mathcal{O} = o) + \sum_{\substack{o \in Obs \\ 
{\cal O}^{-1}(o) \subseteq \overline\varphi}} \prob(\mathcal{O} = o).
\end{eqnarray*}
\end{definition}

This definition provides a measure of how insecure the system is. 
The following propositions shows that a null value for these measures
coincides with (symmetrical) opacity for the system, which is then
secure.

For \lponame, it corresponds to classes either overlapping both
$\varphi$ and $\overline{\varphi}$ or included in $\overline\varphi$
as in \figurename~\ref{fig:lpo0}.  \lponame measures only the classes
that leak their inclusion in $\varphi$.  So classes included in
$\overline\varphi$ are not taken into account.  On the other extremal
point, $\lpo(\A,\varphi,\mathcal{O}) = 1$ when $\varphi$ is always
true.

When \lpsoname is null, it means that each equivalence class ${\cal
  O}^{-1}(o)$ overlaps both $\varphi$ and $\overline{\varphi}$ as in
\figurename~\ref{fig:lpso0}.  On the other hand, the system is totally
insecure when, observing through ${\cal O}$, we have all information
about $\varphi$.  In that case, the predicate $\varphi$ is a union of
equivalence classes ${\cal O}^{-1}(o)$ as in
\figurename~\ref{fig:lpso1} and this can be interpreted in terms of
conditional entropy relatively to ${\cal O}$.  The intermediate case
occurs when some, but not all, observation classes contain only runs
satisfying $\varphi$ or only runs not satisfying $\varphi$, as in
\figurename~\ref{fig:lpsomid}.

\begin{figure}
\centering
\hfill{~}
\subfigure[$\lpo(\A,\varphi,{\cal O}) = 0$]{\label{fig:lpo0}
\begin{tikzpicture}[scale=3/5]
\useasboundingbox (-0.75,0) rectangle (4.75,5);
\path[phee, rounded corners=9pt]  (2.25,1.5) -- (2.5,2.5) -- (3.9,2.9) -- (2.5,3.4) -- (3.75,4.5) -- (2.5,4.5) -- (1.9,3) -- (1.5,4.5) -- (0.25,4.5) -- (0.5,2.5) -- (1.5,2.75) [rounded corners=6pt]-- (1.5,1.5) -- cycle;
\draw[very thin] (0,0) grid (4,5);
\draw[thick] (0,0) -- (0,5) -- (4,5) -- (4,0) -- cycle;
\end{tikzpicture}
}
\hfill{~}\hfill{~}
\subfigure[$0 < \lpo(\A,\varphi,{\cal O}) < 1$]{\label{fig:lpomid}
\begin{tikzpicture}[scale=3/5]
\useasboundingbox (-0.75,0) rectangle (4.75,5);
\path[phee, rounded corners=8.5pt] (0.25,0.25) -- (2,0.1) -- (3.5,0.75) -- (3.75,1.75) -- (2.25,1.5) -- (2.5,2.5) -- (2.3,3.4) -- (3.75,4.95) -- (1.5,4.5) -- (0.25,4.5) -- (0.5,2.5) -- (1.5,2.75) -- (0.5,1.5) -- cycle;
\draw[very thin] (0,0) grid (4,5);
\draw[thick] (0,0) -- (0,5) -- (4,5) -- (4,0) -- cycle;
\path[pattern=north east lines] (1,1) rectangle (2,2);
\path[pattern=north east lines] (1,3) rectangle (2,4);
\end{tikzpicture}
}
\hfill{~}\hfill{~}
\subfigure{\label{fig:lpolegende}
\begin{tikzpicture}[scale=3/5]
\useasboundingbox (-0.25,1) rectangle (4.75,-3.875); % Centrage vertical de la l\'egende % [draw=red]
\node (pheestart) at (0,0) {};
\path[phee] (pheestart) [rounded corners=3pt]-- +(-0.3,0.5) [rounded corners=4pt]-- +(0,1) [rounded corners=7pt]-- +(1,1) [rounded corners=4pt]-- +(1.4,0.3) [rounded corners=3pt]-- +(1,0) [rounded corners=5pt]-- +(0.6,0.5) [rounded corners=2pt]-- cycle;
\node[anchor=west] at ($(pheestart) + (1.5,0.5)$) {\Large$\varphi$};

\begin{scope}[yshift=-1.375cm]
\draw[very thin] (-0.1,-0.1) grid (1.1,1.1);
\node[anchor=west] at (1.5,0.5) {\Large${\cal O}^{-1}(o)$};
\end{scope}

\begin{scope}[yshift=-2.875cm]
\path[fill=auxfill] (0,0) rectangle (1,1);
\draw[very thin] (-0.1,-0.1) grid (1.1,1.1);
\path[pattern=north east lines] (0,0) rectangle (1,1);
\node[anchor=west,text width=2.25cm] at (1.5,0.5) {\parbox{2.2cm}{\scriptsize Classes leaking their inclusion into $\varphi$}};
\end{scope}

\begin{scope}[yshift=-4.375cm]
\draw[very thin] (-0.1,-0.1) grid (1.1,1.1);
\fill[pattern=north west lines] (0,0) rectangle (1,1);
\node[anchor=west,text width=2.25cm] at (1.5,0.5) {\parbox{2.2cm}{\scriptsize Classes leaking their inclusion into $\overline\varphi$}};
\end{scope}
\end{tikzpicture}
}
\hfill{~}
\medskip
\addtocounter{subfigure}{-1}

\hfill{~}
\subfigure[$\lpso(\A,\varphi,{\cal O}) = 0$]{\label{fig:lpso0}
\begin{tikzpicture}[scale=3/5]
\useasboundingbox (-0.75,0) rectangle (4.75,5);
\path[phee, rounded corners=9pt] (0.25,0.25) -- (2,0.1) -- (3.5,0.75) -- (3.75,1.75) -- (2.25,1.5) -- (2.5,2.5) -- (3.9,2.9) -- (2.5,3.4) -- (3.75,4.5) -- (2.5,4.5) -- (1.9,3) -- (1.5,4.5) -- (0.25,4.5) -- (0.5,2.5) -- (1.5,2.75) -- (1.5,1.5) -- (0.5,1.5) -- cycle;
\draw[very thin] (0,0) grid (4,5);
\draw[thick] (0,0) -- (0,5) -- (4,5) -- (4,0) -- cycle;
\end{tikzpicture}
}
\hfill{~}\hfill{~}
\subfigure[$0 < \lpso(\A,\varphi,{\cal O}) < 1$]{\label{fig:lpsomid}
\begin{tikzpicture}[scale=3/5]
\useasboundingbox (-0.75,0) rectangle (4.75,5);
\path[phee, rounded corners=8.5pt] (0.25,0.25) -- (2,0.1) -- (3.5,0.75) -- (3.75,1.75) -- (2.25,1.5) -- (2.5,2.5) -- (2.3,3.4) -- (3.75,4.95) -- (1.5,4.5) -- (0.25,4.5) -- (0.5,2.5) -- (1.5,2.75) -- (0.5,1.5) -- cycle;
\draw[very thin] (0,0) grid (4,5);
\draw[thick] (0,0) -- (0,5) -- (4,5) -- (4,0) -- cycle;
\path[pattern=north east lines] (1,1) rectangle (2,2);
\path[pattern=north east lines] (1,3) rectangle (2,4);
\path[pattern=north west lines] (3,2) rectangle (4,4);
\end{tikzpicture}
}
\hfill{~}\hfill{~}
\subfigure[$\lpso(\A,\varphi,{\cal O}) = 1$]{\label{fig:lpso1}
\begin{tikzpicture}[scale=3/5]
\useasboundingbox (-0.75,0) rectangle (4.75,5);
\path[phee] (0,5) -- (1,5) -- (1,4) -- (3,4) -- (3,3) -- (1,3) -- (1,1) -- (2,1) -- (2,2) -- (3,2) -- (3,0) -- (0,0) -- cycle;
\draw[very thin] (0,0) grid (4,5);
\draw[pattern=north east lines] (0,5) -- (1,5) -- (1,4) -- (3,4) -- (3,3) -- (1,3) -- (1,1) -- (2,1) -- (2,2) -- (3,2) -- (3,0) -- (0,0) -- cycle;
\draw[pattern=north west lines] (1,5) -- (4,5) -- (4,0) -- (3,0) -- (3,2) -- (2,2) -- (2,1) -- (1,1) -- (1,3) -- (3,3) -- (3,4) -- (1,4) -- cycle;
\draw[thick] (0,0) -- (0,5) -- (4,5) -- (4,0) -- cycle;
\end{tikzpicture}
}
\hfill{~}
\caption{Liberal probabilistic asymmetrical and symmetrical opacity.}
\label{fig:lpopicture}
\label{fig:lpsopicture}
\end{figure}
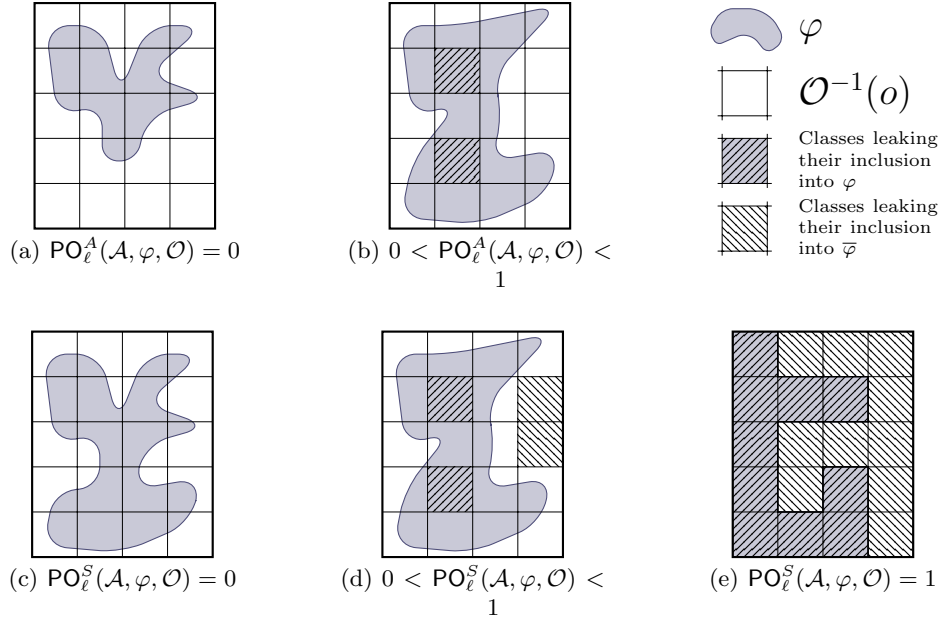

\begin{proposition}\label{prop:lpoprop}~
\begin{enumerate}[label=(\arabic*),topsep=-\baselineskip]
\item $0 \leq \lpo(\A,\varphi,{\cal O}) \leq 1$ and $0 \leq
  \lpso(\A,\varphi,{\cal O}) \leq 1$
\item $\lpo(\A,\varphi,{\cal O}) = 0$ if and only if $\varphi$ is
  opaque on $unProb(\A)$ with respect to $unProb({\cal O})$.\\
  $\lpso(\A,\varphi,{\cal O}) = 0$ if and only if $\varphi$ is
  symmetrically opaque on $unProb(\A)$ with respect to $unProb({\cal
    O})$.
 \item \label{it:lpsoprop:lpo1} $\lpo(\A,\varphi,{\cal O}) = 1$ if
   and only if $\varphi = CRun(\A)$.\\
   $\lpso(\A,\varphi,{\cal O}) = 1$ if and only if $H({\bf
     1}_{\varphi} | {\cal O}) = 0$.
\end{enumerate}
\end{proposition}

\begin{proof}[Proof of Proposition~\ref{prop:lpoprop}.]
\begin{enumerate}[label=(\arabic*),topsep=-\baselineskip,labelsep=3pt]
\item The considered events are mutually exclusive, hence the sum of
  their probabilities never exceeds 1.
\item First observe that a complete run $r_0 a \dots r_n \surd$ has a non null
  probability in $\A$ iff $r_0 a \dots r_n$ is a run in
  $unProb(\A)$.  Suppose $\lpo(\A,\varphi,{\cal O}) = 0$. Recall that ${\cal O}$ is assumed surjective. Then there is
  no observable $o$ such that $\mathcal{O}^{-1}(o) \subseteq \varphi$.
  Conversely, if $\varphi$ is opaque on $unProb(\A)$, there is no observable $o \in Obs$
   such that $\mathcal{O}^{-1}(o) \subseteq \varphi$, hence the null value for $\lpo(\A,\varphi,{\cal O})$.
  The case of \lpsoname is similar, also taking into account the dual case of $\overline\varphi$ in the above.
  \item For \lponame, this is straightforward from the definition.
  For \lpsoname, $H({\bf 1}_{\varphi} | {\cal O}) = 0$ iff
  \[\sum_{\stackrel{o \in Obs}{i \in \{0, 1\}}}
  \prob({\bf 1}_{\varphi} = i | {\cal O} =  o) \cdot \log(\prob({\bf 1}_{\varphi} = i | {\cal O} =  o)) = 0
  \]
  Since all the terms have the same sign, this sum is null if and only
  if each of its term is null.  Setting for every $o \in Obs$, $f(o) =
  \prob({\bf 1}_{\varphi} = 1 | {\cal O} = o) = 1 - \prob({\bf
    1}_{\varphi} = 0 | {\cal O} = o)$, we have: $H({\bf 1}_{\varphi} |
  \mathcal{O}) = 0$ iff $\forall\, o \in Obs$, $f(o) \cdot \log(f(o))
  + (1 - f(o)) \cdot \log(1 - f(o)) = 0$.  Since the equation $x \cdot
  \log(x) + (1-x) \cdot \log(1-x) = 0$ only accepts $1$ and $0$ as
  solutions, it means that for every observable $o$, either all the
  runs $\rho$ such that ${\cal O}(\rho) = o$ are in $\varphi$, or they
  are all not in $\varphi$.  Therefore $H({\bf 1}_{\varphi} | {\cal
    O}) = 0$ iff for every observable $o$, ${\cal O}^{-1}(o) \subseteq
  \varphi$ or ${\cal O}^{-1}(o) \subseteq \overline{\varphi}$, which
  is equivalent to $\lpso(\A,\varphi,{\cal O}) = 1$.
\end{enumerate}
\end{proof}

\subsection{Example: Non-interference}
\label{ex:niwithopacity}
For the systems $\mathcal{A}_3$ and $\mathcal{A}_4$ of
\figurename~\ref{fig:niPA}, we use the predicate $\varphi_{NI}$
which is true if the trace of a run contains the letter $h$.  In both
cases the observation function ${\cal O}_L$ returns the projection of
the trace onto the alphabet $\{\ell_1,\ell_2\}$.  Remark that this
example is an interference property~\cite{goguen82} seen as opacity.
Considered unprobabilistically, both systems are interferent since an
$\ell_2$ not preceded by an $\ell_1$ betrays the presence of an $h$.
However, they differ by how often this case happens.

The runs of $\mathcal{A}_3$ and $\mathcal{A}_4$ and their properties
are displayed in \tablename~\ref{tab:niruns}.  Then we can see that
$[\rho_1]_{{\cal O}_L} = [\rho_2]_{{\cal O}_L}$ overlaps both
$\varphi_{NI}$ and $\overline{\varphi_{NI}}$, while $[\rho_3]_{{\cal
    O}_L}$ is contained totally in $\varphi$.  Hence the \lponame can
be computed for both systems:
\[
\lpo(\mathcal{A}_3,\varphi_{NI},{\cal O}_L) = \frac14 \qquad
\lpo(\mathcal{A}_4,\varphi_{NI},{\cal O}_L) = \frac34
\]
Therefore $\mathcal{A}_3$ is more secure than $\mathcal{A}_4$.
Indeed, the run that is interferent occurs more often in
$\mathcal{A}_4$, leaking information more often.

Note that in this example, \lponame and \lpsoname coincide.  This is
not always the case.  Indeed, in the unprobabilistic setting, both
symmetrical and asymmetrical opacity of $\varphi_{NI}$ with respect to
$\mathcal{O}_L$ express the intuitive notion that ``an external
observer does not know whether an action happened or not''.  The
asymmetrical notion corresponds to the definition of \emph{strong
  nondeterministic non-interference} in~\cite{goguen82} while the
symmetrical one was defined as \emph{perfect security property}
in~\cite{alur06}.

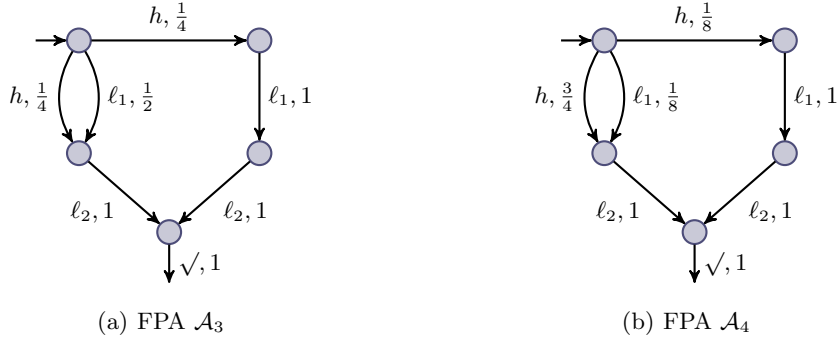
\begin{figure}[t]
\centering
\hfill{~}
\subfigure[FPA $\mathcal{A}_3$]{\label{fig:niPPA1}
\begin{tikzpicture}[thick,initial text=,auto]
\tikzstyle{ynode}=[node distance=2.5cm]
\tikzstyle{xnode}=[node distance=4cm]
\tikzstyle{every state}+=[scale=0.6,minimum size=15pt]
\node[state,initial] (q0) at (0,0) {};
\node[state] (q1) [ynode,below of=q0] {};
\node[state] (q2) [xnode,right of=q0] {};
\node[state] (q3) [ynode,below of=q2] {};
\node (ph) at (barycentric cs:q1=1,q3=1) {};
\node[state] (q4) [node distance=1.75cm,below of=ph] {};
\node (q5) [node distance=0.8cm,below of=q4] {};

\path[->] (q0) edge [bend left] node {$\ell_1,\frac12$} (q1);
\path[->] (q0) edge [bend right] node [swap] {$h,\frac14$} (q1);
\path[->] (q0) edge node {$h,\frac14$} (q2);
\path[->] (q2) edge node {$\ell_1,1$} (q3);
\path[->] (q1) edge node [swap] {$\ell_2,1$} (q4);
\path[->] (q3) edge node {$\ell_2,1$} (q4);
\path[->] (q4) edge node {$\surd,1$} (q5);
\end{tikzpicture}
}
\hfill{~}\hfill{~}
\subfigure[FPA $\mathcal{A}_4$]{\label{fig:niPPA2}
\begin{tikzpicture}[node distance=4cm,thick,initial text=,auto]
\tikzstyle{ynode}=[node distance=2.5cm]
\tikzstyle{xnode}=[node distance=4cm]
\tikzstyle{every state}+=[scale=0.6,minimum size=15pt]
\node[state,initial] (q0) at (0,0) {};
\node[state] (q1) [ynode,below of=q0] {};
\node[state] (q2) [xnode,right of=q0] {};
\node[state] (q3) [ynode,below of=q2] {};
\node (ph) at (barycentric cs:q1=1,q3=1) {};
\node[state] (q4) [node distance=1.75cm,below of=ph] {};
\node (q5) [node distance=0.8cm,below of=q4] {};

\path[->] (q0) edge [bend left] node {$\ell_1,\frac18$} (q1);
\path[->] (q0) edge [bend right] node [swap] {$h,\frac34$} (q1);
\path[->] (q0) edge node {$h,\frac18$} (q2);
\path[->] (q2) edge node {$\ell_1,1$} (q3);
\path[->] (q1) edge node [swap] {$\ell_2,1$} (q4);
\path[->] (q3) edge node {$\ell_2,1$} (q4);
\path[->] (q4) edge node {$\surd,1$} (q5);
\end{tikzpicture}
}
\hfill{~}
\caption{Interferent FPAs $\mathcal{A}_3$ and $\mathcal{A}_4$.}
\label{fig:niPA}
\end{figure}

\begin{table}
\centering
\begin{tabular}{|c|c|c|c|c|}
\hline
$\trace(\rho)$ & $\prob_{\mathcal{A}_3}(\rho)$ & $\prob_{\mathcal{A}_4}(\rho)$ & $\in \varphi_{NI}$? & ${\cal O}_L(\rho)$\\\hline
$\trace(\rho_1) = \ell_1\ell_2\surd$  & 1/2 & 1/8 & 0 & $\ell_1\ell_2$\\\hline
$\trace(\rho_2) = h\ell_1\ell_2\surd$ & 1/4 & 1/8 & 1 & $\ell_1\ell_2$\\\hline
$\trace(\rho_3) = h\ell_2\surd$    & 1/4 & 3/4 & 1 & $\ell_2$\\\hline
\end{tabular}
\caption{Runs of $\mathcal{A}_3$ and $\mathcal{A}_4$.}
\label{tab:niruns}
\end{table}

\section{Measuring the robustness of opacity}\label{sec:restrictive}

The completely opposite direction that can be taken to define a
probabilistic version is a more paranoid one: how much information is
leaked through the system's uncertainty?  For example, on
\figurename~\ref{fig:lpso0}, even though each observation class
contains a run in $\varphi$ and one in $\overline\varphi$, some
classes are \emph{nearly} in $\varphi$.  In some other classes the
balance between the runs satisfying $\varphi$ and the ones not
satisfying $\varphi$ is more even.  Hence for each observation class,
we will not ask if it is included in $\varphi$, but how likely
$\varphi$ is to be true inside this class with a probabilistic measure
taking into account the likelihood of classes.  This amounts to measuring, inside each
observation class, $\overline\varphi$ in the case of asymmetrical
opacity, and the balance between $\varphi$ and $\overline\varphi$ in
the case of symmetrical opacity. Note that these new measures are
relevant only for opaque systems, where the previous liberal measures
are equal to zero.

In~\cite{berard10}, another measure was proposed, based on the notion
of mutual information (from information theory, along similar lines as
in~\cite{smith09}). However, this measure had a weaker link with
possibilistic opacity (see discussion in
Section~\ref{sec:comparison}).  What we call here \hponame is a new
measure, whose relation with possibilistic opacity is expressed by the
second item in Proposition~\ref{prop:hpoprop}.

\subsection{Restricting Asymmetrical Opacity}\label{subsec:hpo}

In this section we extend the notion of asymmetrical opacity in
order to measure how secure the system is.

\subsubsection*{Definition and properties.}
In this case, an observation class is more secure if $\varphi$ is less
likely to be true.  That means that it is easy (as in ``more likely'')
to find a run not in $\varphi$ in the same observation class.  Dually,
a high probability for $\varphi$ inside a class means that few (again
probabilistically speaking) runs will be in the same class yet not in
$\varphi$.

Restrictive probabilistic opacity is defined to measure this effect
globally on all observation classes.  It is tailored to fit the
definition of opacity in the classical sense: indeed, if one class
totally leaks its presence in $\varphi$, \hponame will detect it
(second point in Proposition~\ref{prop:hpoprop}).

\begin{definition}\label{def:hpo}  Let $\varphi$ be a predicate on the 
  complete runs of an FPA $\A$ and ${\cal O}$ an observation
  function.  The \emph{restrictive probabilistic opacity} (\hponame)
  of $\varphi$ on $\A$, with respect to ${\cal O}$, is defined by
\[\frac1{\hpo(\A,\varphi,{\cal O})} = \sum_{o\in Obs} \prob(\mathcal{O}=o)\cdot \frac{1}{\prob(\ind_\varphi=0 \mid \mathcal{O}=o)}\]
\end{definition}

\hponame is the harmonic means (weighted by the probabilities of
observations) of the probability that $\varphi$ is false in a given
observation class.  The harmonic means averages the leakage of
information inside each class.
Since security and robustness are often evaluated on the weakest link,
more weight is given to observation classes with the higher leakage,
\emph{i.e.} those with probability of $\varphi$ being false closest to
$0$.

\medskip The following proposition gives properties of \hponame.
\begin{proposition}\label{prop:hpoprop}~
\begin{enumerate}[label=(\arabic*),topsep=-\baselineskip]
\item $0 \leq \hpo(\A,\varphi,\Obs) \leq 1$
\item $\hpo(\A,\varphi,\Obs) = 0$ if and only if $\varphi$ is not
  opaque on $unProb(\A)$ with respect to $unProb(\Obs)$.
\item $\hpo(\A,\varphi,\Obs) = 1$ if and only if $\varphi = \emptyset$.
\end{enumerate}
\end{proposition}
\begin{proof}
  The first point immediately results from the fact that \hponame is a
  means of values between $0$ and $1$.

  From the definition above,
  \hponame is null if and only if there is one class that is contained
  in $\varphi$. Indeed, this
  corresponds to the case where the value of
  $\frac1{\prob(\ind_\varphi=0 \mid \mathcal{O}=o)}$ goes to
  $+\infty$, for some $o$, as well as the sum.

Thirdly, if $\varphi$ is always false, then \hponame is $1$ since it
is a means of probabilities all of value $1$. Conversely, if \hponame
is $1$, because it is defined as an average of values between $0$ and
$1$, then all these values must be equal to $1$, hence for each $o$,
$\prob(\ind_\varphi=0 \mid \mathcal{O}=o)=1$ which means that
$\prob(\ind_\varphi=0)=1$ and $\varphi$ is false.
\end{proof}

\subsubsection*{Example: Debit Card System.}\label{ex:debitcard}
Consider a Debit Card system in a store.  When a card is inserted, an
amount of money $x$ to be debited is entered, and the client enters
his \textsc{pin} number (all this being gathered as the action
Buy$(x)$).  The amount of the transaction is given probabilistically
as an abstraction of the statistics of such transactions.  Provided
the \textsc{pin} is correct, the system can either directly allow the
transaction, or interrogate the client's bank for solvency.  In order
to balance the cost associated with this verification (bandwidth,
server computation, etc.) with the loss induced if an insolvent client
was debited, the decision to interrogate the bank's servers is taken
probabilistically according to the amount of the transaction.  When
interrogated, the bank can reject the transaction with a certain
probability\footnote{Although the bank process to allow or forbid the
  transaction is deterministic, the statistics of the result can be
  abstracted into probabilities.} or accept it.  This system is
represented by the FPA $\A_{\textrm{card}}$ of
\figurename~\ref{fig:debitcard}.
\begin{figure}[h]
\centering
\begin{tikzpicture}[auto]
\tikzstyle{every state}+=[minimum size=2pt]
\node[state,initial] (q0) at (-1.5,0) {};
\node[state] (qi) at (0,0) {};
\node[state] (high) at (4,2.625) {};
\node[state] (medhigh) at (4,0.875) {};
\node[state] (medlow) at (4,-0.875) {};
\node[state] (low) at (4,-2.625) {};

\tikzstyle{every node}+=[font=\scriptsize]

\foreach \prize in {high,medhigh,medlow,low} {
\node[state] (\prize call) at ($(\prize) + (2,0.5)$) {};
\node[state] (\prize acc) at ($(\prize) + (3,-0.5)$) {};
\node[state] (\prize rej) at ($(\prize call) + (2,0)$) {};
\foreach \res in {acc,rej}{
\node (\prize\res ph) at ($(\prize\res) + (0.875,0)$) {};
\path[->] (\prize\res) edge node {$\surd,1$} (\prize\res ph);
};
};

\path[->] (q0) edge node {Buy$(x)$} (qi);
\path[->] (qi) edge[bend left] node {$x>1000, 0.05$} (high);
\path[->] (qi) edge node[pos=0.875] {$500 < x \leq 1000, 0.2$} (medhigh);
\path[->] (qi) edge node[swap,pos=0.875] {$100 < x \leq 500, 0.45$} (medlow);
\path[->] (qi) edge[bend right] node[swap] {$x \leq 100, 0.3$} (low);

\path[->] (high) edge node[pos=0.75] {Call, $0.95$} (highcall);
\path[->] (highcall) edge node[pos=0.675] {Accept, $0.8$} (highacc);
\path[->] (highcall) edge node {Reject, $0.2$} (highrej);
\path[->] (high) edge node[swap,pos=0.75] {Accept, $0.05$} (highacc);

\path[->] (medhigh) edge node[pos=0.75] {Call, $0.75$} (medhighcall);
\path[->] (medhighcall) edge node[pos=0.675] {Accept, $0.9$} (medhighacc);
\path[->] (medhighcall) edge node {Reject, $0.1$} (medhighrej);
\path[->] (medhigh) edge node[swap,pos=0.75] {Accept, $0.25$} (medhighacc);

\path[->] (medlow) edge node[pos=0.75] {Call, $0.5$} (medlowcall);
\path[->] (medlowcall) edge node[pos=0.675] {Accept, $0.95$} (medlowacc);
\path[->] (medlowcall) edge node {Reject, $0.05$} (medlowrej);
\path[->] (medlow) edge node[swap,pos=0.75] {Accept, $0.5$} (medlowacc);

\path[->] (low) edge node[pos=0.75] {Call, $0.2$} (lowcall);
\path[->] (lowcall) edge node[pos=0.675] {Accept, $0.99$} (lowacc);
\path[->] (lowcall) edge node {Reject, $0.01$} (lowrej);
\path[->] (low) edge node[swap,pos=0.75] {Accept, $0.8$} (lowacc);
\end{tikzpicture}
\vspace{-\baselineskip}
\caption{The Debit Card system $\A_{\textrm{card}}$.}
\label{fig:debitcard}
\end{figure}
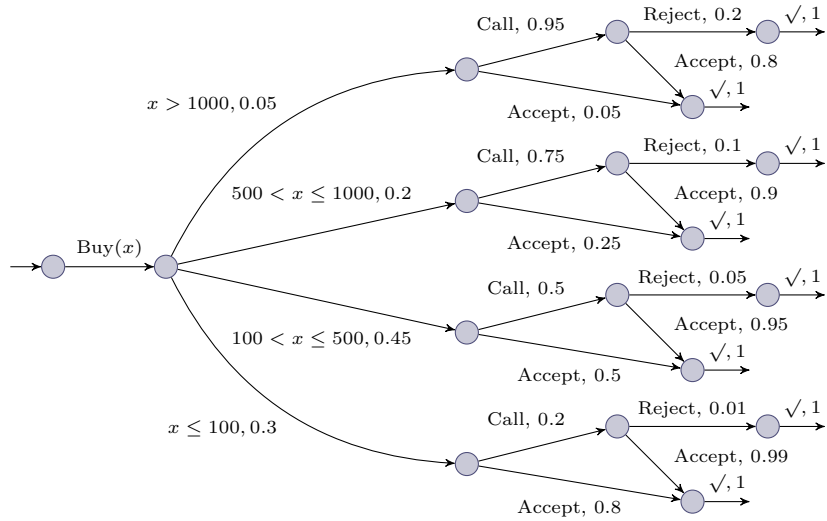

Now assume an external observer can only observe if there has been a
call or not to the bank server.  In practice, this can be achieved,
for example, by measuring the time taken for the transaction to be
accepted (it takes longer when the bank is called), or by spying on
the telephone line linking the store to the bank's servers (detecting
activity on the network or idleness).  Suppose what the external
observer wants to know is whether the transaction was worth more than
500\officialeuro.  By using \hponame, one can assess how this
knowledge can be derived from observation.

Formally, in this case the observables are
$\{\varepsilon,\textrm{Call}\}$, the observation function
$\Obs_{\textrm{Call}}$ being the projection on $\{$Call$\}$.  The
predicate to be hidden to the user is represented by the regular
expression $\varphi_{>500}=\Sigma^*($``$x>1000$''$+$``$500 < x \leq
1000$''$) \Sigma^*$ (where $\Sigma$ is the whole alphabet).  By
definition of \hponame:
\begin{eqnarray*}
\frac{1}{\hpo(\A_{\textrm{card}},\varphi_{>500},\Obs_{\textrm{Call}})} &=& \prob(\Obs_{\textrm{Call}}=\varepsilon) \cdot \frac1{\prob(\neg \varphi_{>500} | \Obs_{\textrm{Call}}=\varepsilon)}
\\&& +\ \prob(\Obs_{\textrm{Call}}=\textrm{Call}) \cdot \frac1{\prob(\neg \varphi_{>500} | \Obs_{\textrm{Call}}=\textrm{Call})}
\end{eqnarray*}

Computing successively $\prob(\Obs_{\textrm{Call}}=\varepsilon)$,
$\prob(\neg \varphi_{>500} | \Obs_{\textrm{Call}}=\varepsilon)$, $\prob(\Obs_{\textrm{Call}}=\textrm{Call})$, and
$\prob(\neg \varphi_{>500} | \Obs_{\textrm{Call}}=\textrm{Call})$ (see Appendix~\ref{app:creditcardcalcul}), we
obtain: \[\hpo(\A_{\textrm{card}},\varphi_{>500},\Obs_{\textrm{Call}})
= \frac{28272}{39377} \ \simeq \ 0.718.\]

\bigskip The notion of asymmetrical opacity, however, fails to capture security
in terms of opacity for both $\varphi$ and $\overline\varphi$.  And so
does the \hponame measure.  Therefore we define in the next section a
quantitative version of symmetrical opacity.

\subsection{Restricting symmetrical opacity}\label{subsec:rpso}

Symmetrical opacity offers a sound framework to analyze the secret of
a binary value.  For example, consider a binary channel with $n$ outputs.  It can be modeled by a
tree-like system branching on $0$ and $1$ at the first level, then
branching on observables $\{o_1,\dots,o_n\}$, as in
\figurename~\ref{fig:binchannel}.  If the system wishes to prevent
communication, the secret of predicate ``the input of the channel was
$1$'' is as important as the secret of its negation; in this case
``the input of the channel was $0$''.  Such case is an example of
\emph{initial opacity}~\cite{bryans08}, since the secret appears only
at the start of each run.  Note that any system with initial opacity
and a finite set of observables can be transformed into a
channel~\cite{andres10}, with input distribution ($p$, $1-p$), which
is the distribution of the secret predicate over
$\{\varphi,\overline\varphi\}$.

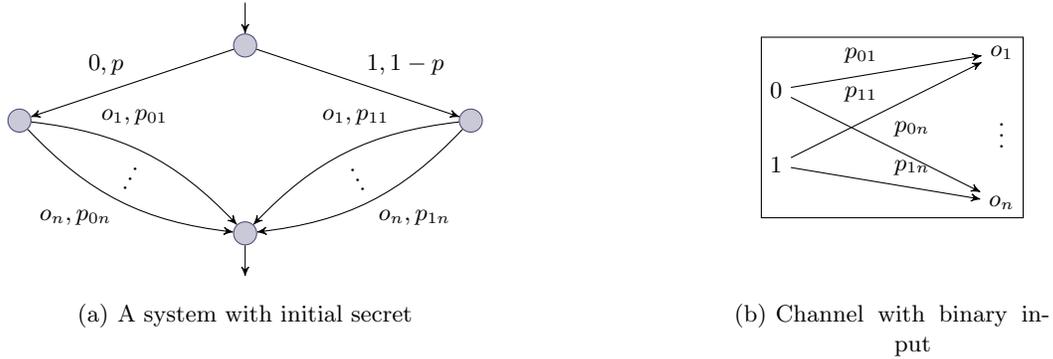
\begin{figure}[h]
\subfigure[A system with initial secret]{\label{fig:binchannelsys}
\begin{tikzpicture}[auto]
\tikzstyle{every state}+=[minimum size=5pt]
\node[state,initial above] (init) at (0,0) {};
\node[state] (i0) at (-3,-1) {};
\node[state] (i1) at (3,-1) {};
\node[state,accepting below] (fin) [node distance=2.5cm,below of=init] {};

\path[->] (init) edge node [swap] {$0,p$} (i0);
\path[->] (init) edge node {$1,1-p$} (i1);

\path[->] (i0) edge [bend left=20] node [pos=0.25] {$o_1,p_{01}$} (fin);
\path (i0) -- node[sloped,anchor=base,rotate=90] {$\dots$} (fin);
\path[->] (i0) edge [bend right=20] node [swap] {$o_n,p_{0n}$} (fin);

\path[->] (i1) edge [bend right=20] node [swap,pos=0.25] {$o_1,p_{11}$} (fin);
\path (i1) -- node[sloped,anchor=base,rotate=90] {$\dots$} (fin);
\path[->] (i1) edge [bend left=20] node {$o_n,p_{1n}$} (fin);

\end{tikzpicture}
}
\hfill{~}
\subfigure[Channel with binary input]{\label{fig:binchanneltrans}
\begin{tikzpicture}[auto]
\node (i0) at (0,0) {$0$};
\node (i1) at (0,-1) {$1$};
\node (od) at (3,-0.5) {$\vdots$};
\node (o1) [above of=od] {$o_1$};
\node (on) [below of=od] {$o_n$};

\draw (on.south east) rectangle (i0.west |- o1.north);
\useasboundingbox ($(on.south east) + (0.25,-1)$) rectangle ($(i0.west |- o1.north) + (-0.25,0)$);

\path[->] (i0) edge node {$p_{01}$} (o1);
\path[->] (i0) edge node {$p_{0n}$} (on);
\path[->] (i1) edge node {$p_{11}$} (o1);
\path[->] (i1) edge node {$p_{1n}$} (on);
\end{tikzpicture}
}
\caption{A system and its associated channel.}
\label{fig:binchannel}
\end{figure}

\subsubsection*{Definition and properties.}
Symmetrical opacity ensures that for each observation class $o$
(reached by a run), the probability of both $\prob(\varphi \mid o)$
and $\prob(\overline\varphi \mid o)$ is strictly above $0$.  That
means that the lower of these probabilities should be above $0$.  In
turn, the lowest of these probability is exactly the complement of the
vulnerability (since $\ind_\varphi$ can take only two values).  That
is, the security is measured with the probability of error in one
guess (inside a given observation class).  Hence, a system will be
secure if, in each observation class, $\varphi$ is \emph{balanced}
with $\overline\varphi$.

\begin{definition}[Restrictive probabilistic symmetric opacity]\label{def:vpo}
  Let $\varphi$ be a predicate on the complete runs of an FPA $\A$
  and ${\cal O}$ an observation function.  The \emph{restrictive
    probabilistic symmetric opacity} (\vponame) of $\varphi$ on $\A$,
  with respect to ${\cal O}$, is defined by
\[\vpo(\A,\varphi,{\cal O}) = \frac{-1}{\sum_{o \in Obs} \prob(\mathcal{O}=o) \cdot \log\left(1-V(\ind_\varphi \mid \mathcal{O}=o)\right)}\]
where $V(\ind_\varphi \mid \mathcal{O}=o) = \max_{i \in \{0,1\}} \prob(\ind_\varphi=i \mid \mathcal{O}=o)$.
\end{definition}

Remark that the definition of \vponame has very few ties with the
definition of \hponame.  Indeed, it is linked more with the notion of
possibilistic symmetrical opacity than with the notion of quantitative
asymmetrical opacity, and thus \vponame is not to be seen as an
extension of \hponame.

In the definition of \vponame, taking $-\log(1-V(\ind_\varphi \mid
\mathcal{O}=o))$ allows to give more weight to very imbalanced
classes, up to infinity for classes completely included either in
$\varphi$ or in $\overline\varphi$.  Along the lines of~\cite{smith09}, the logarithm is used in order to produce a measure in terms
of bits instead of probabilities.  These measures are then averaged
with respect to the probability of each observation class.  The final inversion ensures that the
value is between $0$ and $1$, and can be seen as a normalization
operation.  The above motivations for the definition of \vponame
directly yield the following properties:
\begin{proposition}\label{prop:vpoprop}
\begin{enumerate}[label=(\arabic*),topsep=-\baselineskip]
\item $0 \leq \vpo(\A,\varphi,{\cal O}) \leq 1$
\item $\vpo(\A,\varphi,{\cal O}) = 0$ if and only if $\varphi$ is \emph{not} symmetrically opaque on $unProb(\A)$ with respect to $unProb({\cal O})$.
\item $\vpo(\A,\varphi,{\cal O}) = 1$ if and only if $\forall o \in Obs,\ \prob(\ind_\varphi=1 \mid \mathcal{O}=o) = \frac12$.
\end{enumerate}
\end{proposition}

\begin{proof}[Proof of Proposition~\ref{prop:vpoprop}.]
\begin{enumerate}[label=(\arabic*),topsep=-\baselineskip,labelsep=3pt]
\item Since the vulnerability of a random variable that takes only two values is between $\frac12$ and $1$, we have $1-V(\ind_\varphi \mid \mathcal{O}=o) \in [0,\frac12]$ for all $o \in Obs$.
So $-\log(1-V(\ind_\varphi \mid \mathcal{O}=o)) \in [1,+\infty[$ for any $o$.
The (arithmetic) means of these values is thus contained within the same bounds.
The inversion therefore yields a value between $0$ and $1$.
\item If $\varphi$ is not symmetrically opaque, then for some
  observation class $o$, $\Obs^{-1}(o) \subseteq \varphi$ or
  $\Obs^{-1}(o) \subseteq \overline\varphi$.  In both cases,
  $V(\ind_\varphi \mid \mathcal{O}) = 1$, so $-\log(1-V(\ind_\varphi
  \mid \mathcal{O}=o)) = +\infty$ and the average is also $+\infty$.
  Taking the limit for the inverse gives the value $0$ for \vponame.

  Conversely, if \vponame is $0$, then its inverse is $+\infty$, which
  can only occur if one of the $-\log(1-V(\ind_\varphi \mid
  \mathcal{O}=o))$ is $+\infty$ for some $o$.  This, in turn, means
  that some $V(\ind_\varphi \mid \mathcal{O}=o)$ is $1$, which means
  the observation class of $o$ is contained either in $\varphi$ or in
  $\overline\varphi$.
\item $\vpo(\A,\varphi,{\cal O}) = 1$ iff $\sum_{o \in Obs}
  \prob(\mathcal{O}=o) \cdot \left(-\log\left(1-V(\ind_\varphi \mid
      \mathcal{O}=o)\right)\right) = 1$.  Since this is an average of
  values above $1$, this is equivalent to $-\log\left(1-V(\ind_\varphi
    \mid \mathcal{O}=o)\right)= 1$ for all $o \in Obs$, \emph{i.e.}
  $V(\ind_\varphi \mid \mathcal{O}=o) = \frac12$ for all $o$.  In this
  particular case, we also have $V(\ind_\varphi \mid \mathcal{O}=o) =
  \frac12$ iff $\prob(\ind_\varphi=1 \mid \mathcal{O}=o) = \frac12$
  which concludes the proof.
\end{enumerate}
\end{proof}

\subsubsection*{Example 1: Sale protocol.}
We consider the sale protocol from~\cite{alvim10}, depicted in
\figurename~\ref{fig:ebay}.  Two products can be put on sale, either a
cheap or an expensive one, and two clients, either a rich or a poor
one, may want to buy it.  The products are put on sale according to a
distribution ($\alpha$ and $\overline\alpha = 1-\alpha$) while buyers
behave probabilistically (through $\beta$ and $\gamma$) although
differently according to the price of the item on sale.
\begin{figure}[h]
\centering
\begin{tikzpicture}[auto,level distance=6mm,level/.style={sibling distance=0.95*\textwidth/2^#1}]
\tikzstyle{edge from parent}=[draw,->]
\tikzstyle{state}+=[minimum size=3pt]
\node[state] (init) at (0,0) {}
  child {node[state] (c) {}
    child {node[state] (cp) {} edge from parent node[anchor=south east] {poor, $\beta$}}
    child {node[state] (cr) {} edge from parent node[anchor=south west] {rich, $\overline\beta$}}
    edge from parent node[anchor=south east] {cheap,$\alpha$}}
  child {node[state] (e) {}
    child {node[state] (ep) {} edge from parent node[anchor=south east] {poor, $\gamma$}}
    child {node[state] (er) {} edge from parent node[anchor=south west] {rich, $\overline\gamma$}}
    edge from parent node[anchor=south west] {expensive,$\overline\alpha$}};
    
\path[edge from parent] ($(init.north) + (0,0.5)$) -- (init);
\foreach \fnode in {cp,cr,ep,er} {
  \path[edge from parent] (\fnode) -- ($(\fnode.south) - (0,0.5)$);
  \node[anchor=north west] at (\fnode.south) {$\surd,1$};
}
\end{tikzpicture}
\caption{A simple sale protocol represented as an FPA $\mathcal{S}ale$.}
\label{fig:ebay}
\end{figure}
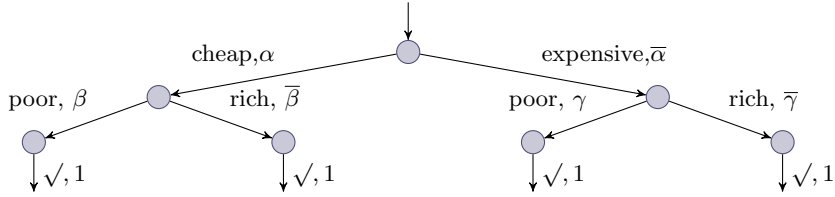
\begin{figure}[h]
\centering
\subfigure[$\vpo(\mathcal{S}ale,\varphi_{\mbox{\scriptsize poor}}, Price)$ when $\alpha=\frac18$]{\label{fig:graphsebayeighth}
\begin{tikzpicture}[scale=4.5]
\node[anchor=south west,inner sep=0pt] at (0,0) {\pgfimage[width=4.5cm]{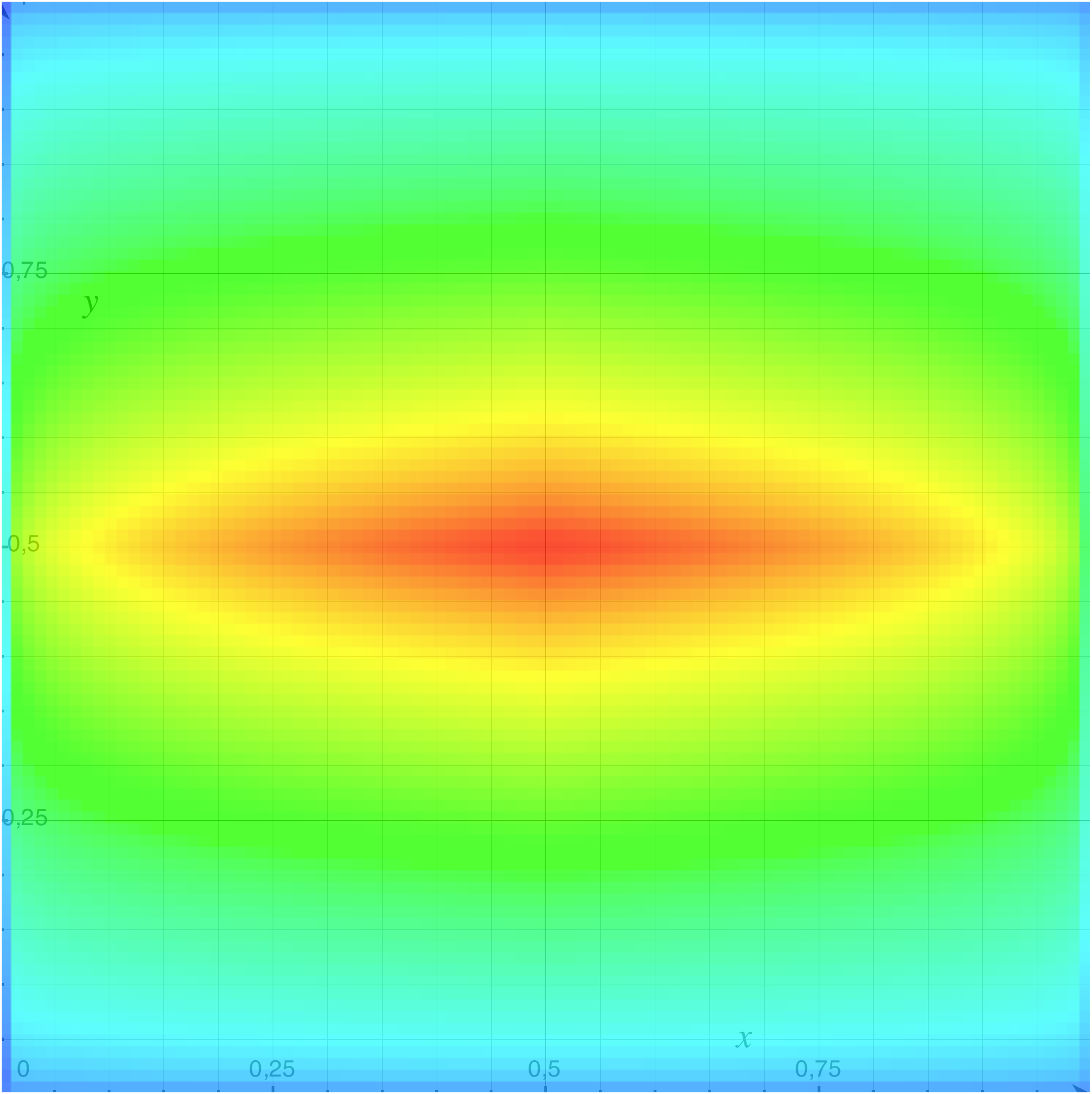}};
\path[->,draw] (0,0) -- (1.1,0) node[anchor=north] {$\beta$};
\path[->,draw] (0,0) -- (0,1.1) node[anchor=east] {$\gamma$};
\foreach \x in {0,0.25,0.5,0.75,1} {%
  \draw (\x,-0.0125) -- (\x,0.0125);
  \node[anchor=north] at (\x,0) {\x};};
\foreach \y in {0,0.25,0.5,0.75,1} {%
  \draw (-0.0125,\y) -- (0.0125,\y);
  \node[anchor=east] at (0,\y) {\y};};
\end{tikzpicture}
}
\subfigure[$\vpo(\mathcal{S}ale,\varphi_{\mbox{\scriptsize poor}}, Price)$ when $\alpha=\frac12$]{\label{fig:graphsebayhalf}
\begin{tikzpicture}[scale=4.5]
\node[anchor=south west,inner sep=0pt] at (0,0) {\pgfimage[width=4.5cm]{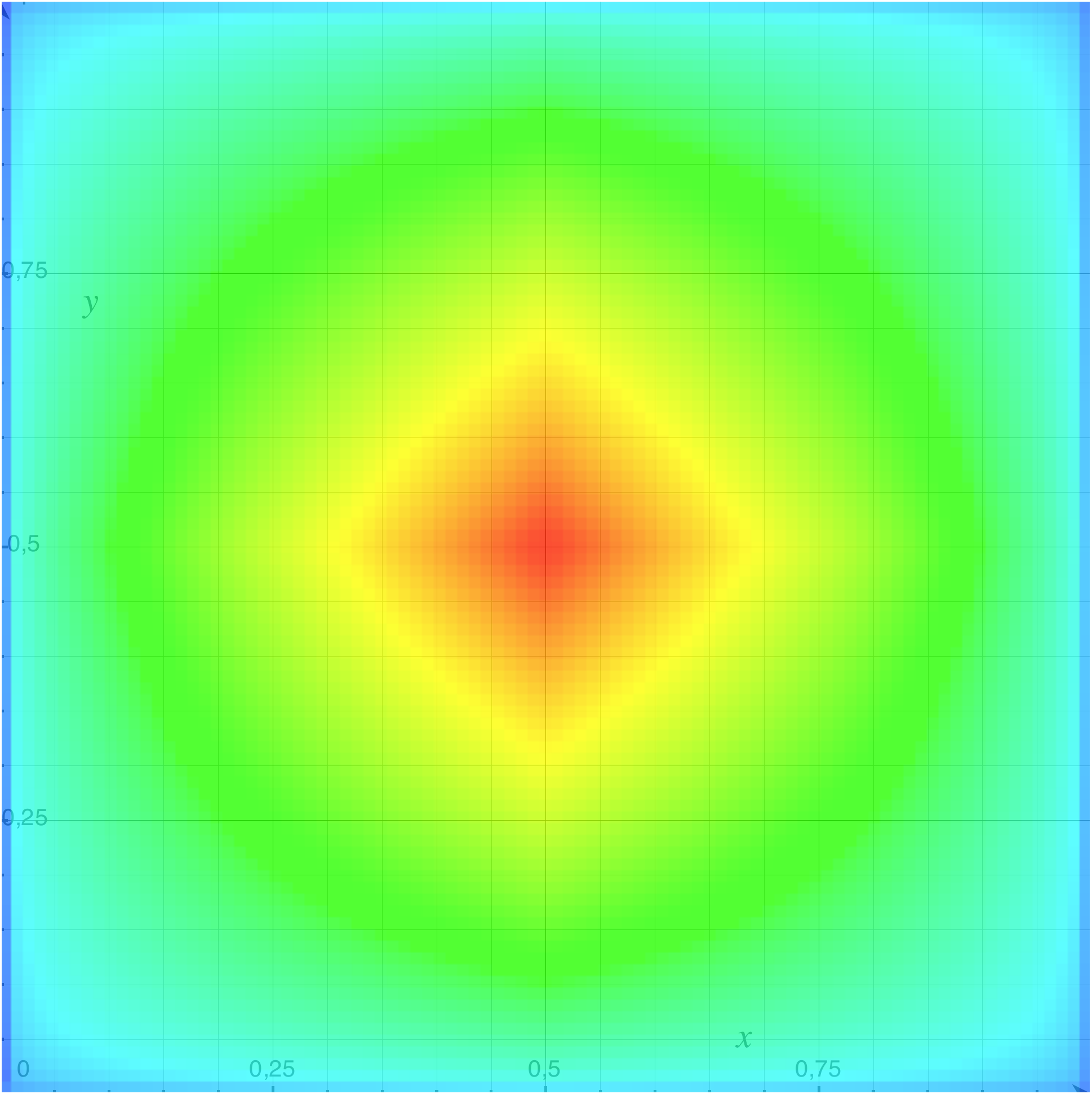}};
\path[->,draw] (0,0) -- (1.1,0) node[anchor=north] {$\beta$};
\path[->,draw] (0,0) -- (0,1.1) node[anchor=east] {$\gamma$};
\foreach \x in {0,0.25,0.5,0.75,1} {%
  \draw (\x,-0.0125) -- (\x,0.0125);
  \node[anchor=north] at (\x,0) {\x};};
\foreach \y in {0,0.25,0.5,0.75,1} {%
  \draw (-0.0125,\y) -- (0.0125,\y);
  \node[anchor=east] at (0,\y) {\y};};
\end{tikzpicture}
}
\caption{\vponame for the sale protocol.}
\label{fig:graphsebay}
\end{figure}
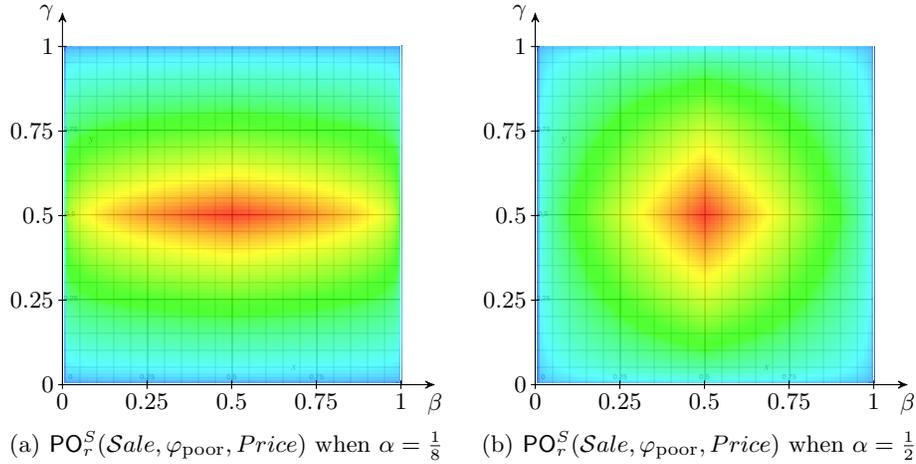
The price of the item is public, but the identity of the buyer should
remain secret.  Hence the observation function $Price$ yields
\emph{cheap} or \emph{expensive}, while the secret is, without loss of
symmetry, the set $\varphi_{\mbox{\scriptsize poor}}$ of runs ending
with \emph{poor}.  The bias introduced by the preference of, say, a
cheap item by the poor client betrays the secret identity of the
buyer.  \vponame allows to measure this bias, and more importantly, to
compare the bias obtained globally for different values of the
parameters $\alpha$, $\beta$, and $\gamma$.

More formally, we have:
\begin{mathpar}
\prob(Price=\mbox{cheap}) = \alpha
\and
\prob(Price=\mbox{expensive}) = \overline\alpha
\and
V(\ind_{\varphi_{\mbox{\scriptsize poor}}} \mid \Obs=\mbox{cheap}) = \max(\beta,\overline\beta)
\and
V(\ind_{\varphi_{\mbox{\scriptsize poor}}} \mid \Obs=\mbox{expensive}) = \max(\gamma,\overline\gamma)
\and
\vpo(\mathcal{S}ale,\varphi_{\mbox{\scriptsize poor}}, Price) = \frac{-1}{\alpha \cdot \log(\min(\beta,\overline\beta)) + \overline\alpha \cdot \log(\min(\gamma,\overline\gamma))}.
\end{mathpar}
The variations of \vponame w.r.t. to $\beta$ and $\gamma$ is depicted for several values of $\alpha$ in
\figurename~\ref{fig:graphsebay}, red meaning higher value for
\vponame. Thus, while the result is symmetric for $\alpha=\frac12$,
the case where $\alpha=\frac18$ gives more importance to the
fluctuations of $\gamma$.

\subsubsection*{Example 2: Dining Cryptographers Protocol.}\label{subsubsec:dcpvpo}
Introduced in~\cite{chaum88}, this problem involves three
cryptographers $C_1$, $C_2$ and $C_3$ dining in a restaurant. At the
end of the meal, their master secretly tells each of them if they
should be paying: $p_i = 1$ iff cryptographer $C_i$ pays, and $p_i =
0$ otherwise.  Wanting to know if one of the cryptographers paid or if
the master did, they follow the following protocol.  They flip a coin
with each of their neighbor, the third one not seeing the result of
the flip, marking $f_{i,j} = 0$ if the coin flip between $i$ and $j$
was heads and $f_{i,j} = 1$ if it was tails.  Then each cryptographer
$C_i$, for $i \in \{1,2,3\}$, announces the value of $r_i = f_{i,i+1}
\oplus f_{i,i-1} \oplus p_i$ (where `$3 + 1=1$',
`$1-1=3$' and `$\oplus$' represents the XOR operator).  If $\bigoplus_{i=1}^3
r_i = 0$ then no one (\emph{i.e.} the master) paid, if
$\bigoplus_{i=1}^3 r_i = 1$, then one of the cryptographers paid, but
the other two do not know who he is.

Here we will use a simplified version of this problem to limit the
size of the model.  We consider that some cryptographer paid for the
meal, and adopt the point of view of $C_1$ who did not pay.  The
anonymity of the payer is preserved if $C_1$ cannot know if $C_2$ or $C_3$
paid for the meal.  In our setting, the predicate $\varphi_2$ is,
without loss of symmetry, ``$C_2$ paid''. 
Note that predicate $\varphi_2$ is well suited for analysis of symmetrical opacity, since detecting that $\varphi_2$ is false gives information on who paid (here $C_3$). The observation
function lets $C_1$ know the results of its coin flips ($f_{1,2}$ and
$f_{1,3}$), and the results announced by the other cryptographers
($r_2$ and $r_3$).  We also assume that the coin used by $C_2$ and $C_3$
has a probability of $q$ to yield heads, and that the master flips a
fair coin to decide if $C_2$ or $C_3$ pays.  It can be assumed that the
coins $C_1$ flips with its neighbors are fair, since it does not affect
anonymity from $C_1$'s point of view.  In order to limit the
(irrelevant) interleaving, we have made the choice to fix the ordering
between the coin flips.

The corresponding FPA $\mathcal{D}$ is depicted on
\figurename~\ref{fig:diningPPA} where all $\surd$ transitions with probability $1$ have been omitted from final (rectangular) states.
On $\mathcal{D}$, the runs
satisfying predicate $\varphi_2$ are the ones where action $p_2$
appears.  The observation function ${\cal O}_1$ takes a run and returns the
sequence of actions over the alphabet $\{h_{1,2}, t_{1,2}, h_{1,3},
t_{1,3}\}$ and the final state reached, containing the value announced
by $C_2$ and $C_3$.

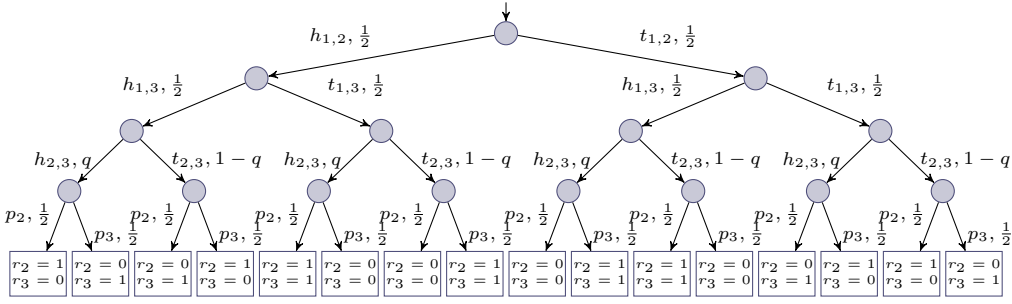
\begin{figure}[h]
\centering
\begin{tikzpicture}[level distance=11mm,level/.style={sibling distance=0.95*\textwidth/2^#1}]
\tikzstyle{level 1}+=[level distance=6mm]
\tikzstyle{level 2}+=[level distance=7mm]
\tikzstyle{level 3}+=[level distance=8mm]
\tikzstyle{level 4}+=[level distance=11mm]
\tikzstyle{textless}=[circle,aux,minimum size=3pt]
\tikzstyle{textfull}=[inner xsep=0pt,draw=auxdraw,font=\tiny,text width=0.75cm,text centered]
\tikzstyle{edgetxt}=[font=\scriptsize]
\tikzstyle{lastleft}=[pos=0.25]
\tikzstyle{lastright}=[pos=0.65]
\tikzstyle{edge from parent}=[draw,->]

\node[textless] (origin) at (0,0) {}
  child {node [textless] (h) {}
    child {node [textless] (hh) {}
      child {node [textless] (hhh) {}
	child {node [textfull] (hhh2) {$r_2=1$ $r_3=0$} edge from parent node[left,edgetxt,lastleft] {$p_2,\frac12$}}
	child {node [textfull] (hhh3) {$r_2=0$ $r_3=1$} edge from parent node[right,edgetxt,lastright] {$p_3,\frac12$}}
      edge from parent node[left,edgetxt] {$h_{2,3},q$}
      }
      child {node [textless] (hht) {}
	child {node [textfull] (hht2) {$r_2=0$ $r_3=1$} edge from parent node[left,edgetxt,lastleft] {$p_2,\frac12$}}
	child {node [textfull] (hht3) {$r_2=1$ $r_3=0$} edge from parent node[right,edgetxt,lastright] {$p_3,\frac12$}}
      edge from parent node[right,edgetxt] {$t_{2,3},1-q$}
      }
    edge from parent node[above left,edgetxt] {$h_{1,3},\frac12$}
    }
    child {node [textless] (ht) {}
      child {node [textless] (hth) {}
	child {node [textfull] (hth2) {$r_2=1$ $r_3=1$} edge from parent node[left,edgetxt,lastleft] {$p_2,\frac12$}}
	child {node [textfull] (hth3) {$r_2=0$ $r_3=0$} edge from parent node[right,edgetxt,lastright] {$p_3,\frac12$}}
      edge from parent node[left,edgetxt] {$h_{2,3},q$}
      }
      child {node [textless] (htt) {}
	child {node [textfull] (htt2) {$r_2=0$ $r_3=0$} edge from parent node[left,edgetxt,lastleft] {$p_2,\frac12$}}
	child {node [textfull] (htt3) {$r_2=1$ $r_3=1$} edge from parent node[right,edgetxt,lastright] {$p_3,\frac12$}}
      edge from parent node[right,edgetxt] {$t_{2,3},1-q$}
      }
    edge from parent node[above right,edgetxt] {$t_{1,3},\frac12$}
    }
  edge from parent node[above left,edgetxt] {$h_{1,2},\frac12$}
  }
  child {node [textless] (t) {}
    child {node [textless] (th) {}
      child {node [textless] (thh) {}
	child {node [textfull] (thh2) {$r_2=0$ $r_3=0$} edge from parent node[left,edgetxt,lastleft] {$p_2,\frac12$}}
	child {node [textfull] (thh3) {$r_2=1$ $r_3=1$} edge from parent node[right,edgetxt,lastright] {$p_3,\frac12$}}
      edge from parent node[left,edgetxt] {$h_{2,3},q$}
      }
      child {node [textless] (tht) {}
	child {node [textfull] (tht2) {$r_2=1$ $r_3=1$} edge from parent node[left,edgetxt,lastleft] {$p_2,\frac12$}}
	child {node [textfull] (tht3) {$r_2=0$ $r_3=0$} edge from parent node[right,edgetxt,lastright] {$p_3,\frac12$}}
      edge from parent node[right,edgetxt] {$t_{2,3},1-q$}
      }
    edge from parent node[above left,edgetxt] {$h_{1,3},\frac12$}
    }
    child {node [textless] (tt) {}
      child {node [textless] (tth) {}
	child {node [textfull] (tth2) {$r_2=0$ $r_3=1$} edge from parent node[left,edgetxt,lastleft] {$p_2,\frac12$}}
	child {node [textfull] (tth3) {$r_2=1$ $r_3=0$} edge from parent node[right,edgetxt,lastright] {$p_3,\frac12$}}
      edge from parent node[left,edgetxt] {$h_{2,3},q$}
      }
      child {node [textless] (ttt) {}
	child {node [textfull] (ttt2) {$r_2=1$ $r_3=0$} edge from parent node[left,edgetxt,lastleft] {$p_2,\frac12$}}
	child {node [textfull] (ttt3) {$r_2=0$ $r_3=1$} edge from parent node[right,edgetxt,lastright] {$p_3,\frac12$}}
      edge from parent node[right,edgetxt] {$t_{2,3},1-q$}
      }
    edge from parent node[above right,edgetxt] {$t_{1,3},\frac12$}
    }
  edge from parent node[above right,edgetxt] {$t_{1,2},\frac12$}
  }
;
\path[edge from parent] ($(origin.north) + (0,0.25)$) -- (origin);
\end{tikzpicture}

\caption{The FPA corresponding to the Dining Cryptographers protocol.}
\label{fig:diningPPA}
\end{figure}

There are 16 possible complete runs in this system, that yield 8 equiprobable
observables:
\[
\begin{array}{r@{}l}
Obs = \{ & (h_{1,2}h_{1,3}(r_2=1,r_3=0)), (h_{1,2}h_{1,3}(r_2=0,r_3=1)),\\
& (h_{1,2}t_{1,3}(r_2=0,r_3=0)), (h_{1,2}t_{1,3}(r_2=1,r_3=1)),\\
& (t_{1,2}h_{1,3}(r_2=0,r_3=0)), (t_{1,2}h_{1,3}(r_2=1,r_3=1)),\\
& (t_{1,2}t_{1,3}(r_2=1,r_3=0)), (t_{1,2}t_{1,3}(r_2=0,r_3=1))
\,\}
\end{array}
\]
Moreover, each observation results in a run in which $C_2$ pays and a
run in which $C_3$ pays, this difference being masked by the secret coin
flip between them.  For example, runs $\rho_h =
h_{1,2}h_{1,3}h_{2,3}p_2(r_2=1,r_3=0)$ and $\rho_t =
h_{1,2}h_{1,3}t_{2,3}p_3(r_2=1,r_3=0)$ yield the same observable $o_0
= h_{1,2}h_{1,3}(r_2=1,r_3=0)$, but the predicate is true in the first
case and false in the second one.  Therefore, if $0 < q < 1$, the
unprobabilistic version of $\mathcal{D}$ is  opaque.
However, if $q \neq \frac12$, for each observable, one of them is
\emph{more likely} to be lying, therefore paying.  In the
aforementioned example, when observing $o_0$, $\rho_h$ has occurred
with probability $q$, whereas $\rho_t$ has occurred with probability
$1-q$.  \vponame can measure this advantage globally.

For each observation class, the vulnerability of $\varphi_2$ is $\max(q,1-q)$.
Hence the \vponame will be
\[\vpo(\mathcal{D},\varphi_2,{\cal O}_1) = \frac{-1}{\log(\min(q,1-q))}\]
The variations of the \vponame when changing the bias on $q$ are depicted in \figurename~\ref{fig:diningVPO}.
Analysis of \vponame according to the variation of $q$ yields that the system is perfectly secure if there is no bias on the coin, and insecure if $q=0$ or $q=1$.

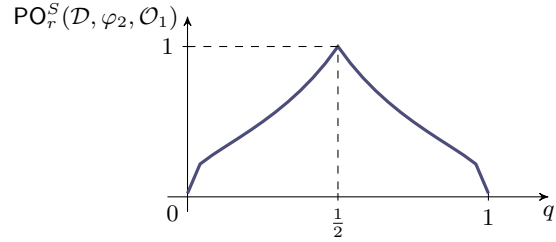
\begin{figure}
\hfill~
% pdflatex --jobname=mscs2012-arxiv-graph-DcpVpo mscs2012-arxiv.tex
\beginpgfgraphicnamed{mscs2012-arxiv-graph-DcpVpo}
\begin{tikzpicture}[yscale=2,xscale=4]
\draw[->] (-2/30,0) -- (1.2,0);
\node[anchor=north] at (1.2,0) {$q$};

\draw[->] (0,-0.1) -- (0,1.2);
\node[anchor=east] at (0,1.2) {$\vpo(\mathcal{D},\varphi_2,{\cal O}_1)$};

\node[anchor=north east] at (0,0) {$0$};

\draw (-1/60,1) -- (1/60,1);
\node[anchor=east] at (-1/60,1) {$1$};

\draw (0.5,-0.025) -- (0.5,0.025);
\node[anchor=north] at (0.5,-0.025) {$\frac12$};
\draw (1,-0.025) -- (1,0.025);
\node[anchor=north] at (1,-0.025) {$1$};

\draw[domain=0:0.9999999999999,very thick,color=auxdraw] plot[id=DcpVpo] function{(-1 * log(2))/ log((x < (1-x)) ? x : (1-x))};
\draw[dashed] (0,1) -- (0.5,1) -- (0.5,0);
\end{tikzpicture}
\endpgfgraphicnamed
\hfill~
\caption{Evolution of the restrictive probabilistic symmetric opacity of the Dining Cryptographers protocol when changing the bias on the coin.}
\label{fig:diningVPO}
\end{figure}

\section{Computing opacity measures}\label{sec:computing}

We now show how all measures defined above can be computed for regular predicates and simple
observation functions. The method relies on a synchronized product
between an SA $\A$ and a deterministic FA $\cal K$, similarly
to~\cite{courcoubetis98}. This product (which can be considered
pruned of its unreachable states and states not reaching a final state)
constrains the unprobabilistic version of
$\A$ by synchronizing it with $\mathcal{K}$.
The probability of $\lang(\mathcal{K})$ is then obtained by solving a system of equations associated with this product.
The computation of all measures results in 
applications of this operation with several automata.

\subsection{Computing the probability of a substochastic automaton}
\label{subsec:bg-computls}
Given an SA $\A$, a system of equations can be derived on the
probabilities for each state to yield an accepting run. This allows to
compute the probability of all complete runs of $\A$ by a technique
similar to those used in~\cite{courcoubetis98,hansson94,bianco95} for
probabilistic verification.

\begin{definition}[Linear system of a substochastic automata]
  Let $\A = \langle \Sigma, Q, \Delta, q_0\rangle$ be a substochastic
  automaton where any state can reach a final state. The {\em linear
    system associated with $\A$} is the following system ${\cal S}_\A$
  of linear equations over $\mathbb{R}$:
\[{\cal S}_\A = \left(X_q  =  \sum_{q' \in Q}  \alpha_{q,q'} X_{q'} 
+ \beta_q\right)_{ q \in Q}\]
\[\textrm{where}\quad\alpha_{q,q'} = \sum_{a \in \Sigma} \Delta(q)(a, q') 
\mbox{ and } \beta_q = \Delta(q)(\surd)\]
\end{definition}

When non-determinism is involved, for instance in Markov Decision
Processes~\cite{courcoubetis98,bianco95}, two systems of inequations
are needed to compute maximal and minimal probabilities. Here, without
non-determinism, both values are the same, hence Lemma
\ref{lem:solsys} is a particular case of the results
in~\cite{courcoubetis98,bianco95}, where uniqueness is ensured by the
hypothesis (any state can reach a final state). The probability can
thus be computed in polynomial time by solving the linear system
associated with the SA.

\begin{lemma}\label{lem:solsys}
  Let $\A = \langle \Sigma, Q, \Delta, q_0\rangle$ be a substochastic
  automaton and define for all $q \in Q$,  $L^\A_q =
    \prob{(CRun_q(\A))}$. Then $(L^\A_q)_{q \in Q}$
  is the unique solution of the system ${\cal S}_\A$.
\end{lemma}

\subsection{Computing the probability of a regular language}

In order to compute the probability of a language inside a system, we build a substochastic automaton that corresponds to the intersection of the system and the language, then compute the probability as above.

\begin{definition}[Synchronized product]
  Let $\A = \langle \Sigma, Q, \Delta, q_0\rangle$ be a substochastic
  automaton and let ${\cal K} = \langle Q \times \Sigma \times Q, Q_K, \Delta_K, q_K,
  F\rangle$ be a deterministic finite automaton. The synchronized
  product $\A || {\cal K}$ is the substochastic automaton $\langle
  \Sigma, Q \times Q_K, \Delta', (q_0, q_K)\rangle$ where transitions
  in $ \Delta'$ are defined by: if $q_1 \rightarrow \mu \in \Delta$,
  then $(q_1,r_1) \rightarrow \nu \in \Delta'$ where for all $a \in
  \Sigma$ and $(q_2,r_2) \in Q \times Q_K$,
\[\nu(a, (q_2, r_2)) = \left\{
\begin{array}{ll}
\mu(a, q_2)  & \mbox{ if }  r_1 \xrightarrow{q_1,a,q_2} r_2  \in \Delta_K \\
0 & \mbox{ otherwise}
\end{array}
\right.
\]

\[\mbox{and} \quad \nu(\surd) = \left\{
\begin{array}{ll}
  \mu(\surd)  & \mbox{ if }  r_1 \in F \\
  0 & \mbox{ otherwise}
\end{array}
\right.
\]
\end{definition}
In this synchronized product, the behaviors are constrained by the
finite automaton.  Actions not allowed by the automaton are trimmed,
and states can accept only if they correspond to a valid behavior of
the DFA.  Note that this product is defined on SA in order to allow
several intersections.  The correspondence between the probability of
a language in a system and the probability of the synchronized product
is laid out in the following lemma.
\begin{lemma} \label{easylemma}
  Let $\A = \langle \Sigma, Q, \Delta, q_0\rangle$ be an SA and $K$
  a regular language over $Q \times \Sigma \times Q$ accepted by a deterministic finite
  automaton $\mathcal{K}=\langle Q \times \Sigma \times Q, Q_K, \Delta_K, q_K,
  F\rangle$. Then
\[\prob_\A{(K) } = L^{\A || {\cal K}}_{(q_0, q_K)}\]
\end{lemma}

\begin{proof}
  Let $\rho \in CRun(\A)$ with $\trace(\rho) \in K$ and $\rho = q_0
  \xrightarrow{a_1} q_1 \cdots \xrightarrow{a_n} q_n \surd$.  Since
  $\trace(\rho) \in K$ and $\mathcal{K}$ is deterministic, there is a
  unique run $\rho_K = q_K \xrightarrow{a_1} r_1 \cdots
  \xrightarrow{a_n} r_n$ in $\mathcal{K}$ with $r_n \in F$.  Then the
  sequence $\rho' = (q_0,q_K) \xrightarrow{a_1} (q_1,r_1) \cdots
  \xrightarrow{a_n} (q_n,r_n)$ is a run of $\A || {\cal K}$.
  There is a one-to one-match between runs of $\A || {\cal K}$ and
  pairs of runs in $\A$ and $\mathcal{K}$ with the same trace.
  Moreover,
\begin{eqnarray*}
  \prob_{\A || {\cal K}}(\rho') &=& \Delta'(q_0,q_K)(a_1,(q_1,r_1)) \times \cdots \times \Delta'(q_n,r_n)(\surd) \\
  & = & \Delta(q_0)(a_1,q_1) \times \cdots \times \Delta(q_n)(\surd) \\
  \prob_{\A || {\cal K}}(\rho') &=& \prob_\A(\rho).
\end{eqnarray*}

Hence
\[
\prob_\A(K) \ =\  \sum_{\{\rho \mid  \trace(\rho)\in K\}} \prob_\A(\rho) \ =\   \sum_{\rho' \in Run(\A || {\cal K})} \prob_{\A || {\cal K}}(\rho') \ =\   \prob_{\A || {\cal K}}(Run(\A || {\cal K}))
\]
and therefore from Lemma~\ref{lem:solsys},
$\prob_\A(K) = L^{\A || {\cal K}}_{(q_0, q_0')}$.
\end{proof}

\subsection{Computing all opacity measures}

All measures defined previously can be computed as long as, for $i \in \{0,1\}$ and $o \in Obs$, all probabilities
\begin{mathpar}
\prob(\ind_\varphi = i)
\and
\prob(\Obs=o)
\and
\prob(\ind_\varphi = i,\Obs=o)
\end{mathpar}
can be computed.
Indeed, even deciding whether $\Obs^{-1}(o) \subseteq \varphi$ can be done by testing $\prob(\Obs=o)>0 \wedge \prob(\ind_\varphi=0,\Obs=o)=0$.

Now suppose $Obs$ is a finite set, $\varphi$ and all $\Obs^{-1}(o)$ are regular sets.
Then one can build deterministic finite automata $\A_\varphi$, $\A_{\overline\varphi}$, $\A_o$ for $o \in Obs$ that accept respectively $\varphi$, $\overline\varphi$, and $\Obs^{-1}(o)$.

 Synchronizing automaton $\A_\varphi$
  with $\A$ and pruning it yields a substochastic automaton
  $\A||\A_\varphi$.  By Lemma~\ref{easylemma}, the probability
  $\prob(\ind_\varphi=1)$ is then computed by solving the linear system
  associated with $\A||\A_\varphi$.
  Similarly, one obtain $\prob(\ind_\varphi = 0)$ (with $\A_{\overline\varphi}$), $\prob(\Obs=o)$ (with $\A_o$), $\prob(\ind_\varphi=1,\Obs=o)$ (synchronizing $\A||\A_\varphi$ with $\A_o$), and $\prob(\ind_\varphi=0,\Obs=o)$ (synchronizing $\A||\A_{\overline\varphi}$ with $\A_o$).

\begin{theorem}\label{thm:computing}
Let $\A$ be an FPA. If $Obs$ is a finite set, $\varphi$ is a regular set and for $o \in Obs$, $\Obs^{-1}(o)$ is a regular set, then for $\genericpo \in \{\lpo,\lpso,\hpo,\vpo\}$, $\genericpo(\A,\varphi,\Obs)$ can be computed.
\end{theorem}
The computation of opacity measures is done in polynomial time in the size of $Obs$ and DFAs $\A_\varphi$, $\A_{\overline\varphi}$, $\A_o$.

A prototype tool implementing this algorithm was developed in Java~\cite{eftenie10}, yielding numerical values for measures of opacity.

\section{Comparison of the measures of opacity}\label{sec:comparison}

In this section we compare the discriminating power of the measures
discussed above. As described above, the liberal measures evaluate the
leak, hence $0$ represents the best possible value from a security
point of view, producing an opaque system. For such an opaque system,
the restrictive measure evaluate the robustness of this opacity. As a
result, $1$ is the best possible value.

\subsection{Abstract examples}
The values of these metrics are first compared for extremal cases of
\figurename~\ref{fig:comprepartition}.  These values are displayed in
\tablename~\ref{tab:compvalues}.

\tikzstyle{abstractexample}=[scale=0.825]
\tikzstyle{abstractexample}=[scale=0.5]
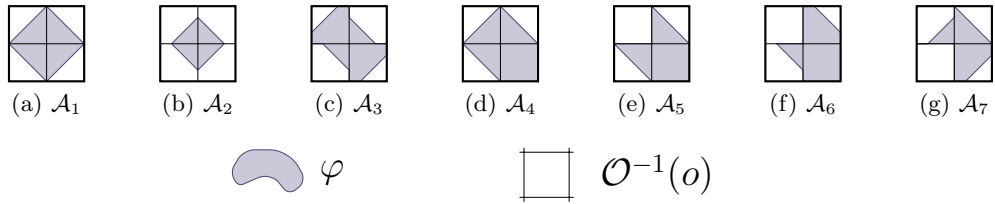
\begin{figure}[h]
\centering
\hfill{~}
\subfigure[$\A_{\arabic{subfigure}}$]{\label{fig:compperfect}
\begin{tikzpicture}[abstractexample]
\path[phee] (1,0) -- (0,1) -- (1,2) -- (2,1) -- cycle;
\draw[very thin] (0,0) grid (2,2);
\draw[thick] (0,0) -- (0,2) -- (2,2) -- (2,0) -- cycle;
\end{tikzpicture}
}
\hfill{~}
\subfigure[$\A_{\arabic{subfigure}}$]{\label{fig:comprpoperfect}
\begin{tikzpicture}[abstractexample]
\path[phee] (1,0.3) -- (0.3,1) -- (1,1.7) -- (1.7,1) -- cycle;
\draw[very thin] (0,0) grid (2,2);
\draw[thick] (0,0) -- (0,2) -- (2,2) -- (2,0) -- cycle;
\end{tikzpicture}
}
\hfill{~}
\subfigure[$\A_{\arabic{subfigure}}$]{\label{fig:compvpouseless}
\begin{tikzpicture}[abstractexample]
\path[phee] (1,0.3) -- (0.3,1) -- (0,1) -- (0,1.3) --(0.7,2) -- (1,2) -- (1,1.7) -- (1.7,1) -- (2,1) -- (2,0.7) -- (1.3,0) -- (1,0) -- cycle;
\draw[very thin] (0,0) grid (2,2);
\draw[thick] (0,0) -- (0,2) -- (2,2) -- (2,0) -- cycle;
\end{tikzpicture}
}
\hfill{~}
\subfigure[$\A_{\arabic{subfigure}}$]{\label{fig:compvperfect}
\begin{tikzpicture}[abstractexample]
\path[phee] (1,0) -- (0,1) --  (1,2)  -- (2,1) -- (2,0) -- cycle;
\draw[very thin] (0,0) grid (2,2);
\draw[thick] (0,0) -- (0,2) -- (2,2) -- (2,0) -- cycle;
\end{tikzpicture}
}
% Middle
\hfill{~}
\subfigure[$\A_{\arabic{subfigure}}$]{\label{fig:complpo}
\begin{tikzpicture}[abstractexample]
\path[phee] (1,0) -- (0,1) -- (1,1) --  (1,2)  -- (2,1) -- (2,0) -- cycle;
\draw[very thin] (0,0) grid (2,2);
\draw[thick] (0,0) -- (0,2) -- (2,2) -- (2,0) -- cycle;
\end{tikzpicture}
}
\hfill{~}
\subfigure[$\A_{\arabic{subfigure}}$]{\label{fig:comprpo}
\begin{tikzpicture}[abstractexample]
\path[phee] (1,0.3) -- (0.3,1) -- (1,1) --  (1,2)  -- (1.3,2) -- (2,1.3) -- (2,1) -- (2,0) -- (1,0) -- cycle;
\draw[very thin] (0,0) grid (2,2);
\draw[thick] (0,0) -- (0,2) -- (2,2) -- (2,0) -- cycle;
\end{tikzpicture}
}
\hfill{~}
\subfigure[$\A_{\arabic{subfigure}}$]{\label{fig:compweird}
\begin{tikzpicture}[abstractexample]
\path[phee] (1,0) -- (1,1) -- (0,1) -- (0.3,1) -- (1,1.7) -- (1,2) -- (2,1) -- (2,0.7) -- (1.3,0) -- cycle ;
\draw[very thin] (0,0) grid (2,2);
\draw[thick] (0,0) -- (0,2) -- (2,2) -- (2,0) -- cycle;
\end{tikzpicture}
}

\subfigure{\label{fig:complegende}
\begin{tikzpicture}[scale=3/5]
\node (pheestart) at (0,0) {};
\path[phee] (pheestart) [rounded corners=3pt]-- +(-0.3,0.5) [rounded corners=4pt]-- +(0,1) [rounded corners=7pt]-- +(1,1) [rounded corners=4pt]-- +(1.4,0.3) [rounded corners=3pt]-- +(1,0) [rounded corners=5pt]-- +(0.6,0.5) [rounded corners=2pt]-- cycle;
\node[anchor=west] at ($(pheestart) + (1.5,0.5)$) {\Large$\varphi$};
\end{tikzpicture}
\hspace{2cm}
\begin{tikzpicture}[scale=3/5]
\draw[very thin] (-0.1,-0.1) grid (1.1,1.1);
\node[anchor=west] at (1.5,0.5) {\Large${\cal O}^{-1}(o)$};
\end{tikzpicture}
}
\caption{Example of repartition of probabilities of $\ind_\varphi$ and
  $\cal O$ in 7 cases.}
\label{fig:comprepartition}
\end{figure}

\begin{table}
\centering
\newcounter{system}
\begin{tabular}{|c||c|c|c|c|c|c|c|}
\hline System &
\stepcounter{system}\subref{fig:compperfect} $\A_{\thesystem}$ &
\stepcounter{system}\subref{fig:comprpoperfect} $\A_{\thesystem}$ &
\stepcounter{system}\subref{fig:compvpouseless} $\A_{\thesystem}$ &
\stepcounter{system}\subref{fig:compvperfect} $\A_{\thesystem}$ &
\stepcounter{system}\subref{fig:complpo} $\A_{\thesystem}$ &
\stepcounter{system}\subref{fig:comprpo} $\A_{\thesystem}$ &
\stepcounter{system}\subref{fig:compweird} $\A_{\thesystem}$ \\ 
\hline \lponame & $0$ & $0$ & $0$ & $\frac14$ & $\frac14$ & $\frac14$ & $0$ \\ 
\hline \lpsoname & $0$ & $0$ & $0$ & $\frac14$ & $\frac12$ & $\frac12$ & $\frac14$ \\ 
\hline \hponame & $\frac12$ & $\frac34$ & $\frac38$ & $0$ & $0$ & $0$ & $\frac{12}{25}$ \\ 
\hline \vponame & $1$ & $\frac12$ & $\frac12$ & $0$ & $0$ & $0$ & $0$ \\ 
\hline 
\end{tabular} 
\caption{Values of the different opacity measures for systems of \figurename~\ref{fig:compperfect}-\subref{fig:compweird}.}
\label{tab:compvalues}
\end{table}

First, the system $\A_1$ of \figurename~\ref{fig:compperfect} is
intuitively very secure since, with or without observation, an
attacker has no information whether $\varphi$ was true or not.  This
security is reflected in all symmetrical measures, with highest scores
possibles in all cases.  It is nonetheless deemed more insecure for
\hponame, since opacity is perfect when $\varphi$ is always false.

The case of $\A_2$ of \figurename~\ref{fig:comprpoperfect} only
differs from $\A_1$ by the
global repartition of $\varphi$ in $Run(\A)$.  The information an
attacker gets comes not from the observation, but from $\varphi$
itself.  Therefore \vponame, which does not remove the information
available before observation, evaluates this system as less secure
than $\A_1$.  Measures based on information
theory~\cite{smith09,berard10} would consider this system as secure.
However, such measures lack strong ties with opacity, which depend
only on the information available to the observer, wherever this
information comes from. In
addition, \hponame finds $\A_2$ more secure than $\A_1$: $\varphi$ is
verified less often.  Note that the complement would not change the
value for symmetrical measures, while being insecure for \hponame
(with $\hpo=\frac14$).

However, since each observation class is considered individually,
\vponame does not discriminate $\A_2$ and $\A_3$ of
\figurename~\ref{fig:compvpouseless}.  Here, the information is the
same in each observation class as for $\A_2$, but the repartition of
$\varphi$ gives no advantage at all to an attacker without
observation.

When the system is not opaque (resp. symmetrically opaque), \hponame
(resp. \vponame) cannot discriminate them, and \lponame
(resp. \lpsoname) becomes relevant.  For example, $\A_4$ is not
opaque for the classical definitions, therefore $\hpo = \vpo = 0$ and
both $\lpo>0$ and $\lpso > 0$.

System, $\A_5$ of \figurename~\ref{fig:complpo} has a greater $\lpso$
than $\A_4$.  However, \lponame is unchanged since the class
completely out of $\varphi$ is not taken into account.  Remark that
system $\A_7$ is opaque but not symmetrically opaque, hence the
relevant measures are $\lpso$ and $\hpo$.  Also note that once a
system is not opaque, the repartition of classes that do not leak
information is not taken into account, hence equal values in the cases
of $\A_5$ and $\A_6$.

\subsection{A more concrete example}

Consider the following programs $P_1$ and $P_2$, inspired
from~\cite{smith09}, where $k$ is a given parameter, \texttt{random}
select uniformly an integer value (in binary) between its two
arguments and $\&$ is the bitwise \emph{and}:

\hfill{~}
\begin{tabular}{cl}
$P_1$:
& $H:=$ \texttt{random}$(0,2^{8k}-1)$;\\
& if $H \mod 8 = 0$ then \\
&\quad $L:=H$\\
& else\\
& \quad $L:=-1$\\
& fi
\end{tabular}
\hfill{~}\hfill{~}
\begin{tabular}{cl}
$P_2$: %\\
& $H:=$ \texttt{random}$(0,2^{8k}-1)$;\\
& $L:= H \ \&\  0^{7k}1^{k}$
\end{tabular}
\hfill{~}

In both cases, the value of $H$, an integer over $8k$ bits, is
 supposed to remain secret, and
cannot be observed directly, while the value of $L$ is public. Thus 
the observation is the ``$L:=\ldots$'' action. Intuitively, $P_1$
divulges the exact value of $H$ with probability $\frac18$.  On the
other hand, $P_2$ leaks the value of one eighth of its bits (the least
significant ones) at every execution.  

These programs can be translated into FPAs $\A_{P_1}$ and $\A_{P_2}$ of
\figurename~\ref{fig:exsmithfpa}.  In order to have a boolean
predicate, the secret is not the value of variable $H$, but whether
$H=L$: $\varphi_= = \left\{(H=x)(L=x) \mid x \in
  \{0,\dots,2^{8k}-1\}\right\}$. Non opacity then means that the
attacker discovers the secret value. Weaker predicates can also be
considered, like equality of $H$ with a particular value or $H$
belonging to a specified subset of values, but we chose the simplest
one.
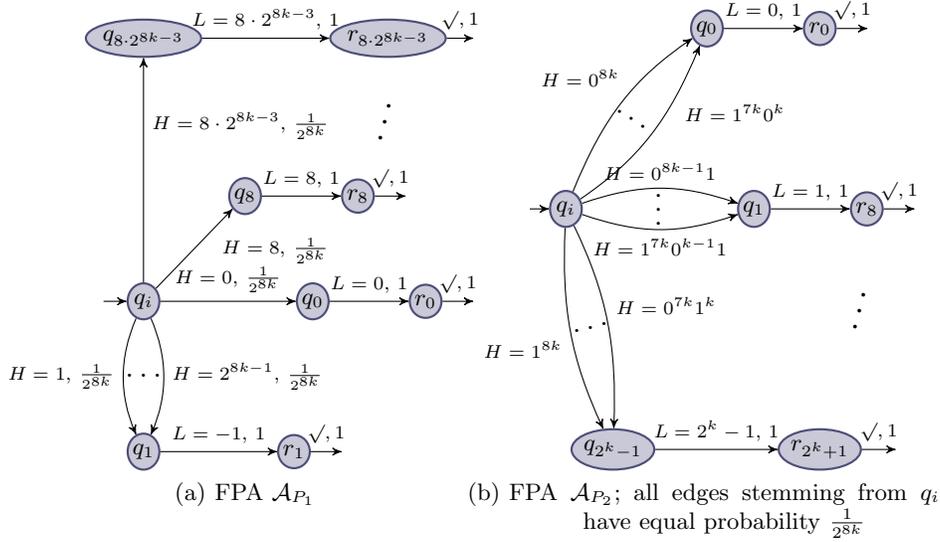
\begin{figure}
\centering
\subfigure[FPA $\A_{P_1}$]{\label{fig:fpasmith1}
\begin{tikzpicture}[auto,xscale=0.75]
\tikzstyle{every node}+=[font=\scriptsize]
\tikzstyle{every state}+=[ellipse,minimum width=12pt,minimum height=10pt,thick,inner xsep=0pt,inner ysep=2pt,font=\small]

\node[state,initial] (q0) at (0,0) {$q_i$};
\node[state] (q1) at (3,0) {$q_0$};
\node[state] (q3) at (0,3.5) {$q_{8 \cdot 2^{8k-3}}$};
\node[state] (q2) at (barycentric cs:q1=3,q3=2) {$q_8$};
\node[state] (q4) at (0,-2) {$q_1$};

\path[->] (q0) edge node {$H=0$, $\frac1{2^{8k}}$} (q1);
\path[->] (q0) edge node[swap,pos=0.75] {$H=8$, $\frac1{2^{8k}}$} (q2);
\path[->] (q0) edge node[pos=0.7,swap] {$H=8 \cdot 2^{8k-3}$, $\frac1{2^{8k}}$} (q3);
\path[->] (q0) edge [bend right] node [swap] (lab1) {$H=1$, $\frac1{2^{8k}}$} (q4);
\path[->] (q0) edge [bend left] node (lab2) {$H=2^{8k-1}$, $\frac1{2^{8k}}$} (q4);
\path (lab1) -- node[sloped,anchor=base] {\large$\ldots$} (lab2);

\tikzstyle{response}=[node distance=1.5cm]
\tikzstyle{ph}=[node distance=0.75cm]
\node[state] (r1) [response,right of=q1] {$r_0$};
\node (ph1) [ph,right of=r1] {};
\path[->] (q1) edge node {$L=0$, $1$} (r1);
\path[->] (r1) edge node {$\surd,1$} (ph1);

\node[state] (r2) [response,right of=q2] {$r_8$};
\node (ph2) [ph,right of=r2] {};
\path[->] (q2) edge node {$L=8$, $1$} (r2);
\path[->] (r2) edge node {$\surd,1$} (ph2);

\node[state] (r3) [node distance=3.25cm,right of=q3] {$r_{8 \cdot 2^{8k-3}}$};
\node (ph3) [node distance=1.25cm,right of=r3] {};
\path[->] (q3) edge node {$L=8 \cdot 2^{8k-3}$, $1$} (r3);
\path[->] (r3) edge node {$\surd,1$} (ph3);

\path (r3) -- node[sloped,pos=0.75] {\Large$\ldots$} (r2);

\node[state] (r4) [response,node distance=2cm,right of=q4] {$r_1$};
\node (ph4) [ph,right of=r4] {};
\path[->] (q4) edge node {$L=-1$, $1$} (r4);
\path[->] (r4) edge node {$\surd,1$} (ph4);
\end{tikzpicture}
}
\hspace{-0.625cm}
\subfigure[FPA $\A_{P_2}$; all edges stemming from $q_i$ have equal probability $\frac1{2^{8k}}$]{\label{fig:fpasmith2}
\begin{tikzpicture}[auto,xscale=0.625,yscale=0.8]
\tikzstyle{every node}+=[font=\scriptsize]
\tikzstyle{every state}+=[ellipse,minimum width=12pt,minimum height=10pt,thick,inner xsep=0pt,inner ysep=2pt,font=\small]

\node[state,initial] (q0) at (0,0) {$q_i$};
\node[state] (q1) at (3,3) {$q_0$};
\node[state] (q3) at (1,-4) {$q_{2^k-1}$};
\node[state] (q2) at (4,0) {$q_1$};

\path[->] (q0) edge [bend left=17.5] node[pos=0.55] (lab1a) {$H=0^{8k}$} (q1);
\path[->] (q0) edge [bend right=17.5] node[swap,pos=0.5] (lab1b) {} node[swap,pos=0.75] {$H=1^{7k}0^k$} (q1);
\path[->] (q0) edge [bend left=15] node (lab2a) {$H=0^{8k-1}1$} (q2);
\path[->] (q0) edge [bend right=15] node[swap] (lab2b) {$H=1^{7k}0^{k-1}1$} (q2);
\path[->] (q0) edge [bend left=17.5] node (lab3a) {$H=0^{7k}1^k$} (q3);
\path[->] (q0) edge [bend right=17.5] node[swap] (lab3b) {$H=1^{8k}$} (q3);

\path (lab1a) -- node[sloped,pos=0.6,anchor=base] {\large$\ldots$} (lab1b);
\path (lab2a) -- node[sloped,anchor=base] {\large$\ldots$} (lab2b);
\path (lab3a) -- node[sloped,anchor=base] {\large$\ldots$} (lab3b);

\tikzstyle{response}=[node distance=1.5cm]
\tikzstyle{ph}=[node distance=0.75cm]
\node[state] (r1) [response,right of=q1] {$r_0$};
\node (ph1) [ph,right of=r1] {};
\path[->] (q1) edge node {$L=0$, $1$} (r1);
\path[->] (r1) edge node {$\surd,1$} (ph1);

\node[state] (r2) [response,right of=q2] {$r_8$};
\node (ph2) [ph,right of=r2] {};
\path[->] (q2) edge node {$L=1$, $1$} (r2);
\path[->] (r2) edge node {$\surd,1$} (ph2);

\node[state] (r3) [node distance=2.75cm,right of=q3] {$r_{2^k+1}$};
\node (ph3) [node distance=1.125cm,right of=r3] {};
\path[->] (q3) edge node {$L=2^k-1$, $1$} (r3);
\path[->] (r3) edge node {$\surd,1$} (ph3);

\path (r2) -- node [sloped,pos=0.55] {\Large\dots} (r3);
\end{tikzpicture}
}
\caption{FPAs for programs $P_1$ and $P_2$.}
\label{fig:exsmithfpa}
\end{figure}
First remark that $\varphi_=$ is not opaque on $P_1$ in the classical
sense (both symmetrically or not).  Hence both \hponame and \vponame
are null.  On the other hand, $\varphi_=$ is opaque on $P_2$, hence
\lponame and \lpsoname are null.  The values for all measures are
gathered in \tablename~\ref{tab:smithexresults}.
\begin{table}
\centering
\begin{tabular}{|c|c|c|c|c|c|c|}
Program & $\lpo$ & $\lpso$ & $\hpo$ & $\vpo$ \\\hline\hline
$P_1$ & $\frac18$ & $1$ & $0$ & $0$ \\\hline
$P_2$ & $0$ & $0$ & $1-\frac1{2^{7k}}$ & $\frac1{7k}$ \\\hline
\end{tabular}
\caption{Opacity measures for programs $P_1$ and $P_2$.}
\label{tab:smithexresults}
\end{table}
Note that only restrictive opacity for $P_2$ depends on $k$.  This
comes from the fact that in all other cases, both $\varphi_=$ and the
equivalence classes scale at the same rate with $k$.  In the case of
$P_2$, adding length to the secret variable $H$ \emph{dilutes}
$\varphi_=$ inside each class.  Hence the greater $k$ is, the hardest
it is for an attacker to know that $\varphi_=$ is true, thus to crack
asymmetrical opacity.  Indeed, it will tend to get false in most
cases, thus providing an easy guess, and a low value for symmetrical
opacity.

\subsection{Crowds protocol}\label{sec:crowds}

The anonymity protocol known as \emph{crowds} was introduced
in~\cite{reiter98} and recently studied in the probabilistic framework
in~\cite{chatzikokolakis08,andres10}.  When a user wants to send a
message (or request) to a server without the latter knowing the origin
of the message, the user routes the message through a crowd of $n$
users.  To do so, it selects a user randomly in the crowd (including
himself), and sends him the message.  When a user receives a message
to be routed according to this protocol, it either sends the message
to the server with probability $1-q$ or forwards it to a user in the
crowd, with probability $q$.  The choice of a user in the crowd is
always equiprobable.  Under these assumptions, this protocol is known
to be secure, since no user is more likely than another to be the
actual initiator; indeed its \hponame is very low.  However, there can
be $c$ corrupt users in the crowd which divulge the identity of the
person that sent the message to them.  In that case, if a user sends
directly a message to a corrupt user, its identity is no longer
protected. The goal of corrupt users is therefore not to transmit
messages, hence they cannot initiate the protocol. The server and the
corrupt users cooperate to discover the identity of the
initiator. \hponame can measure the security of this system, depending
on $n$ and $c$.

First, consider our protocol as the system in
\figurename~\ref{fig:crowdsPPA}. In this automaton, states $1, \dots, n-c$
corresponding to honest users are duplicated in order to differentiate
their behavior as initiator or as the receiver of a message from the
crowd. The predicate we want to
be opaque is $\varphi_i$ that contains all the runs in which $i$ is
the initiator of the request.  The observation function ${\cal O}$
returns the penultimate state of the run, \textit{i.e.} the honest
user that will be seen by the server or a corrupt user.

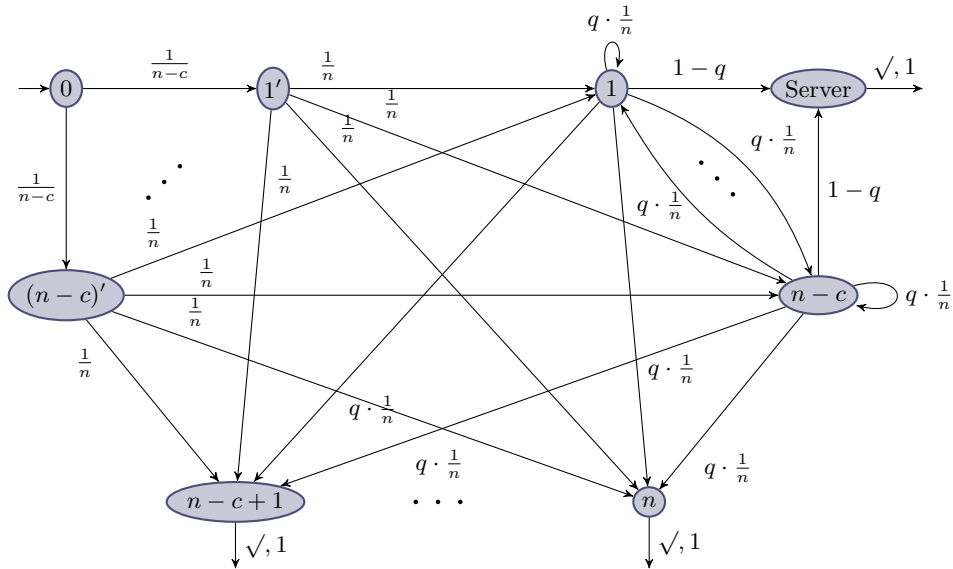
\begin{figure}
\centering
\begin{tikzpicture}[node distance=2.75cm,auto,initial text=]
\tikzstyle{every state}+=[ellipse,minimum width=12pt,minimum height=10pt,thick,inner xsep=0pt,inner ysep=2pt]
\tikzstyle{every node}+=[font=\small]

\node[state,initial] (init) at (0,0) {$0$};
\node[state] (q1) [right of=init] {$1'$};
\node[state] (q2) [below of=init] {$(n-c)'$};
\path (q1) -- (q2) node[anchor=south,midway,sloped] {\huge\dots};
\path[->] (init) edge node {$\frac{1}{n-c}$} (q1);
\path[->] (init) edge node [swap] {$\frac{1}{n-c}$} (q2);

\node[state] (server) at (10,0) {Server};
\node[state] (r1) [left of=server] {$1$};
\node[state] (r2) [below of=server] {$n-c$};
\node (ph) [node distance=1.5cm,right of=server] {};
\path[->] (server) edge node {$\surd,1$} (ph);
\path (r1) -- (r2) node[anchor=south,midway,sloped] {\huge\dots};
\path[->] (r1) edge node {$1-q$} (server);
\path[->] (r2) edge node [swap] {$1-q$} (server);
\path[->] (r1) edge [out=-20,in=110] node {$q \cdot \frac1n$} (r2);
\path[->] (r1) edge [loop above] node {$q \cdot \frac1n$} (r1);
\path[->] (r2) edge [out=150,in=-60] node [anchor=east] {$q \cdot \frac1n$} (r1);
\path[->] (r2) edge [loop right] node {$q \cdot \frac1n$} (r2);

\node (botdots) at ($(barycentric cs:init=1,server=1) - (0,5.5)$) {\huge\dots};
\node[state] (p1) [left of=botdots] {$n-c+1$};
\node[state] (p2) [right of=botdots] {$n$};
\node (ph1) [node distance=1cm,below of=p1] {};
\node (ph2) [node distance=1cm,below of=p2] {};
\path[->] (p1) edge node {$\surd,1$} (ph1);
\path[->] (p2) edge node {$\surd,1$} (ph2);

\path[->] (q1) edge node [pos=0.125] {$\frac1n$} (r1);
\path[->] (q1) edge node [pos=0.17] {$\frac1n$} (r2);
\path[->] (q1) edge node [pos=0.125] {$\frac1n$} (p1);
\path[->] (q1) edge node [pos=0.125] {$\frac1n$} (p2);
\path[->] (q2) edge node [pos=0.125] {$\frac1n$} (r1);
\path[->] (q2) edge node [pos=0.125] {$\frac1n$} (r2);
\path[->] (q2) edge node [pos=0.125,swap] {$\frac1n$} (p1);
\path[->] (q2) edge node [pos=0.125] {$\frac1n$} (p2);
\path[->] (r1) edge node [pos=0.75] {$q \cdot \frac1n$} (p1);
\path[->] (r1) edge node [pos=0.75] {$q \cdot \frac1n$} (p2);
\path[->] (r2) edge node [pos=0.75] {$q \cdot \frac1n$} (p1);
\path[->] (r2) edge node [pos=0.75] {$q \cdot \frac1n$} (p2);
\end{tikzpicture}
\caption{FPA $\mathcal{C}_n^c$ for Crowds protocol with $n$ users, among whom $c$ are corrupted.}
\label{fig:crowdsPPA}
\end{figure}

For sake of brevity, we write `$i \rightsquigarrow$' to denote the
event ``a request was initiated by $i$'' and `$\rightsquigarrow j$'
when ``$j$ was detected by the
adversary'', which means that $j$ sent the message either to a corrupt
user or to the server, who both try to discover who the initiator was.
The abbreviation $i \rightsquigarrow j$ stands for $i \rightsquigarrow
\wedge \rightsquigarrow j$.  Notation `$\neg i \rightsquigarrow$'
means that ``a request was initiated by someone else than $i$'';
similarly, combinations of this notations are used in the sequel. We
also use the Kronecker symbol $\delta_{ij}$ defined by $\delta_{ij} =
1$ if $i = j$ and $0$ otherwise.

\subsubsection*{Computation of the probabilities.}
All probabilities $\prob(i\rightsquigarrow j)$ can be
 automatically computed using the method described in
Section~\ref{sec:computing}.  For example, $\prob(1
\rightsquigarrow (n-c))$, the probability for the first user to
initiate the protocol while the last honest user is detected, can be
computed from substochastic automaton $\mathcal{C}_n^c || \mathcal{A}_{1
  \rightsquigarrow (n-c)}$ depicted on
\figurename~\ref{fig:crowdsRestrictedSa}.
In this automaton, the only duplicated state remaining is $1'$.

This SA can also be represented by a transition matrix (like a Markov chain), which is given in \tablename~\ref{tab:crowdsRestrictedSAmatrix}.
An additional column indicates the probability for the $\surd$ action, which ends the run (here it is either $1$ if the state is final and $0$ if not).

The associated system is
represented in \tablename~\ref{tab:crowdsSystem} where $L_S$ corresponds to the ``Server'' state. 
Each line of the system is given by the outgoing probabilities of the corresponding state in the SA, or alternatively by the corresponding line of the matrix.
Resolving it (see Appendix~\ref{app:systemcrowds})
yields, $L_i = \frac{q}{n}$ for all $i \in \{1, \dots, n-c-1\}$,
$L_{n-c} = 1 - \frac{q \cdot (n-c-1)}{n}$, $L_{1'} = \frac1n$, and
$L_0 = \frac1{(n-c) \cdot n}$.  Therefore, $\prob(1 \rightsquigarrow
(n-c)) = \frac1{(n-c) \cdot n}$.

\begin{figure}
\centering
\begin{tikzpicture}[node distance=3cm,auto,initial text=]
\tikzstyle{every state}+=[ellipse,minimum width=12pt,minimum height=10pt,thick,inner xsep=0pt,inner ysep=2pt]
\tikzstyle{every node}+=[font=\small]

\node[state,initial] (init) at (0,-3) {$0$};
\node[state] (q1) [node distance=1.5cm,right of=init] {$1'$};
\path[->] (init) edge node {$\frac{1}{n-c}$} (q1);

\node[state] (server) at (8,0) {Server};
\node[state] (r1) at (server -| q1) {$1$};
\node[state] (r2) [below of=server] {$n-c$};
\node (ph) [node distance=1.25cm,right of=server] {};
\path[->] (server) edge node {$\surd,1$} (ph);
\path (r1) -- (r2) 
    node[state,pos=0.33] (r1bis) {$2$}
    node[anchor=south,pos=0.7,sloped] {\huge\dots};
\path[->] (r2) edge node [swap] {$1-q$} (server);
\path[->] (r1) edge [bend left=20] node {$q \cdot \frac1n$} (r2);
\path[->] (r1) edge [loop left] node {$q \cdot \frac1n$} (r1);
\path[->] (r2) edge [bend left=20] node {$q \cdot \frac1n$} (r1);
\path[->] (r2) edge [loop right] node {$q \cdot \frac1n$} (r2);
\path[->] (r1) edge [bend left=10] node {$q \cdot \frac1n$} (r1bis);
\path[->] (r1bis) edge [bend left=10] node[pos=0.125] {$q \cdot \frac1n$} (r1);
\path[->] (r1bis) edge [loop, out=-120,in=-150,min distance=0.5cm] node[pos=0.25] {$q \cdot \frac1n$} (r1bis);
\path[->] (r1bis) edge [bend left=15] node {$q \cdot \frac1n$} (r2);
\path[->] (r2) edge [bend left=15] node[anchor=east,pos=0.6] {$q \cdot \frac1n$} (r1bis);

\node (botdots) [node distance=1cm,below of=r2] {\huge\dots};
\node[state] (p1) [node distance=1.875cm,left of=botdots] {$n-c+1$};
\node[state] (p2) [node distance=1.5cm,right of=botdots] {$n$};
\node (ph1) [node distance=0.75cm,below of=p1] {};
\node (ph2) [node distance=0.75cm,below of=p2] {};
\path[->] (p1) edge node {$\surd,1$} (ph1);
\path[->] (p2) edge node {$\surd,1$} (ph2);

\path[->] (q1) edge node [pos=0.5] {$\frac1n$} (r1);
\path[->] (q1) edge [bend right] node {$\frac1n$} (r1bis);
\path[->] (q1) edge node [swap,pos=0.5] {$\frac1n$} (r2);
\path[->] (r2) edge node [swap,pos=0.75] {$q \cdot \frac1n$} (p1);
\path[->] (r2) edge node [pos=0.75] {$q \cdot \frac1n$} (p2);
\end{tikzpicture}
\caption{SA $\mathcal{C}_n^c || \mathcal{A}_{1 \rightsquigarrow (n-c)}$ corresponding to runs where user $1$ initiates the protocol and user $(n-c)$ is detected.}
\label{fig:crowdsRestrictedSa}
\end{figure}
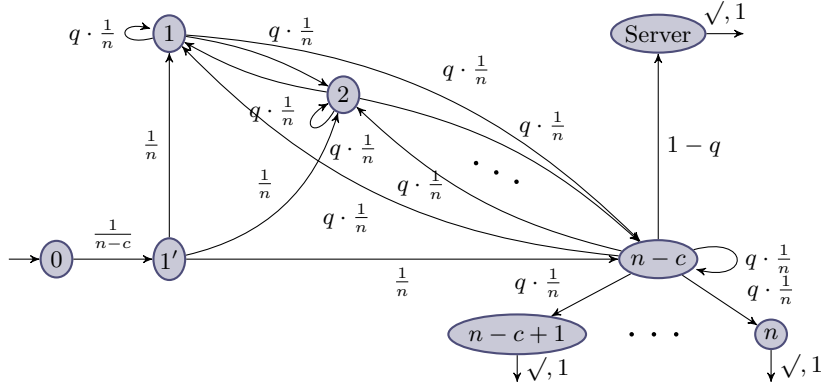

\begin{table}
\[
\begin{array}{c|*{10}c}
 & 0 & 1' & 1& \cdots & n-c & n-c+1 & \cdots & n & \textit{Server} & \surd \\ \hline
0 & 0 & \frac{1}{n-c} & 0 & \cdots & 0 & 0 & \cdots & 0 & 0 & 0 \\
1' & 0 & 0 & \frac1n & \cdots & \frac1n & 0 & \cdots & 0 & 0 & 0 \\
1 & 0 & 0 & q\cdot\frac1n & \cdots & q\cdot\frac1n & 0 & \cdots & 0 & 0 & 0 \\
\vdots & \vdots & \vdots & \vdots & \ddots & \vdots & \vdots & \ddots & \vdots & \vdots & \vdots \\
n-c-1 & 0 & 0 & q\cdot\frac1n & \cdots & q\cdot\frac1n & 0 & \cdots & 0 & 0  & 0\\
n-c & 0 & 0 & q\cdot\frac1n & \cdots & q\cdot\frac1n &  q\cdot\frac1n & \cdots &  q\cdot\frac1n & 1-q & 0 \\
n-c+1 & 0 & 0 & 0 & \cdots & 0 & 0 & \cdots & 0 & 0 & 1 \\
\vdots & \vdots & \vdots & \vdots & \ddots & \vdots & \vdots & \ddots & \vdots & \vdots & \vdots \\
n & 0 & 0 & 0 & \cdots & 0 & 0 & \cdots & 0 & 0 & 1 \\
\textit{Server} & 0 & 0 & 0 & \cdots & 0 & 0 & \cdots & 0 & 0 & 1 \\
\end{array}
\]
\caption{The matrix giving the transition probabilities between the states of $\mathcal{C}_n^c || \mathcal{A}_{1 \rightsquigarrow (n-c)}$.}
\label{tab:crowdsRestrictedSAmatrix}
\end{table}

\begin{table}
\hfill{~}
\parbox{4.5cm}{
\[\left\{
\begin{array}{rcl}
L_0 &=& \frac1{n-c} \cdot L_{1'} \\
L_{1'} &=& \sum_{i = 1}^{n-c} \frac1n \cdot L_i \\
L_1 &=& \sum_{i = 1}^{n-c} \frac{q}{n} \cdot L_i \\
&\vdots& \\
L_{n-c-1} &=& \sum_{i = 1}^{n-c} \frac{q}{n} \cdot L_i
\end{array}
\right.
\]}
\hfill{~}\hfill{~}
\parbox{6.5cm}{
\[\left\{
\begin{array}{rcl}
L_{n-c} &=& (1-q) \cdot L_{S} + \sum_{i = 1}^n \frac{q}{n} \cdot L_i \\
L_{n-c+1} &=& 1\\
&\vdots& \\
L_{n} &=& 1\\
L_{S} &=& 1
\end{array}
\right.
\]}
\hfill{~}
\caption{Linear system associated to SA $\mathcal{C}_n^c || \mathcal{A}_{1 \rightsquigarrow (n-c)}$ of \figurename~\ref{fig:crowdsRestrictedSa}.}
\label{tab:crowdsSystem}
\end{table}

In this case, simple reasoning on the symmetries of the model allows
to derive other probabilities $\prob(i \rightsquigarrow j)$.  Remark
that the probability for a message to go directly from initiator to
the adversay (who cannot be the server) is $\frac{c}{n}$: it only
happens if a corrupt user is
chosen by the initiator.  If a honest user is chosen by the initiator,
then the length of the path will be greater, with probability
$\frac{n-c}{n}$.  By symmetry
all honest users have equal probability to be the initiator, and equal
probability to be detected.  Hence \( \prob(i \rightsquigarrow) =
\prob(\rightsquigarrow j) = \frac{1}{n-c} \).

Event $i \rightsquigarrow j$ occurs when $i$ is chosen as the
initiator (probability $\frac{1}{n-c}$), and either (1) if $i =j$ and
$i$ chooses a corrupted user to route its message, or (2) if a honest
user is chosen and $j$ sends the message to a corrupted user or the
server (the internal route between honest users before $j$ is
irrelevant).  Therefore
\begin{eqnarray*}
\prob(i \rightsquigarrow j) &=& \frac{1}{n-c} \cdot \left(\delta_{ij} \cdot \frac{c}{n} + \frac{1}{n-c} \cdot \frac{n-c}{n}\right) \\
\prob(i \rightsquigarrow j) &=& \frac{1}{n-c} \cdot \left(\delta_{ij} \cdot \frac{c}{n} + \frac{1}{n}\right)
\end{eqnarray*}
The case when $i$ is not the initiator is derived from this probability:
\begin{eqnarray*}
\prob(\neg i \rightsquigarrow j) &=& \sum_{\substack{k=1\\k \neq i}}^{n-c} {\prob(k \rightsquigarrow j)} \\
\prob(\neg i \rightsquigarrow j) &=& \frac{1}{n-c} \cdot \left((1-\delta_{ij}) \cdot \frac{c}{n} + \frac{n-c-1}{n}\right)
\end{eqnarray*}
Conditional probabilities thus follow:
\[
\prob(i \rightsquigarrow | \rightsquigarrow j) = \frac{\prob(i \rightsquigarrow j)}{\prob(\rightsquigarrow j)} = \delta_{ij} \cdot \frac{c}{n} + \frac{1}{n}
\]
\[
\prob(\neg i \rightsquigarrow | \rightsquigarrow j) = \frac{\prob(\neg i \rightsquigarrow j)}{\prob(\rightsquigarrow j)} = (1-\delta_{ij}) \cdot \frac{c}{n} + \frac{n-c-1}{n}
\]
Interestingly, these probabilities do not depend on $q$\footnote{This stems from the fact that the original models had either the server or the corrupt users as attackers, not both at the same time.}.

\subsubsection*{Computation of \hponame.}\label{par:crowdshpo}

From the probabilities above, one can compute an analytical value for $\hpo(\mathcal{C}_n^c,\ind_{\varphi_i},\Obs)$.
\begin{eqnarray*}
\frac1{\hpo(\mathcal{C}_n^c,\varphi_i,\Obs)} &=& \sum_{j=1}^{n-c} \prob(\rightsquigarrow j) \frac1{\prob(\neg i \rightsquigarrow | \rightsquigarrow j)}\\
&=& (n-c-1) \cdot \frac1{n-c} \cdot \frac{n}{n-1} + \frac1{n-c} \cdot \frac{n}{n-c-1}\\
&=& \frac{n}{n-c} \left(\frac{n-c-1}{n-1} + \frac1{n-c-1}\right)\\
\frac1{\hpo(\mathcal{C}_n^c,\varphi_i,\Obs)} &=& \frac{n \cdot (n^2 + c^2 - 2nc -n + 2c)}{(n-c) \cdot (n-1) \cdot (n-c-1)}
\end{eqnarray*}
Hence
\[\hpo(\mathcal{C}_n^c,\varphi_i,\Obs) = \frac{(n-c) \cdot (n-1) \cdot (n-c-1)}{n \cdot (n^2 + c^2 - 2nc -n + 2c)}\]
which tends to $1$ as $n$ increases to $+\infty$ (for a fixed number
of corrupted users).  The evolution of \hponame is represented in
\figurename~\ref{fig:crowdshpograph} where blue means low and red
means high.  If the proportion of corrupted users is fixed, say $n =
4c$, we obtain
\[\hpo(\mathcal{C}_{4c}^c,\varphi_i,\Obs) = \frac{(4c-1) \cdot (9c-3)}{4c \cdot (9c-2)}\]
which also tends to $1$ as the crowds size increases.  When there are
no corrupted users,
\[\hpo(\mathcal{C}_n^0,\varphi_i,\Obs) = \frac{n-1}n,\] which is close to $1$, but never exactly, since $\varphi_i$ is not always false, although of decreasing proportion as the crowds grows.
This result has to be put in parallel with the one
from~\cite{reiter98}, which states that crowds is secure since each
user is ``beyond suspicion'' of being the initiator, but ``absolute
privacy'' is not achieved.
\begin{figure}
\centering
\subfigure[Evolution of $\hpo(\mathcal{C}_n^c,\varphi_i,\Obs)$ with $n$  and $c$.][Evolution of $\hpo(\mathcal{C}_n^c,\varphi_i,\Obs)$ with $n$  and $c$. Red meaning a value close to $1$ and blue meaning close to $0$.]{
\label{fig:crowdshpograph}
\begin{tikzpicture}[xscale=0.25,yscale=0.1175]
\useasboundingbox (-2,-2.5) rectangle (21.75,22);
\node[anchor=south west,inner sep=0pt] at (1,0) {\pgfimage[width=4.75cm]{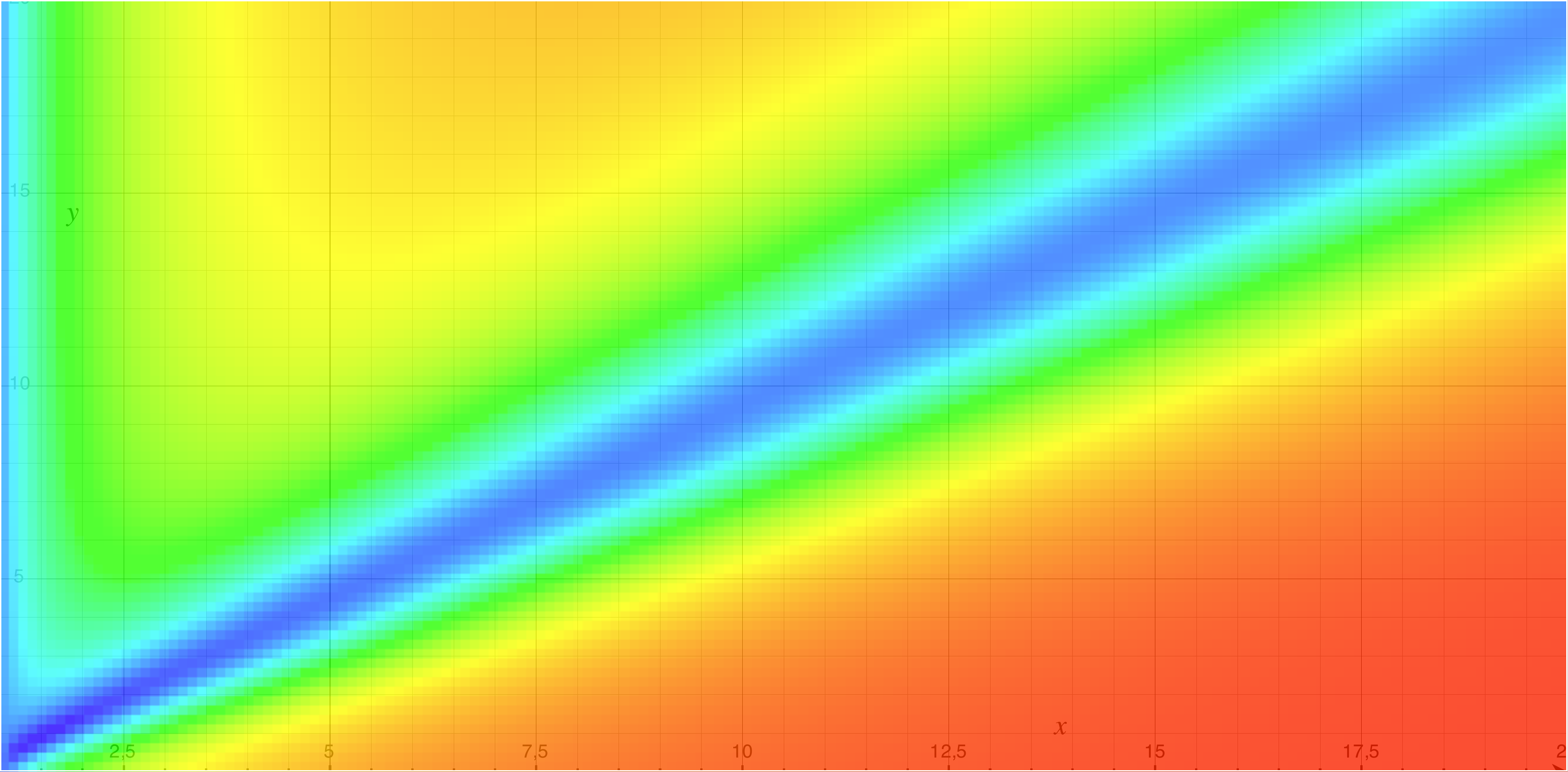}};
\fill[white] (0,0) -- (20,20) -- (0,20) -- cycle;
\path[->,draw] (0,0) -- (21,0) node[anchor=west] {$n$};
\path[->,draw] (0,0) -- (0,22) node[anchor=south] {$c$};
\foreach \x in {0,2.5,...,20} {%
  \draw (\x,-0.25) -- (\x,0.25);
  \node[anchor=north] at (\x,0) {\x};};
\foreach \y in {0,5,...,20} {%
  \draw (-0.125,\y) -- (0.125,\y);
  \node[anchor=east] at (0,\y) {\y};};
\end{tikzpicture}
}
\hfill{~}
\subfigure[Evolution of $\vpo(\mathcal{C}_n^c,\varphi_i,\Obs)$ with $n$ when \mbox{$c=5$}.]{
\label{fig:crowdsvpograph}
% pdflatex --jobname=mscs2012-arxiv-graph-CrowdsVpo mscs2012-arxiv.tex
\beginpgfgraphicnamed{mscs2012-arxiv-graph-CrowdsVpo}
\begin{tikzpicture}[xscale=0.25,yscale=3]
%\useasboundingbox (-1,-0.1) rectangle (22.5,1.15);
\path[->,draw] (0,0) -- (21,0) node[anchor=west] {$n$};
\path[->,draw] (0,0) -- (0,1.05) node[anchor=south] {$\vpo(\mathcal{C}_n^5,\varphi_i,\Obs)$};
\foreach \x in {0,2,...,20} {%
  \draw (\x,-0.01) -- (\x,0.01);
  \node[anchor=north] at (\x,0) {\x};};
\foreach \y in {0,0.25,0.5,0.75,1} {%
  \draw (-0.125,\y) -- (0.125,\y);
  \node[anchor=east] at (0,\y) {\y};};
  \draw[very thick,color=auxdraw] plot[domain=6.0000000000000001:12,id=CrowdsVpoInf,samples=1000] function{(x-5)/((log(x)*(x-5)/log(2))-(log(x-6)/log(2)))};
  \draw[very thick,color=auxdraw] plot[domain=12:20,id=CrowdsVpoSup,samples=1000] function{(x-5)/((log(x)*(x-5)/log(2))-(log(6)/log(2)))};
\end{tikzpicture}
\endpgfgraphicnamed
}
\caption{Evolution of restrictive opacity with the size of the crowd.}
\end{figure}
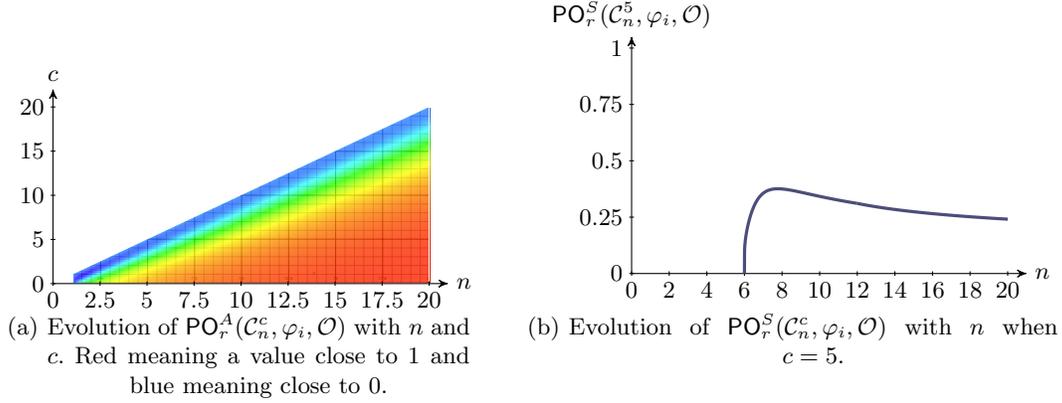

\subsubsection*{Computation of \vponame.}\label{subsubsec:crowdsvpo}

From the probabilities computed above, we obtain that if $i \neq j$,
\[V(i \rightsquigarrow | \rightsquigarrow j) = \max\left(\frac1n,\frac{n-1}n\right) = \frac1n \max(1,n-1).\]
Except in the case of $n = 1$ (when the system is non-opaque, hence $\vpo(\mathcal{C}^0_1,\varphi_1,\mathcal{O}) = 0$), $V(i \rightsquigarrow | \rightsquigarrow j) = \frac{n-1}n$.

In the case when $i=j$
\[V(i \rightsquigarrow | \rightsquigarrow i) =
\max\left(\frac{c+1}n,\frac{n-c-1}n\right) = \frac1n
\max(c+1,n-c-1).\] That means the vulnerability for the observation
class corresponding to the case when $i$ is actually detected depends
on the proportion of corrupted users in the crowd.  Indeed, $V(i
\rightsquigarrow | \rightsquigarrow i) = \frac{c+1}n$ if and only if
$n \leq 2(c+1)$.  The two cases shall be separated.

\begin{description}
\item[When $n \leq 2(c+1)$.] The message is initially more likely to be sent to a corrupt user or to the initiator himself than to any other user in the crowd:
\begin{eqnarray*}
\sum_{j=1}^{n-c} \prob(\rightsquigarrow j) \cdot \log(1-V(i \rightsquigarrow | \rightsquigarrow j)) &=&
\frac{(n-c-1) \cdot \log\left(\frac1n\right) + \log\left(\frac{n-c-1}n\right)}{n-c}\\
&=&\frac1{n-c} \cdot \left(\log(n-c-1) - (n-c) \cdot \log(n)\right) \\
&=&\frac{\log(n-c-1)}{n-c} - \log(n)
\end{eqnarray*}
\[\mbox{Hence}\quad \vpo(\mathcal{C}^c_n,\varphi_i,\mathcal{O}) \ =\  \frac1{\log(n) - \frac{\log(n-c-1)}{n-c}}\]
\item[When $n > 2(c+1)$.] The message is initially more likely to be sent to a honest user different from the initiator:
\begin{eqnarray*}
\sum_{j=1}^{n-c} \prob(\rightsquigarrow j) \cdot \log(1-V(i \rightsquigarrow | \rightsquigarrow j)) &=&
\frac{(n-c-1) \cdot \log\left(\frac1n\right) + \log\left(\frac{c+1}n\right)}{n-c}\\
&=& \frac1{n-c} \cdot (\log(c+1) - (n-c) \cdot \log(n))\\
&=& \frac{\log(c+1)}{n-c} - \log(n)
\end{eqnarray*}
\[\mbox{Hence}\quad \vpo(\mathcal{C}^c_n,\varphi_i,\mathcal{O}) \ =\  \frac1{\log(n) - \frac{\log(c+1)}{n-c}}\]
\end{description}
The evolution of \vponame for $c=5$ is depicted in \figurename~\ref{fig:crowdsvpograph}.

One can see that actually the \vponame decreases when $n$ increases.
That is because when there are more users in the crowd, user $i$ is
less likely to be the initiator.  Hence the predicate chosen does not
model anonymity as specified in~\cite{reiter98} but a stronger
property since \vponame is based on the definition of symmetrical
opacity.  Therefore it is meaningful in terms of security properties
only when both the predicate and its negation are meaningful.

\section{Dealing with nondeterminism}\label{sec:schedulers}

The measures presented above were all defined in the case of fully
probabilistic finite automata.  However, some systems present
nondeterminism that cannot reasonably be abstracted away.  For
example, consider the case of a system, in which a malicious user
Alice can control certain actions. The goal of Alice is to establish a
covert communication channel with an external observer Bob.  Hence she
will try to influence the system in order to render communication
easier.  Therefore, the actual security of the system as observed by
Bob should be measured against the best possible actions for Alice.
Formally, Alice is a scheduler who, when facing several possible
output distributions $\{\mu_1,\dots,\mu_n\}$, can choose whichever
distribution $\nu$ on $\{1,\dots,n\}$ as weights for the $\mu_i$s.
The security as measured by opacity is the minimal security of all
possible successive choices.

\subsection{The nondeterministic framework}

Here we enlarge the setting of probabilistic automata considered
before with nondeterminism.  There are several outgoing distribution
from a given state instead of a single one.

\begin{definition}[Nondeterministic probabilistic automaton]
  A \emph{nondeterministic probabilistic automaton} (NPA) is a tuple
  $\langle \Sigma, Q, \Delta, q_0 \rangle$ where
\begin{itemize}
  \item $\Sigma$ is a finite set of actions;
  \item $Q$ is a finite set of states;
  \item $\Delta: Q \rightarrow \powerset(\probset((\Sigma \times Q)
    \uplus \{\surd\}))$ is a nondeterministic probabilistic transition
    function;
  \item $q_0$ is the initial state;
\end{itemize}
where $\powerset(A)$ denotes the set of finite subsets of $A$.
\end{definition}

The choice over the several possible distributions is made by the
\emph{scheduler}.  It does not, however, selects one distribution to
be used, but can give weight to the possible distributions.

\begin{definition}[Scheduler]\label{def:scheduler}
  A scheduler on $\A = \langle \Sigma, Q, \Delta, q_0 \rangle$ is a
  function
\[\sigma: Run(\A) \rightarrow \probset(\probset((\Sigma \times Q) \uplus \{\surd\}))\]
such that $\sigma(\rho)(\nu) > 0 \Rightarrow \nu \in
\Delta(\lst(\rho))$.

The set of all schedulers for $\A$ is denoted $\Adv_{\A}$ (the dependence on $\A$ will be omitted if clear from the context).
\end{definition}

Observe that the choice made by a scheduler can depend on the
(arbitrarily long) history of the execution.  A scheduler $\sigma$ is
\emph{memoryless} if there exists a function $\sigma': Q \rightarrow
\probset(\probset((\Sigma \times Q) \uplus \{\surd\}))$ such that
$\sigma(\rho) = \sigma'(\lst(\rho))$.  Hence a memoryless scheduler
takes only into account the current state.

\begin{definition}[Scheduled NPA]
NPA $\A = \langle \Sigma, Q, \Delta, q_0 \rangle$ scheduled by $\sigma$ is the (infinite) FPFA $\A_{/\sigma} = \langle \Sigma, Run(\A), \Delta',\varepsilon\rangle$ where
\[\Delta'(\rho)(a,\rho') = \sum_{\mu \in \Delta(q)} \sigma(\rho)(\mu) \cdot \mu(a,q') \quad\mbox{if}\ \rho' = \overbrace{q_0 \rightarrow \dots \rightarrow q}^{\rho} \,\trans{a}\, q'\]
and $\Delta'(\rho)(a,\rho') = 0$ otherwise.
\end{definition}
A scheduled NPA behaves as an FPFA, where the outgoing distribution is
the set of all possible distributions weighted by the scheduler.

All measures defined in this paper on fully probabilistic finite
automata can be extended to non-deterministic probabilistic automata.
First note that all measures can be defined on infinite systems,
although they cannot in general be computed automatically, even with
proper restrictions on predicate and observables.  From the security
point of view, opacity in the case of an NPA should be the measure for
the FPFA obtained with the worst possible scheduler.
Hence the leak evaluated by the liberal
measures (\lponame and \lpsoname) is the greatest possible, and the
robustness evaluated by the restrictive measures (\hponame and
\vponame) is the weakest possible.
\begin{definition}
Let $\A$ be an NPA, $\varphi$ a predicate, and $\Obs$ an observation function.
\[\mbox{For } \genericpo \in \{\lpo,\lpso\}, \quad \widehat{\genericpo}(\A,\varphi,\Obs) = \max_{\sigma \in \Adv} \genericpo(\A_{/\sigma},\varphi,\Obs).\]
\[\mbox{For } \genericpo \in \{\hpo,\vpo\}, \quad \widehat{\genericpo}(\A,\varphi,\Obs) = \min_{\sigma \in \Adv} \genericpo(\A_{/\sigma},\varphi,\Obs).\]
\end{definition}

\subsection{The expressive power of schedulers}\label{subsec:memorylessnotenough}
In the context of analysis of security systems running in a hostile
environment, it is quite natural to consider the scheduler to be under
control of the adversary.  However if not constrained this gives the
adversary an unreasonably strong power even for obviously secure
systems as it can reveal certain secrets.  Also several classes of
schedulers have been proposed in order to avoid considering
unrealistic power of unconstrained schedulers and the ability of these
classes to reach supremum probabilities~\cite{giro09}. We now
investigate this problem for quantitative opacity.

First we show that memoryless schedulers are not sufficiently expressive,
with the following counterexample.
\begin{theorem}\label{thm:ctre-ex}
  There exists an NPA $\mathcal{B}$ such that the value
  $\widehat{\hpo}(\mathcal{B},\varphi,\Obs)$ cannot be reached by a
  memoryless scheduler.
\end{theorem}

\begin{proof}
  Consider the NPA $\mathcal{B}$ of
  \figurename~\ref{fig:cexmemsystem}.  Transitions on $a$ and $b$
  going to state $q_1$ (along with the westbound $\surd$) are part of
  the same probabilistic transition indicated by the arc linking the
  outgoing edges (and similarly eastbound).  Let $\varphi$ be the
  (regular) predicate consisting of runs whose trace projected onto
  $\{a,b\}$ is in $(ab)^+ + (ab)^*a$ (so $a$ and $b$ must alternate).
  Let $\Obs$ be the observation function that keeps the last $o_i$ of
  the run.  Hence there are only three observables, $\varepsilon$,
  $o_1$, and $o_2$.  Intuitively, a scheduler can introduce a bias in
  the next letter read from state $q_0$.

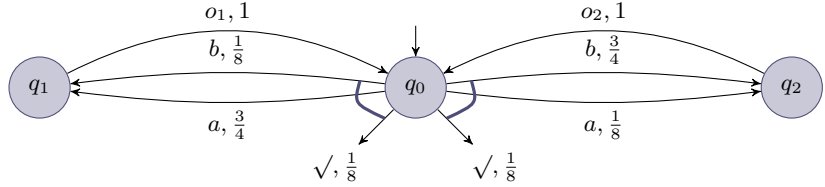
\begin{figure}
\centering
\begin{tikzpicture}[auto,node distance=5cm]
\tikzstyle{grouptrans}=[draw,very thick,auxdraw]
\node[state,initial above] (q0) at (0,0) {$q_0$};
\node[state,left of=q0] (q1) {$q_1$};
\node[state,right of=q0] (q2) {$q_2$};
\node[node distance=1.5cm,below left of=q0] (ph1) {$\surd,\frac18$};
\node[node distance=1.5cm,below right of=q0] (ph2) {$\surd,\frac18$};

\path[->] (q0) edge [bend left=7] node {$a,\frac34$} (q1);
\path[->] (q0) edge [bend right=7] node[pos=0.07,anchor=south] (a1) {} node [swap] {$b,\frac18$} (q1);
\path[->] (q0) edge node[pos=0.25,anchor=south] (s1) {} (ph1);

\node (control) at ($(barycentric cs:a1=1,s1=1) + (-0.25,-0.17)$) {};
\path[grouptrans] (a1.south) .. controls (control) .. (s1.south);

\path[->] (q1) edge [bend left=27] node {$o_1,1$} (q0);

\path[->] (q0) edge [bend right=7] node [swap] {$a,\frac18$} (q2);
\path[->] (q0) edge [bend left=7] node[pos=0.07,anchor=south] (a2) {} node {$b,\frac34$} (q2);
\path[->] (q0) edge node[pos=0.25,anchor=south] (s2) {} (ph2);

\node (control) at ($(barycentric cs:a2=1,s2=1) + (0.25,-0.17)$) {};
\path[grouptrans] (a2.south) .. controls (control) .. (s2.south);

\path[->] (q2) edge [bend right=27] node [swap] {$o_2,1$} (q0);
\end{tikzpicture}
\caption{A nondeterministic probabilistic automaton $\mathcal{B}$.}
\label{fig:cexmemsystem}
\end{figure}

First consider a memoryless scheduler $\sigma_p$. It can only choose
once what weight will be affected to each transition.  This choice is
parametrized by probability $p$ that represents the weight of
probability of the $q_1$ transition.  The scheduled NPA 
$\mathcal{B}_{/\sigma_p}$ is depicted on
\figurename~\ref{fig:cexmemnomem}.  The probabilities can be computed
using the technique laid out in Section~\ref{sec:computing}. We obtain
the following probabilities (see Appendix~\ref{app:schedulercalcul} for details):
\begin{mathpar}
\prob(\varepsilon) = \frac18
\and
\prob(o_1) = \frac78 \cdot p
\and
\prob(o_2) = \frac78 \cdot (1-p)
\and
\prob(\varphi,\varepsilon) = 0
\and
\prob(\varphi,o_1) = \frac{p}{25p^2 - 25p + 58} \cdot \frac{5p + 49}{8}
\and
\prob(\varphi,o_2) = \frac{1-p}{25p^2 - 25p + 58} \cdot \frac{15p + 7}{4}
\end{mathpar}
Which yields
\[\frac1{\hpo(\mathcal{B}_{/\sigma_p},\varphi,\Obs)} = 
\frac18 + \frac{49}{8} \cdot f(p) \cdot \left(\frac{p}{7f(p) - 5p - 49} + \frac{1-p}{7f(p) -30p - 14}\right)
\]
with $f(p) = 25p^2 - 25p + 58$ (see Appendix~\ref{app:cexmemhponomem}).
It can be shown\footnote{With the help of tools such as WolframAlpha.} that regardless of $p$,
$\hpo(\mathcal{B}_{/\sigma_p},\varphi,\Obs)$ never falls below $0.88$.

\begin{figure}
\centering
\subfigure[Fully probabilistic finite automaton $\mathcal{B}_{/\sigma_p}$]{
\label{fig:cexmemnomem}
\begin{tikzpicture}[auto,node distance=3cm]
\tikzstyle{every state}+=[shape=ellipse,minimum size=5pt,inner sep=2pt]
\node[state,initial above] (q0) at (0,1.75) {$q_0$};
\node[state] (q1) at (-2.,0) {$q_1$};
\node[state] (q2) at (2.,0) {$q_2$};
\node[node distance=1.25cm,below of=q0] (ph) {$\surd,\frac18$};

\path[->] (q0) edge (ph);

\path[->] (q0) edge [bend left=10,pos=0.9] node {$a,\frac34p$} (q1);
\path[->] (q0) edge [bend right=10,pos=0.25] node [swap] {$b,\frac18p$} (q1);

\path[->] (q1) edge [bend left=75,distance=2cm] node {$o_1,1$} (q0);

\path[->] (q0) edge [bend right=10,pos=0.9] node [swap] {$a,\frac18(1-p)$} (q2);
\path[->] (q0) edge [bend left=10,pos=0.25] node {$b,\frac34(1-p)$} (q2);

\path[->] (q2) edge [bend right=75,distance=2cm] node [swap] {$o_2,1$} (q0);
\end{tikzpicture}
}
\hfill{~}
\subfigure[Fully probabilistic finite automaton $\mathcal{B}_{/\sigma_m}$]{
\label{fig:cexmemmemory}
\begin{tikzpicture}[auto,node distance=2cm]
\tikzstyle{every state}+=[shape=ellipse,minimum size=5pt,inner sep=2pt]
\node[state,initial above] (q01) at (0,0) {$q_0,e$};
\node[state,below of=q01] (q1) {$q_1$};
\node[state,node distance=2cm,right of=q1] (q02) {$q_0,o$};
\node[state,above of=q02] (q2) {$q_2$};
\node[node distance=1.5cm,left of=q01] (ph1) {$\surd,\frac18$};
\node[node distance=1.5cm,right of=q02] (ph2) {$\surd,\frac18$};

\path[->] (q01) edge (ph1);
\path[->] (q01) edge [bend left=20] node {$a,\frac34$} (q1);
\path[->] (q01) edge [bend right=20] node [swap] {$b,\frac18$} (q1);

\path[->] (q1) edge node {$o_1,1$} (q02);

\path[->] (q02) edge (ph2);
\path[->] (q02) edge [bend right=20] node [swap] {$a,\frac18$} (q2);
\path[->] (q02) edge [bend left=20] node {$b,\frac34$} (q2);

\path[->] (q2) edge node [swap] {$o_2,1$} (q01);
\end{tikzpicture}
}
\caption{Scheduled automata.}
\end{figure}
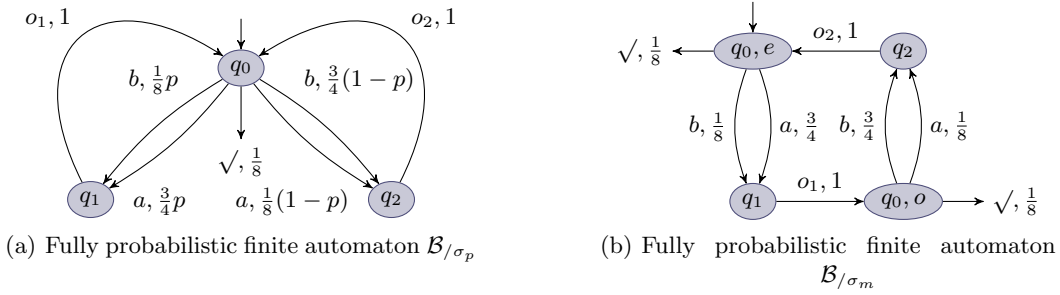

Now consider a scheduler $\sigma_m$ with memory who will try to
maximize the realization of $\varphi$.  In order to achieve that, it
introduces a bias towards taking the letter which will fulfill
$\varphi$: first an $a$, then a $b$, etc.  Hence on the even
positions, it will choose only transition to $q_1$ (with probability
$1$) while it will choose the transition to $q_2$ on odd positions.
The resulting FPFA is depicted on \figurename~\ref{fig:cexmemmemory}.
In this case, the probabilities of interest are:
\begin{mathpar}
\prob(\varepsilon) = \frac18
\and
\prob(o_1) = \frac7{15}
\and
\prob(o_2) = \frac78 \cdot \frac7{15}
\\
\prob(\varphi,\varepsilon) = 0
\and
\prob(\varphi,o_1) = \frac3{14}
\and
\prob(\varphi,o_2) = \frac34 \cdot \frac3{14}
\end{mathpar}
Probability $\prob(o_1)$ can be obtained by noticing that the execution has to stop after an odd number of letters from $\{a,b\}$ have been read.
The probability of stopping after exactly $n$ letters from $a$ or $b$ is $\frac18 \cdot \left(\frac78\right)^n$.
Therefore
\[
\prob(o_1) \ =\  \frac18 \cdot \sum_{i\geq 0} \left(\frac78\right)^{2i+1} \ =\  \frac18 \cdot \frac78 \cdot \frac1{1- \frac{49}{64}} \ =\  \frac18 \cdot \frac78 \cdot \frac{64}{15} \ =\  \frac7{15}.
\]
Similar reasoning yield the other probabilities.
The computation of \hponame from these values (see Appendix~\ref{app:cexmemhpomem}) gives $\hpo(\mathcal{B}_{/\sigma_m},\varphi,\Obs) = \frac{88192}{146509} \simeq 0.60$.

Hence a lower security is achieved by a scheduler provided it has (a
finite amount of) memory. 
\end{proof} 
Note that this
example used \hponame, but a similar argument could be adapted for the
other measures.

\subsection{Restricted schedulers}

What made a scheduler with memory more powerful than the one without
in the counterexample of Section~\ref{subsec:memorylessnotenough} was
the knowledge of the truth value of $\varphi$ and exactly what was
observed.  More precisely, if the predicate and the observables are
regular languages represented by finite deterministic and complete
automata (FDCA), schedulers can be restricted to choices according to
the current state of these automata and the state of the system.  We
conjecture that this knowledge is sufficient to any scheduler to
compromise security at the best of its capabilities.

Let $\varphi \subseteq CRun(\A)$ be a regular predicate represented
by an FDCA $\A_\varphi$.  Let $\Obs: CRun(\A)
\rightarrow\{o_1,\dots,o_n\}$ be an observation function such that for
$1 \leq i \leq n$, $\Obs^{-1}(o_i)$ is a regular set represented by an
FDCA $\A_{o_i}$.  Consider the synchronized product $\A_{\varphi,\Obs}
=\A_\varphi || \A_{o_1} || \dots || \A_{o_n}$, which is also an FDCA,
and denote by $Q_{\varphi,\Obs}$ its set of states.  Let
$\A_{\varphi,\Obs}(\rho)$ be the state of $\A_{\varphi,\Obs}$ reached
after reading $\rho$.
\begin{definition}[Restricted $(\varphi,\Obs)$-scheduler]
  A scheduler $\sigma$ for $\A$ is said $(\varphi,\Obs)$-restricted
  if there exists a function $\sigma': (Q_{\varphi,\Obs} \times Q)
  \rightarrow \probset(\probset((\Sigma \times Q) \uplus \{\surd\}))$
  such that for any run $\rho \in Run(\A)$, $\sigma(\rho) =
  \sigma'(\A_{\varphi,\Obs}(\rho),\lst(\rho))$.
\end{definition}

Remark that memoryless schedulers are always $(\varphi,\Obs)$-restricted.

\begin{proposition}\label{prop:restr}
  If $\sigma$ is $(\varphi,\Obs)$-restricted, then $\A_{/\sigma}$ is
  isomorphic to a finite FPFA.
\end{proposition}
 These schedulers keep all information about the predicate and the
observation.  We conjecture that the relevant supremum is reached by a
$(\varphi,\Obs)$-restricted scheduler.

\begin{proof}[Sketch of proof of Proposition~\ref{prop:restr}]
It can be shown that if $\sigma$ is $(\varphi,\Obs)$-restricted, then:
\begin{enumerate}[topsep=-\baselineskip,label={(\arabic*)\ }]
\item $\sigma$ is a memoryless scheduler for the product $\A || \A_{\varphi,\Obs}$
\item and $\left(\A || \A_{\varphi,\Obs}\right)_{/\sigma} = \A_{/\sigma}$.
\end{enumerate}
\end{proof}

\section{Conclusion}\label{sec:conclusion}

In this paper we introduced two dual notions of probabilistic opacity.
The liberal one measures the probability for an attacker observing a
random execution of the system to be able to gain information he can
be sure about.  The restrictive one measures the level of certitude in
the information acquired by an attacker observing the system.  The
extremal cases of both these notions coincide with the possibilistic
notion of opacity, which evaluates the existence of a leak of sure
information.  These notions yield measures that generalize either the
case of asymmetrical or symmetrical opacity, thus providing four
measures. 

However, probabilistic opacity is not always easy to compute,
especially if there are an infinite number of observables.
Nevertheless, automatic computation is possible when dealing with
regular predicates and finitely many regular observation classes.
A prototype tool was
implemented in Java, and can be used for numerical computation of
opacity values.

In future work we plan to explore more of the properties of
probabilistic opacity, to instantiate it to known security measures
(anonymity, non-interference, etc.). Also, we want to extend the study
of the non-deterministic case, by investigating the expressiveness of
schedulers.

\noindent\textbf{Acknowledgments.} We wish to thank Catuscia Palamidessi 
for interesting discussions on this work. Thanks also to the reviewers
for their helpful comments and to Adrian Eftenie for implementing the
\textsc{tpot} tool.

\smallskip
The authors are supported by projects ImpRo (ANR-2010-BLAN-0317) (French Government), CoChaT (DIGITEO-2009-27HD) (R\'egion \^Ile de France), inVEST (ERC Starting Grant 279499), the NSERC Discovery Individual grant No. 13321 (Government of Canada), and by the FQRNT Team grant No. 167440 (Quebec's Government).

\bibliographystyle{splncs}
\bibliography{bibliography}

\appendix
\section{Computation of \hponame for the debit card system}
\label{app:creditcardcalcul}

We give here the details of the computation of \hponame in the example of the debit cards system of Section~\ref{ex:debitcard}.

\begin{eqnarray*}
\prob(\Obs_{\textrm{Call}}=\varepsilon) &=&
\prob(\Obs_{\textrm{Call}}=\varepsilon,x>1000) \:+\:
\prob(\Obs_{\textrm{Call}}=\varepsilon,500 < x \leq 1000) \\&&\:+\:
\prob(\Obs_{\textrm{Call}}=\varepsilon,100 < x \leq 500) \:+\:
\prob(\Obs_{\textrm{Call}}=\varepsilon,x \leq 100)\\
&=&
\prob(\Obs_{\textrm{Call}}=\varepsilon) \cdot \prob(x>1000) \\&&\:+\:
\prob(\Obs_{\textrm{Call}}=\varepsilon) \cdot \prob(500 < x \leq 1000) \\&&\:+\:
\prob(\Obs_{\textrm{Call}}=\varepsilon) \cdot \prob(100 < x \leq 500) \\&&\:+\:
\prob(\Obs_{\textrm{Call}}=\varepsilon) \cdot \prob(x \leq 100) \\
&=& 0.05 \cdot 0.05 + 0.25 \cdot 0.2 + 0.5 \cdot 0.45 + 0.8 \cdot 0.3 \\
\prob(\Obs_{\textrm{Call}}=\varepsilon) &=& 0.5175
\end{eqnarray*}
\[\prob(\Obs_{\textrm{Call}}=\textrm{Call}) \ =\ 1 - \prob(\Obs_{\textrm{Call}}=\varepsilon) \ =\  0.4825\]
\begin{eqnarray*}
\prob(\neg\varphi_{>500}|\Obs_{\textrm{Call}}=\varepsilon) &=&
\frac{\prob(\neg\varphi_{>500},\Obs_{\textrm{Call}}=\varepsilon)}{\prob(\Obs_{\textrm{Call}}=\varepsilon)} \\
&=&\frac{\prob(x \leq 100,\Obs_{\textrm{Call}}=\varepsilon) \:+\:\prob(100 < x \leq 500,\Obs_{\textrm{Call}}=\varepsilon)}{\prob(\Obs_{\textrm{Call}}=\varepsilon)} \\
&=&\frac{\prob(\Obs_{\textrm{Call}}=\varepsilon|x \leq 100) \cdot \prob(x \leq 100)}{\prob(\Obs_{\textrm{Call}}=\varepsilon)} \\
&&+\:\frac{\prob(\Obs_{\textrm{Call}}=\varepsilon|100 < x \leq 500) \cdot \prob(100 < x \leq 500)}{\prob(\Obs_{\textrm{Call}}=\varepsilon)} \\
&=& \frac{0.8 \cdot 0.3 + 0.5 \cdot 0.45}{0.5175} \\
\prob(\neg\varphi_{>500}|\Obs_{\textrm{Call}}=\varepsilon) &=& \frac{0.465}{0.5175} \ \simeq\ 0.899
\end{eqnarray*}
\begin{eqnarray*}
\prob(\neg\varphi_{>500}|\Obs_{\textrm{Call}}=\textrm{Call}) &=&
\frac{\prob(\neg\varphi_{>500},\Obs_{\textrm{Call}}=\textrm{Call})}{\prob(\Obs_{\textrm{Call}}=\textrm{Call})} \\
&=&\frac{\prob(x \leq 100,\Obs_{\textrm{Call}}=\textrm{Call}) \:+\:\prob(100 < x \leq 500,\Obs_{\textrm{Call}}=\textrm{Call})}{\prob(\Obs_{\textrm{Call}}=\textrm{Call})} \\
&=&\frac{\prob(\Obs_{\textrm{Call}}=\textrm{Call}|x \leq 100) \cdot \prob(x \leq 100)}{\prob(\Obs_{\textrm{Call}}=\textrm{Call})} \\
&&+\:\frac{\prob(\Obs_{\textrm{Call}}=\textrm{Call}|100 < x \leq 500) \cdot \prob(100 < x \leq 500)}{\prob(\Obs_{\textrm{Call}}=\textrm{Call})} \\
&=& \frac{0.2 \cdot 0.3 + 0.5 \cdot 0.45}{0.4825} \\
\prob(\neg\varphi_{>500}|\Obs_{\textrm{Call}}=\textrm{Call}) &=& \frac{0.285}{0.4825} \ \simeq\ 0.591
\end{eqnarray*}
\begin{eqnarray*}
\frac{1}{\hpo(\A_{\textrm{card}},\varphi_{>500},\Obs_{\textrm{Call}})} &=& 0.5175 \cdot \frac{0.5175}{0.465} + 0.4825 \cdot \frac{0.4825}{0.285}\\
\frac{1}{\hpo(\A_{\textrm{card}},\varphi_{>500},\Obs_{\textrm{Call}})} &=& \frac{39377}{28272} \ \simeq \ 1.393
\end{eqnarray*}
The last line was obtained by reducing the one above with the help of
the formal computation tool \emph{WolframAlpha}.

\section{Resolution of the linear system for Crowds protocol}
\label{app:systemcrowds}

It can be seen in the system of \tablename~\ref{tab:crowdsSystem} (page~\pageref{tab:crowdsSystem}) that $L_1=L_2 = \cdots = L_{n-c-1} =
q \cdot L_{1'}$ and $L_{n-c+1} = \cdots = L_n = L_S=1$. Therefore, it
suffices to eliminate $L_{1'}$ and compute $L_0$, $L_1$ and $L_{n-c}$.
\[\left\{\begin{array}{rcl}
L_0 &=& \frac1{q(n-c)} \cdot L_1 \\
L_1 &=& q \left(\frac{n-c-1}n \cdot L_1 + \frac1n \cdot L_{n-c}\right) \\
L_{n-c} &=& 1 - \frac{q(n-c)}{n} + L_1
\end{array}\right.\]
The line for $L_{n-c}$ is obtained as follows:
\begin{eqnarray*}
L_{n-c} &=& (1-q) \cdot L_{S} + \sum_{i = 1}^{n} \frac{q}{n} \cdot L_i \\
L_{n-c} &=& (1-q) \cdot L_{S} + \sum_{i = 1}^{n-c} \frac{q}{n} \cdot L_i + \sum_{i = n-c+1}^{n} \frac{q}{n} \cdot L_i \\
L_{n-c} &=&  1-q + L_1 +\frac{q \cdot c}{n}\\
L_{n-c} &=&  1 - \frac{q(n-c)}{n} + L_1
\end{eqnarray*}
This yields, for $L_1$:
\begin{eqnarray*}
L_1 &=& \frac{q}{n} (n-c) \cdot L_1 + \frac{q}{n} \left(1 - \frac{q(n-c)}{n}\right) \\
\frac{n}{q} \cdot L_1 &=& (n-c) \cdot L_1 + 1 - \frac{q(n-c)}{n} \\
L_1 \left(\frac{n}{q} - (n-c)\right) &=& \frac{n-q(n-c)}{n} \\
L_1 &=& \frac{q}{n}
\end{eqnarray*}
The other values are easily deduced from $L_1$.

\section{Calculations in the proof of Theorem~\ref{thm:ctre-ex}}

\subsection{Probabilities in $\mathcal{B}_{/\sigma_p}$.}
\label{app:schedulercalcul}
We compute the probabilities of several events in the automaton $\mathcal{B}_{/\sigma_p}$, reproduced on \figurename~\ref{fig:cexmemnomembis}.
Recall that $\varphi = (ab)^+ + (ab)^*a$ and $\Obs$ is the last letter read, so the set of observables is $Obs = \{\varepsilon, o_1,o_2\}$.

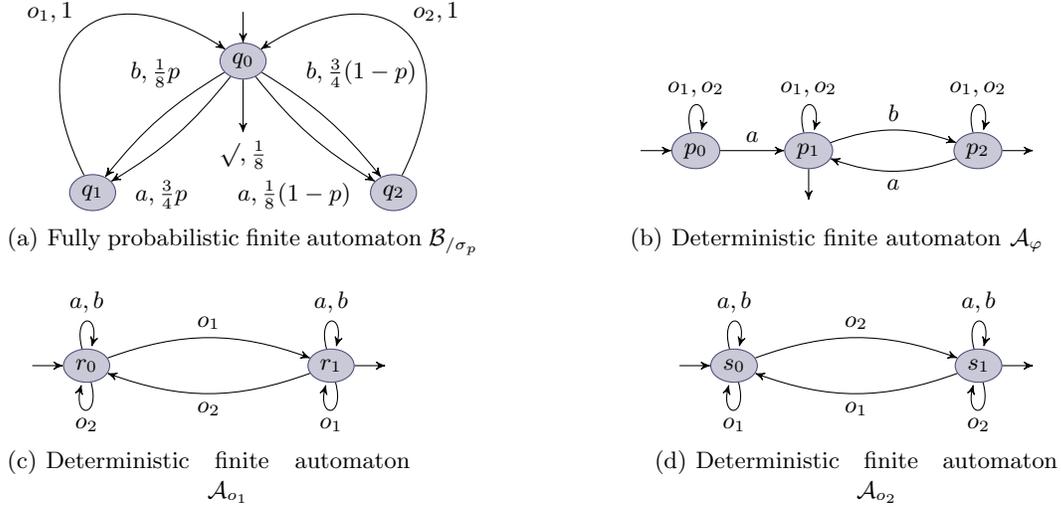
\begin{figure}
\centering
\subfigure[Fully probabilistic finite automaton $\mathcal{B}_{/\sigma_p}$]{
\label{fig:cexmemnomembis}
\begin{tikzpicture}[auto,node distance=3cm]
\tikzstyle{every state}+=[shape=ellipse,minimum size=5pt,inner sep=2pt]
\node[state,initial above] (q0) at (0,1.75) {$q_0$};
\node[state] (q1) at (-2.,0) {$q_1$};
\node[state] (q2) at (2.,0) {$q_2$};
\node[node distance=1.25cm,below of=q0] (ph) {$\surd,\frac18$};

\path[->] (q0) edge (ph);

\path[->] (q0) edge [bend left=10,pos=0.9] node {$a,\frac34p$} (q1);
\path[->] (q0) edge [bend right=10,pos=0.25] node [swap] {$b,\frac18p$} (q1);

\path[->] (q1) edge [bend left=75,distance=2cm] node {$o_1,1$} (q0);

\path[->] (q0) edge [bend right=10,pos=0.9] node [swap] {$a,\frac18(1-p)$} (q2);
\path[->] (q0) edge [bend left=10,pos=0.25] node {$b,\frac34(1-p)$} (q2);

\path[->] (q2) edge [bend right=75,distance=2cm] node [swap] {$o_2,1$} (q0);
\end{tikzpicture}
}
\hfill{~}
\subfigure[Deterministic finite automaton $\A_\varphi$]{
\label{fig:cexmemphi}
\begin{tikzpicture}[auto,node distance=1.5cm]
\tikzstyle{every state}+=[shape=ellipse,minimum size=5pt,inner sep=2pt]
\node[state,initial left] (q0) at (0,0) {$p_0$};
\node[state,right of=q0,accepting below] (q1) {$p_1$};
\node[state,node distance=2.25cm,right of=q1,accepting right] (q2) {$p_2$};

\path[->] (q0) edge node {$a$} (q1);
\path[->] (q1) edge [bend left=20] node {$b$} (q2);
\path[->] (q2) edge [bend left=20] node {$a$} (q1);

\path[->] (q0) edge [loop above] node {$o_1,o_2$} (q0);
\path[->] (q1) edge [loop above] node {$o_1,o_2$} (q1);
\path[->] (q2) edge [loop above] node {$o_1,o_2$} (q2);
\end{tikzpicture}
}

\subfigure[Deterministic finite automaton $\A_{o_1}$]{
\label{fig:cexmemo1}
\begin{tikzpicture}[auto,node distance=3.25cm]
\tikzstyle{every state}+=[shape=ellipse,minimum size=5pt,inner sep=2pt]
\node[state,initial left] (q0) at (0,0) {$r_0$};
\node[state,right of=q0,accepting right] (q1) {$r_1$};

\path[->] (q0) edge [bend left=20] node {$o_1$} (q1);
\path[->] (q1) edge [bend left=20] node {$o_2$} (q0);

\path[->] (q0) edge [loop above] node {$a,b$} (q0);
\path[->] (q1) edge [loop above] node {$a,b$} (q1);
\path[->] (q0) edge [loop below] node {$o_2$} (q0);
\path[->] (q1) edge [loop below] node {$o_1$} (q1);
\end{tikzpicture}
}
\hfill{~}
\subfigure[Deterministic finite automaton $\A_{o_2}$]{
\label{fig:cexmemo2}
\begin{tikzpicture}[auto,node distance=3.25cm]
\tikzstyle{every state}+=[shape=ellipse,minimum size=5pt,inner sep=2pt]
\node[state,initial left] (q0) at (0,0) {$s_0$};
\node[state,right of=q0,accepting right] (q1) {$s_1$};

\path[->] (q0) edge [bend left=20] node {$o_2$} (q1);
\path[->] (q1) edge [bend left=20] node {$o_1$} (q0);

\path[->] (q0) edge [loop above] node {$a,b$} (q0);
\path[->] (q1) edge [loop above] node {$a,b$} (q1);
\path[->] (q0) edge [loop below] node {$o_1$} (q0);
\path[->] (q1) edge [loop below] node {$o_2$} (q1);
\end{tikzpicture}
}
\caption{Automata for the computation of $\prob(\varphi,o_1)$ and $\prob(\varphi,o_2)$.}
\end{figure}

The computation of the probability $\prob(\varphi,o_1)$ in FPFA
$\mathcal{B}_{/\sigma_p}$ goes as follows.  We write $\overline{p} =
1-p$ for brevity.  First we build the synchronized product
$\mathcal{B}_{/\sigma_p} || \A_\varphi || \A_{o_1}$, as depicted in
\figurename~\ref{fig:cexmemproduct1}.
\begin{figure}
\centering
\begin{tikzpicture}[auto,node distance=2cm]
\tikzstyle{every state}+=[shape=rectangle,rounded corners=4pt,minimum size=5pt,inner sep=3pt]
\tikzstyle{placeholder}=[node distance=1.5cm]
\node[state, initial above] at (0,0) (q000) {$q_0,p_0,r_0$};

\node[state,below left of=q000] (q210) {$q_2,p_1,r_0$};
\node[state,below left of=q210] (q010) {$q_0,p_1,r_0$};
\node[state, below left of=q010] (q120) {$q_1,p_2,r_0$};
\node[state, below right of=q010] (q220) {$q_2,p_2,r_0$};

\node[state,below right of=q000] (q110) {$q_1,p_1,r_0$};
\node[state,below right of=q110] (q011) {$q_0,p_1,r_1$};
\node[placeholder,right of=q011] (ph011) {};
\path[->] (q011) edge node {$\surd,\frac18$} (ph011);
\node[state, below left of=q011] (q221) {$q_2,p_2,r_1$};
\node[state, below right of=q011] (q121) {$q_1,p_2,r_1$};

\node[state,below right of=q220] (q020) {$q_0,p_2,r_0$};
\node[state,node distance=2.75cm,below of=q020] (q021) {$q_0,p_2,r_1$};
\node[node distance=0.75cm,above of=q021,anchor=south] (ph021) {$\surd,\frac18$};
\path[->] (q021) edge (ph021);
\node[state,below left of=q021] (q211) {$q_2,p_1,r_1$};
\node[state,below right of=q021] (q111) {$q_1,p_1,r_1$};

\path[->] (q000) edge node [swap] {$a,\frac{\overline{p}}8$} (q210);
\path[->] (q000) edge node {$a,\frac{3p}4$} (q110);
\path[->] (q210) edge node [swap] {$o_2,1$} (q010);
\path[->] (q110) edge node {$o_1,1$} (q011);

\path[->] (q010) edge node [swap] {$b,\frac{p}8$} (q120);
\path[->] (q010) edge node {$b,\frac{3\overline{p}}4$} (q220);

\path[->] (q011) edge node [swap] {$b,\frac{3\overline{p}}4$} (q221);
\path[->] (q011) edge node{$b,\frac{p}8$} (q121);

\path[->] (q220) edge node [swap] {$o_2,1$} (q020);
\path[->] (q221) edge node {$o_2,1$} (q020);

\path[->] (q120) edge node [swap] {$o_1,1$} (q021);
\path[->] (q121) edge node {$o_1,1$} (q021);

\path[->] (q020) edge node [swap,pos=0.75] {$a,\frac{\overline{p}}8$} (q210);
\path[->] (q020) edge node [pos=0.75] {$a,\frac{3p}4$} (q110);

\path[->] (q021) edge node [swap] {$a,\frac{\overline{p}}8$} (q211);
\path[->] (q021) edge node {$a,\frac{3p}4$} (q111);

\path[->] (q111) edge [bend right=40,swap] node {$o_1,1$} (q011);
\path[->] (q211) edge [bend left=40] node {$o_2,1$} (q010);
\end{tikzpicture}
\caption{Substochastic Automaton $\mathcal{B}_{/\sigma_p} || \A_\varphi || \A_{o_1}$.}
\label{fig:cexmemproduct1}
\end{figure}
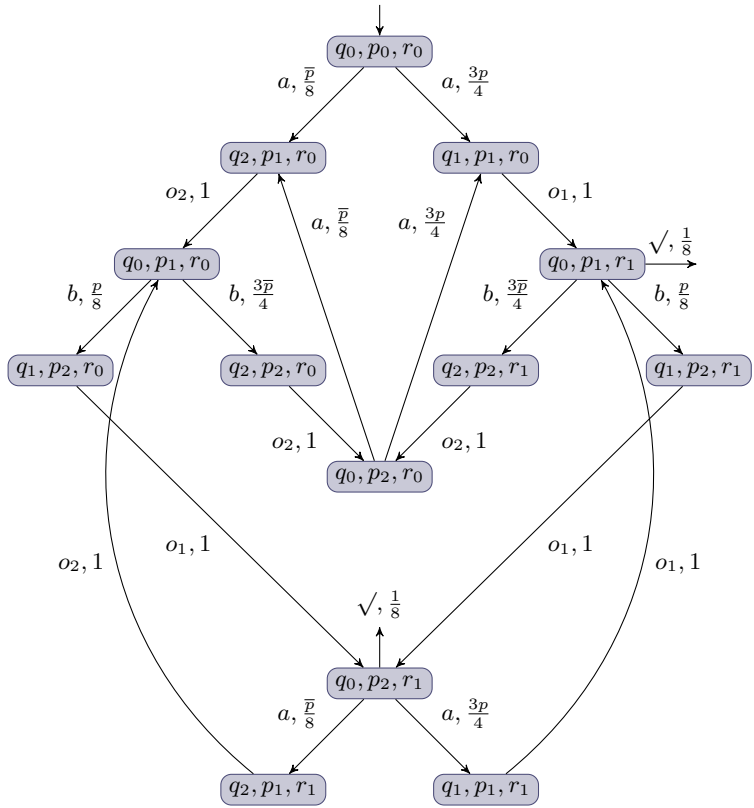
The linear system of \tablename~\ref{tab:cexmemsystem} is built from this automaton.
\begin{table}
\[\left\{\begin{array}{rcl}
x_{000} &=& \frac18 \overline{p} \ x_{210} + \frac34 p \ x_{110} \\
x_{210} &=& x_{010} \\
x_{110} &=& x_{011} \\
x_{010} &=& \frac18 p \ x_{120} + \frac34 \overline{p} \ x_{220} \\
x_{011} &=& \frac18 + \frac34\overline{p} \ x_{221} + \frac18 p \ x_{121} \\
x_{120} &=& x_{021} \\
x_{220} &=& x_{020} \\
x_{221} &=& x_{020} \\
x_{121} &=& x_{021} \\
x_{020} &=& \frac18 \overline{p} \ x_{210} + \frac34 p \ x_{110} \\
x_{021} &=& \frac18 + \frac18 \overline{p} \ x_{211} + \frac34 p \ x_{111} \\
x_{211} &=& x_{010} \\
x_{111} &=& x_{011}
\end{array}\right.\]
\caption[Linear system associated to the SA $\mathcal{B}_{/\sigma_p} || \A_\varphi || \A_{o_1}$.]{Linear system associated to the SA $\mathcal{B}_{/\sigma_p} || \A_\varphi || \A_{o_1}$. The variables names indicate the corresponding state in the automaton; for example $x_{210}$ corresponds to state $(q_2,p_1,r_0)$.}
\label{tab:cexmemsystem}
\end{table}
This system can be trimmed down in order to remove redundancy, and since only the value of $x_{000} = \prob(\varphi,o_1)$ is of interest:
\[\left\{\begin{array}{rcl}
x_{000} &=& \frac18 \overline{p} \ x_{010} + \frac34 p \ x_{011} \\
x_{010} &=& \frac18 p \ x_{021} + \frac34 \overline{p} \ x_{020} \\
x_{011} &=& \frac18 + \frac34\overline{p} \ x_{020} + \frac18 p \ x_{021} \\
x_{020} &=& \frac18 \overline{p} \ x_{010} + \frac34 p \ x_{011} \\
x_{021} &=& \frac18 + \frac18 \overline{p} \ x_{010} + \frac34 p \ x_{011}
\end{array}\right.
\Longleftrightarrow
\left\{\begin{array}{rcl}
x_{000} &=& \frac18 \overline{p} \ x_{010} + \frac34 p \ x_{011} \\
x_{010} &=& \frac18 p \ x_{021} + \frac34 \overline{p} \ x_{020} \\
x_{011} &=& \frac18 + x_{010} \\
x_{020} &=& x_{000} \\
x_{021} &=& \frac18 + x_{000}
\end{array}\right.\]
We therefore obtain:
\[\left\{\begin{array}{rcl}
x_{000} &=& \frac18 \overline{p} \ x_{010} + \frac34 p \left(\frac18 + x_{010}\right) \\
x_{010} &=& \frac18 p \left(\frac18 + x_{000}\right) + \frac34 \overline{p} \ x_{000}
\end{array}\right.
\textrm{So}\ 
\left\{\begin{array}{rcl}
x_{000} &=& \frac18 \overline{p} \ x_{010} + \frac34 p \left(\frac18 + x_{010}\right) \\
x_{010} &=& \frac1{64} p + \frac18 p \  x_{000} + \frac34 \overline{p} \ x_{000}
\end{array}\right.\]
As a result
\[x_{000} \ =\  \frac18 \overline{p} \ \left(\frac1{64} p + \frac18 p \  x_{000} + \frac34 \overline{p} \ x_{000}\right) + \frac34 p \left(\frac18 + \frac1{64} p + \frac18 p \  x_{000} + \frac34 \overline{p} \ x_{000}\right)\]
In the sequel, we replace $x_{000}$ with $x$ for readability's sake.
\begin{eqnarray*}
 && x \ =\  \frac18 \overline{p} \ \left(\frac1{64} p + \frac18 p \  x + \frac34 \overline{p} \ x\right) + \frac34 p \left(\frac18 + \frac1{64} p + \frac18 p \  x + \frac34 \overline{p} \ x\right) \\
&\Longleftrightarrow&
x \ =\  \frac1{512} p \overline{p} + \frac1{64} p \overline{p} x + \frac3{32} \overline{p}^2 x + \frac3{32} p + \frac3{256} p^2 + \frac3{32} p^2 x + \frac9{16} p \overline{p} x \\
&\Longleftrightarrow& x \left(1 - \frac1{64} p \overline{p} -\frac3{32} \overline{p}^2 - \frac3{32} p^2 - \frac9{16} p \overline{p} \right) \ = \ \frac1{512} p \overline{p} + \frac3{32} p + \frac3{256} p^2 \\
&\Longleftrightarrow& x \ =\  \frac{\frac18 p \overline{p} + 6p + \frac34p^2}{64- p \overline{p} - 6\overline{p}^2 - 6p^2 - 36 p \overline{p}} \\
&\Longleftrightarrow& x \ =\  \frac{\frac18 p (1-p) + 6p + \frac34p^2}{64 - 37 p (1-p) - 6 (p^2 - 2p +1) - 6p^2} \\
&\Longleftrightarrow& x \ =\  \frac{\frac18 p - \frac18 p^2 + 6p + \frac34 p^2}{64 - 37p + 37p^2 - 6p^2 + 12p - 6 -6p^2} \\
&\Longleftrightarrow& x \ =\ \frac18 \cdot p\cdot  \frac{5p + 49}{25p^2 - 25p + 58}
\end{eqnarray*}

\bigskip
The same technique can be applied to the computation of $\prob(\varphi,o_2)$ in $\mathcal{B}_{/\sigma_p}$.
The product is depicted on \figurename~\ref{fig:cexmemproduct2}, and the linear system obtained boils down to
\[\left\{\begin{array}{rcl}
x_{000} &=& \frac18 \overline{p} \left(\frac18 + x_{010}\right) + \frac34 p \ x_{010} \\
x_{010} &=& \frac18 p \ x_{000} + \frac34 \overline{p} \left(\frac18 + x_{000}\right)
\end{array}\right.\]
\begin{figure}
\centering
\begin{tikzpicture}[auto,node distance=2cm]
\tikzstyle{every state}+=[shape=rectangle,rounded corners=4pt,minimum size=5pt,inner sep=3pt]
\tikzstyle{placeholder}=[node distance=1.5cm]
\node[state, initial above] at (0,0) (q000) {$q_0,p_0,s_0$};

\node[state,below left of=q000] (q210) {$q_2,p_1,s_0$};
\node[state,below left of=q210] (q011) {$q_0,p_1,s_1$};
\node[placeholder,left of=q011] (ph011) {};
\path[->] (q011) edge node [swap] {$\surd,\frac18$} (ph011);
\node[state, below right of=q011] (q121) {$q_1,p_2,s_1$};
\node[state, below left of=q011] (q221) {$q_2,p_2,s_1$};

\node[state,below right of=q000] (q110) {$q_1,p_1,s_0$};
\node[state,below right of=q110] (q010) {$q_0,p_1,s_0$};
\node[state, below right of=q010] (q220) {$q_2,p_2,s_0$};
\node[state, below left of=q010] (q120) {$q_1,p_2,s_0$};

\node[state,below right of=q121] (q020) {$q_0,p_2,s_0$};
\node[state,node distance=2.75cm,below of=q020] (q021) {$q_0,p_2,s_1$};
\node[node distance=0.75cm,above of=q021,anchor=south] (ph021) {$\surd,\frac18$};
\path[->] (q021) edge (ph021);
\node[state,below left of=q021] (q211) {$q_2,p_1,s_1$};
\node[state,below right of=q021] (q111) {$q_1,p_1,s_1$};

\path[->] (q000) edge node [swap] {$a,\frac{\overline{p}}8$} (q210);
\path[->] (q000) edge node {$a,\frac{3p}4$} (q110);
\path[->] (q210) edge node [swap] {$o_2,1$} (q011);
\path[->] (q110) edge node {$o_1,1$} (q010);

\path[->] (q011) edge node {$b,\frac{p}8$} (q121);
\path[->] (q011) edge node [swap] {$b,\frac{3\overline{p}}4$} (q221);

\path[->] (q010) edge node {$b,\frac{3\overline{p}}4$} (q220);
\path[->] (q010) edge node [swap] {$b,\frac{p}8$} (q120);

\path[->] (q220) edge node [swap] {$o_2,1$} (q021);
\path[->] (q221) edge node {$o_2,1$} (q021);

\path[->] (q120) edge node {$o_1,1$} (q020);
\path[->] (q121) edge node [swap] {$o_1,1$} (q020);

\path[->] (q020) edge node [swap,pos=0.75] {$a,\frac{\overline{p}}8$} (q210);
\path[->] (q020) edge node [pos=0.75] {$a,\frac{3p}4$} (q110);

\path[->] (q021) edge node [swap] {$a,\frac{\overline{p}}8$} (q211);
\path[->] (q021) edge node {$a,\frac{3p}4$} (q111);

\path[->] (q111) edge [bend right=40,swap] node {$o_1,1$} (q010);
\path[->] (q211) edge [bend left=40] node {$o_2,1$} (q011);
\end{tikzpicture}
\caption{Substochastic Automaton $\mathcal{B}_{/\sigma_p} || \A_\varphi || \A_{o_2}$.}
\label{fig:cexmemproduct2}
\end{figure}
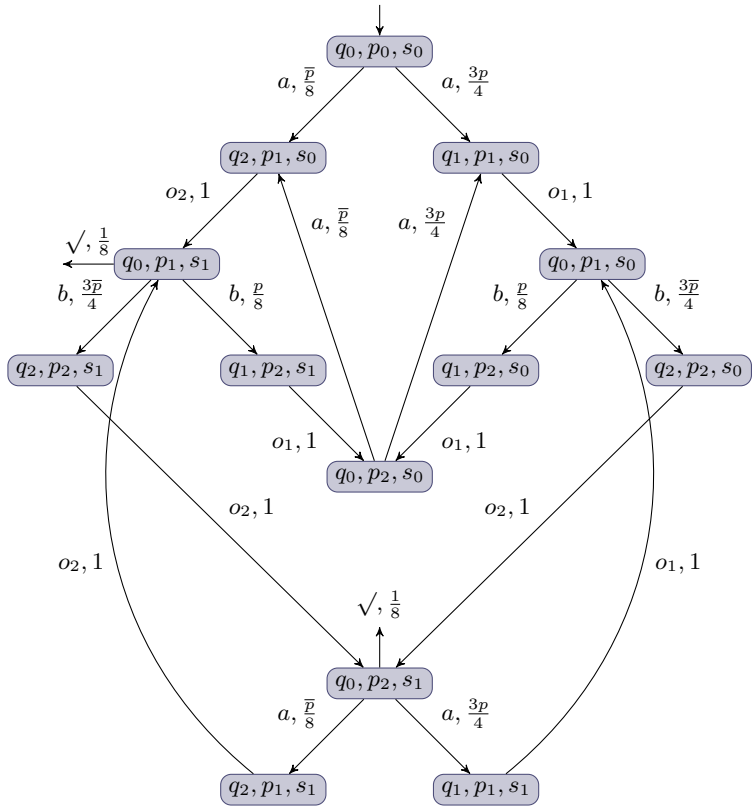
As before, $x_{000}$ is replaced by $x$ for readability; we solve:
\begin{eqnarray*}
x &=&  \frac18 \overline{p} \left(\frac18 + \frac18 p x + \frac34 \overline{p} \left(\frac18 + x\right)\right) + \frac34 p \ \left(\frac18 p x + \frac34 \overline{p} \left(\frac18 + x\right)\right) \\
x &=& \frac1{64} \overline{p} + \frac1{64} p \overline{p} x + \frac3{256} \overline{p}^2 + \frac3{32} \overline{p}^2 x + \frac3{32} p^2 x + \frac9{128} p \overline{p} + \frac9{16} p \overline{p} x \\
x &=& \frac{%
\frac1{64} \overline{p} + \frac3{256} \overline{p}^2 + \frac9{128} p \overline{p}
}{%
1 - \frac1{64} p \overline{p} - \frac3{32}  \overline{p}^2 - \frac3{32} p^2 - \frac9{16} p \overline{p}
} \\
x &=& \frac{%
4 \overline{p} + 3 \overline{p}^2 + 18 p \overline{p}
}{%
256 - 24 \overline{p}^2 - 24 p^2 - 148 p \overline{p}
} \\
x &=& \frac{%
(1-p) (4 + 3 - 3p + 18p)
}{%
256 - 24 p^2 + 48 p - 24 - 24 p^2 - 148 p + 148 p^2
} \\
x &=& \frac{%
(1-p) (15p+7)
}{%
232 - 100p + 100 p^2
} \\
x &=& \frac{%
(1-p) (15p+7)
}{%
4(25p^2 - 25p + 58)
}
\end{eqnarray*}

\subsection{Computation of \hponame for $\mathcal{B}_{/\sigma_p}$}
\label{app:cexmemhponomem}

In the sequel, we write $f(p) = 25p^2 - 25p + 58$.
We have $\prob(\overline\varphi | \varepsilon) = 1$ and
\begin{eqnarray*}
\prob(\overline\varphi | o_1) &=& 1 - \frac{\prob(\varphi,o_1)}{\prob(o_1)} \\
\prob(\overline\varphi | o_1) &=& 1 - \frac18 \cdot p \cdot \frac{5p + 49}{25p^2 - 25p + 58} \cdot \frac8{7p} \\
\prob(\overline\varphi | o_1) &=& \frac{7f(p) - 5p - 49}{7f(p)}
\end{eqnarray*}
\begin{eqnarray*}
\prob(\overline\varphi | o_2) &=& 1 - \frac{\prob(\varphi,o_2)}{\prob(o_2)} \\
\prob(\overline\varphi | o_2) &=& 1 - \frac{(1-p) (15p+7)}{4(25p^2 - 25p + 58)} \cdot \frac8{7(1-p)} \\
\prob(\overline\varphi | o_2) &=& \frac{7f(p) - 30p -14}{7f(p)}
\end{eqnarray*}
\begin{eqnarray*}
\frac1{\hpo\left(\mathcal{B}_{/\sigma_p},\varphi,\Obs\right)} &=& \prob(\varepsilon) + \prob(o_1) \cdot \frac1{\prob(\overline\varphi|o_1)} + \prob(o_2) \cdot \frac1{\prob(\overline\varphi|o_2)} \\
\frac1{\hpo\left(\mathcal{B}_{/\sigma_p},\varphi,\Obs\right)} &=& \frac18 + \frac78 \cdot p \cdot \frac{7f(p)}{7f(p) - 5p - 49} + \frac78 \cdot (1-p) \cdot \frac{7f(p)}{7f(p) - 30p -14} \\
\frac1{\hpo\left(\mathcal{B}_{/\sigma_p},\varphi,\Obs\right)} &=& \frac18 + \frac{49 f(p)}8 \left(\frac{p}{7f(p) - 5p - 49} + \frac{1-p}{7f(p) - 30p -14}\right)
\end{eqnarray*}

\subsection{Computation of \hponame for $\mathcal{B}_{/\sigma_m}$}
\label{app:cexmemhpomem}

We have:
\begin{mathpar}
\prob(\overline\varphi | \varepsilon) = 1
\and
\prob(\overline\varphi|o_1) = 1 - \frac3{14} \cdot \frac{15}7 = \frac{53}{98}
\and
\prob(\overline\varphi|o_2) = 1 - \frac34 \cdot \frac3{14} \cdot \frac{15}7 \cdot \frac87 = \frac{208}{343}
\end{mathpar}
Therefore:
\begin{eqnarray*}
\frac1{\hpo(\mathcal{B}_{/\sigma_m},\varphi,\Obs)} &=& \frac18 + \frac7{15} \cdot \frac{98}{53} + \frac78 \cdot \frac7{15} \cdot \frac{343}{208} \\
\frac1{\hpo(\mathcal{B}_{/\sigma_m},\varphi,\Obs)} &=& \frac{146509}{88192} \\
\hpo(\mathcal{B}_{/\sigma_m},\varphi,\Obs) &=& \frac{88192}{146509}\\
\hpo(\mathcal{B}_{/\sigma_m},\varphi,\Obs) &\simeq& 0.60
\end{eqnarray*}

\end{document}